\documentclass[times,sort&compress,3p]{elsarticle}
\journal{.}
\usepackage[labelfont=bf]{caption}
\usepackage{graphicx}               

\usepackage{dutchcal}
\usepackage{amsmath,amsfonts,amssymb,amsthm,booktabs,color, bm,epsfig,graphicx,hyperref,url}
\usepackage{dsfont}

\usepackage[normalem]{ulem} 
\usepackage{enumitem}

 \newcommand{\0}{\bm{0}}

 \newcommand{\bA}{{\bf{A}}}
 \newcommand{\bbeta}{{\bm{\beta}}}
 \newcommand{\bb}{{\bf{b}}}

 \newcommand{\bD}{{\bm{D}}}

 \newcommand{\bDel}{{\bm{\Delta}}}
 \newcommand{\beps}{{\bm{\varepsilon}}}
 \newcommand{\beeta}{\bm{\eta}}
 \newcommand{\bgam}{\bm{\gamma}}
 \newcommand{\bGam}{\bm{\Gamma}}
 \newcommand{\bI}{{\bf{I}}}

 \newcommand{\bOm}{\bm{\Omega}}
 \newcommand{\bom}{\bm{\omega}}

 \newcommand{\bSig}{\bm{\Sigma}}

 \newcommand{\btau}{\bm{\tau}}

 \newcommand{\bx}{{\bm{x}}}
 \newcommand{\bX}{{\bf{X}}}
 \newcommand{\bxi}{{\bm{\xi}}}
 
 \newcommand{\by}{{\bf{y}}}
 
 \newcommand{\bz}{{\bf{z}}}

 \newcommand{\diag}{\text{diag}}
 \newcommand{\E}{\mathbb{E}}
 \newcommand{\EC}{\mbox{{EC}}}

 \renewcommand{\Pr}{\mathbb{P}}
 \newcommand{\R}{\mathbb{R}}

 \newcommand{\SN}{{\mbox{SN}}}
 \newcommand{\SUE}{{\mbox{SUE}}}
 \newcommand{\SUN}{{\mbox{SUN}}}
 \newcommand{\SUT}{{\mbox{SUT}}}

\let\oldproofname=\proofname
\renewcommand{\proofname}{\rm\bf{\oldproofname}}

\theoremstyle{plain}

\newtheorem{exa}{Example}
\newtheorem{rmk}{Remark}
\newtheorem{prop}{Proposition}
\newtheorem{lemma}{Lemma}
\newtheorem{corollary}{Corollary}

\theoremstyle{definition}

\begin{document}
\begin{frontmatter}

\title{Conjugacy properties of multivariate unified skew-elliptical distributions}

\author[1]{Maicon J. Karling}
\author[2]{Daniele Durante}
\author[1]{Marc G. Genton}

\address[1]{Statistics Program, King Abdullah University of Science and Technology, Thuwal, Saudi Arabia}
\address[2]{Department of Decision Sciences and Institute for Data Science and Analytics, Bocconi University, Milan, Italy
}

\cortext[mycorrespondingauthor]{Corresponding author. Email address: maicon.karling@kaust.edu.sa}

\begin{abstract}
The broad class of multivariate unified skew-normal (SUN) distributions has been recently shown to possess important conjugacy properties. When used as priors for the coefficients vector in probit, tobit, and multinomial probit models, these distributions yield posteriors that still belong to the  SUN family. Although this result has led to advancements in Bayesian inference and computation, its applicability beyond likelihoods associated with fully-observed, discretized, or censored realizations from multivariate Gaussian models remains yet unexplored. This article covers such a gap by proving that the wider family of multivariate unified skew-elliptical  (SUE) distributions, which extends SUNs to more general perturbations of elliptical densities, guarantees conjugacy for broader classes of models, beyond those relying on fully-observed, discretized or censored Gaussians. Such a result leverages the closure under linear combinations, conditioning and marginalization of SUE to prove that this family is conjugate to the likelihood induced by multivariate regression models for fully-observed, censored or dichotomized realizations from skew-elliptical distributions. This advancement enlarges the set of models that enable conjugate Bayesian inference to general formulations arising from elliptical and skew-elliptical families, including the multivariate Student's $t$ and skew-$t$, among others.
\end{abstract}

\begin{keyword} 
Bayesian statistics \sep
multivariate unified skew-elliptical distribution \sep
unified skew-normal distribution.
\MSC[2020] Primary 62F15 \sep
Secondary 62H10
\end{keyword}

\end{frontmatter}


\section{Introduction\label{sec:1}}

Conjugacy is a central property within Bayesian statistics. A prior $\mathcal{p}(\bbeta)$ for the vector of parameters $\bbeta \in \mathcal{B} \subset \mathbb{R}^p$ is conjugate to the likelihood $\mathcal{p}(\by \mid \bbeta)$ of the observed data $\by  \in \mathcal{Y} \subset \mathbb{R}^n$, if the induced posterior $\mathcal{p}(\bbeta \mid \by)$ still belongs to the same class of distributions of the assumed prior. This important property implies that when $\mathcal{p}(\bbeta)$ belongs to a known and tractable family, Bayesian inference under the induced posterior can also leverage the tractability of such a family, thereby circumventing the  challenges that arise in Bayesian computation and inference under intractable posterior distributions. Despite the relevance of this property, identifying known and tractable conjugate priors for the likelihoods induced by commonly-used statistical models is often challenging. Remarkably, until recently, conjugacy in regression settings was mainly established for univariate or multivariate normal responses $\by$ with Gaussian priors for the coefficients $\bbeta$, thus hindering potentials of conjugate Bayesian inference beyond this specific setting.

To address the aforementioned gap,  \citet{durante19} has recently shown that also general probit models admit conjugate priors, with these priors belonging to the known family of unified skew-normal (SUN) distributions  \citep{arellano06a}. Such a class includes multivariate Gaussians as a special case, and extends these symmetric distributions through the perturbation of the corresponding density via a factor that coincides with the cumulative distribution function of a multivariate normal, thereby inducing skewness. Crucially, SUNs  (i) have a known normalizing constant and moment-generating function, (ii) admit a tractable stochastic representation, and (iii) preserve the closure under linear combinations, conditioning and marginalization of the original multivariate Gaussians \citep{arellano06a,azzalini14,arellano22}. These properties facilitate Bayesian inference under the induced SUN posterior and, consequently, have motivated rapid subsequent research to establish SUN conjugacy for broader classes of models beyond classical probit representations. Relevant advancements along these lines include dynamic multivariate probit \citep{fasano2021closed}, multinomial probit \citep{fasano2022class}, probit Gaussian processes \citep{cao2022scalable}, tobit models \citep{anceschi23} and, more generally, any representation inducing likelihoods proportional to the kernel of a SUN  \citep{anceschi23}. Such a latter result crucially includes also important skewed extensions of classical probit, multinomial probit, and tobit \citep[e.g,][]{chen1999new,sahu03,arellano2005skew,bolfarine2007skew,lachos07,bazan10,lachos2010likelihood,hutton11,benavoli2020skew,galarza2022algorithm}, along with earlier conjugacy results for the parameters of skew-normal distributions \citep{arellano2009shape,vieira2015nonparametric,canale2016bayesian}; 
 see also \citet{fasano2022scalable} and \citet{onorati2022extension} for additional results in binary regression settings, and \citet{durante2023skewed} for an extension of the Bernstein--von Mises theorem which clarifies the crucial role played by skewed extensions of multivariate Gaussians in Bayesian approximations and asymptotic theory.

Although the above contributions substantially enlarge the class of routinely-implemented statistical models that admit conjugate priors, all these formulations are based on fully-observed, discretized or censored Gaussian or skew-normal representations. As discussed above, this class includes multivariate linear regression along with probit, multinomial probit, and tobit models, among others, thus covering a core subset of formulations that are often employed within statistics. Nonetheless, in several applications, there is still interest in extensions of these representations which replace the Gaussian or skew-normal assumption for the error terms with alternative distributions. Popular examples in applications are generalizations of linear regression, probit, multinomial probit, and tobit models that rely on Student's $t$ or skew-$t$ error terms to incorporate robustness \citep[e.g.,][]{zellner1976bayesian,lange1989robust,albert1993bayesian,arellano2012student,marchenko12,dagne12,garay2015bayesian,matos2018multivariate,lachos2021heckman,lachos22,valeriano23}. More generally, several important contributions \citep[e.g.,][]{spanos1994modeling,branco02,islam14,barros2018generalized,zhang23} have also focused on multivariate elliptical \citep{fang90} and skew-elliptical \citep{azzalini1999statistical,branco01,azzalini2003distributions,fang03,arellano10b,adcock20} distributions, which include multivariate Gaussians, skew-normals, Student's $t$ and skew-$t$ as special cases, thereby providing a large class of practically-relevant models. However, despite the relevance of such a family, there is a lack of general, unified and tractable solutions for Bayesian inference within these settings. This is arguably due to the fact that, to date, no general conjugacy results have been established for generic models arising from fully-observed, censored, or dichotomized realizations from elliptical and skew-elliptical distributions.

Motivated by the above discussion, we cover such a gap by proving that multivariate unified skew-elliptical (SUE) distributions \citep{arellano06a,arellano10b} are the conjugate priors to the aforementioned class of models. From a technical perspective, the derivation of this result is based on specifying a general joint SUE distribution for the parameters $\bbeta$ and the noise vector $\beps$ of the response  $\by$, and then leveraging the closure under linear combination, marginalization and conditioning of SUE to prove that both $\mathcal{p}(\bbeta)$ and $\mathcal{p}(\bbeta \mid \by)$ are SUE, whenever $\mathcal{p}(\by \mid \bbeta)$ is proportional to a suitable likelihood induced by a fully-observed, censored or dichotomized elliptical or skew-elliptical distribution. This focus on the joint distribution serves only as a technical strategy to identify, under a classical Bayesian setting, which SUE priors are conjugate to specific likelihoods, thereby yielding SUE posterior distributions $\mathcal{p}(\bbeta \mid \by) \propto \mathcal{p}(\bbeta)\mathcal{p}(\by \mid \bbeta)$ via the standard application of the Bayes rule. These novel results are obtained within Section~\ref{sec:3}, leveraging both available and newly-derived SUE properties outlined in Section~\ref{sec:2}. As discussed in Sections~\ref{sec:3.1}--\ref{sec:3.3} (see Examples~\ref{exsun}--\ref{exsut3}), these advancements include, as a special case, the conjugacy properties derived in   \citet{anceschi23} for SUNs, while extending these properties to other models of potential practical interest, such as, for example, generalizations of  linear regression, probit and tobit models to Student's $t$ or skew-$t$ error terms. For these formulations, we show that the corresponding conjugate priors are multivariate unified skew-$t$ (SUT) \citep{arellano10b,wang24}. Concluding remarks can be found in Section~\ref{sec:4}, where we also clarify that besides the practical consequences for some special cases of the general results in Section~\ref{sec:3}, the conjugacy properties we derive are of broader and independent interest in expanding the theoretical analysis of the SUE family.


\section{General overview and properties of multivariate unified skew-elliptical (SUE)  distributions\label{sec:2}}

Sections~\ref{sec:2.1}--\ref{sec:2.2} provide an overview of the SUE family along with its special cases, whereas Section~\ref{sec:2.3} comprises both available closure properties and newly-derived ones that are required to prove  the novel conjugacy results within Section \ref{sec:3}. To ease the presentation, in the following, we adopt a different notation between random variables and the associated realizations only when the distinction between the two is not clear from the context.


\subsection{Multivariate unified skew-elliptical distributions}\label{sec:2.1}
Multivariate unified skew-elliptical (SUE)  distributions   \citep[e.g.,][]{arellano06a,arellano10b} arise from the perturbation of elliptical densities  \citep[e.g.,][]{fang90}, defined as
\begin{eqnarray*}\label{denskew}
    f_m(\bar{\bz} - \bxi; \bOm, g^{(m)}) =  |\bOm|^{-1/2} g^{(m)}[(\bar{\bz} - \bxi)^\top \bOm^{-1} (\bar{\bz} - \bxi)], \qquad \bar{\bz} \in \R^m, 
\end{eqnarray*}
where $\bxi \in \R^m$ denotes a location parameter, $\bOm \in \R^{m \times m}$ corresponds to a symmetric positive-definite dispersion matrix, and $g^{(m)}(\cdot): \R^+ \to \R^+$ characterizes the so-called density generator. Recalling  \citet[][]{fang90}, different choices of such a density generator $g^{(m)}$ yield a broad class of routinely-implemented elliptical densities, covering multivariate Gaussians, Student's $t$, Cauchy, logistic and Laplace,  among others. As such, for a generic vector $\bar{\bz}$ from an elliptical distribution it is customary to adopt the general notation $\bar{\bz} \sim \EC_{m}(\bxi, \bOm, g^{(m)})$, with $\bxi$, $\bOm$ and $g^{(m)}$ parameterizing such a distribution. Refer to Chapters 1--3 of  \citet[][]{fang90} for an in-depth treatment of the multivariate elliptical family, including details on the definition and properties of the density generator  $g^{(m)}$.

Due to its generality, the multivariate elliptical family has been subject of a substantial interest that has led to the development of broader classes of distributions introducing skewness in the above representation. An important and comprehensive example in this direction is provided by the SUE family \citep[][]{arellano06a,arellano10b}. Leveraging a parameterization that agrees with the unified skew-normal (SUN) sub-family introduced by \citet{arellano06a},  and with the general selection representation  in \citet{arellano06b} (see also equation 19 in \citet{arellano10b} and Section 7.1.3 of \citet{azzalini14}) a random vector $\bz \in \R^m$ has a multivariate unified skew-elliptical (SUE) distribution, i.e., $\bz \sim \SUE_{m,q}(\bxi, \bOm, \bDel, \btau, \bar\bGam, g^{(m+q)})$, if its density $\mathcal{p}(\bz)$ can be expressed as
\begin{eqnarray}\label{pdfSUE}
   \mathcal{p}(\bz) = f_{m}(\bz - \bxi; \bOm, g^{(m)}) \frac{F_q[\btau + \bDel^\top \bar\bOm^{-1}\bom^{-1} (\bz - \bxi); \bar\bGam - \bDel^\top \bar\bOm^{-1} \bDel, g_{Q(\bz)}^{(q)}]}{F_q(\btau; \bar\bGam, g^{(q)})}, \qquad {\bz} \in \R^m, 
\end{eqnarray}
where $f_{m}(\bz - \bxi; \bOm, g^{(m)})$ corresponds to the previously-defined elliptical density --- evaluated at $\bz$ --- with density generator $ g^{(m)}$, location $\bxi \in \R^m$, and positive-definite dispersion matrix $\bOm \in \R^{m \times m}$ with associated scales and correlations in  $\bom =\diag(\bOm)^{1/2} \in \R^{m \times m}$ and $\bar{\bOm} = \bom^{-1} \bOm \bom^{-1} \in \R^{m \times m}$, respectively.  In addition, $\btau \in \R^q$ is a truncation parameter, $\bDel \in \R^{m \times q}$ denotes a shape matrix, whereas $\bar\bGam \in \R^{q \times q}$ corresponds to a positive-definite dispersion matrix. Finally, $Q(\bz)$ is a quadratic form defined as $Q(\bz)= (\bz - \bxi)^\top \bOm^{-1} (\bz - \bxi) \in \R^+$, whereas {$g_{Q(\bz)}^{(q)}(u) = g^{(m+q)}[u + Q(\bz)]/g^{(m)}[Q(\bz)]$} denotes the elliptical conditional density generator. In \eqref{pdfSUE}, the quantity responsible for inducing skewness is   the $q$-dimensional centered elliptical cumulative distribution function  $F_q(\, \cdot \, ; \bar\bGam - \bDel^\top \bar\bOm^{-1} \bDel, g_{Q(\bz)}^{(q)})$ with density generator $g_{Q(\bz)}^{(q)}$, and dispersion matrix defined in the second argument. As clarified in  Lemma~\ref{lemRedundant}, the density in~ \eqref{pdfSUE} includes as a special case the one of classical elliptical distributions, which can be obtained by setting $\btau =\0$ and $\bDel =\0$. Such a result highlights that  $\btau$ and $\bDel $ play a crucial role in inducing skewness. Let us also emphasize that, in this article, the notation $\bar\bGam$  is used to denote the Pearson-correlation matrix defined as $\bar\bGam = \bgam^{-1} \bGam \bgam^{-1}$, where $\bgam = \diag(\bGam)^{1/2}$. 

To further clarify the  SUE construction, it shall be emphasized that the density expressed in \eqref{pdfSUE} can be directly obtained from the selection representation
\begin{eqnarray}\label{distY}
    \bz \stackrel{d}{=} (\bar{\bz} \mid \bar{\bz}_0 > \0), \qquad  \qquad   \begin{bmatrix}
      \bar{\bz}\\
        \bar{\bz}_0
    \end{bmatrix}
    \sim
    \EC_{m+q} \left(
        \begin{bmatrix}
            \bxi\\
            \btau
        \end{bmatrix},
        \begin{bmatrix}
            \bOm & \bom \bDel\\
            \bDel^\top \bom & \bar\bGam
        \end{bmatrix},
        g^{(m+q)}
    \right),
\end{eqnarray}
where $\bar{\bz}_0 > \0$ indicates the event ``each component of $\bar{\bz}_0$ is positive''. More specifically, under the above representation and leveraging the closure under marginalization of elliptical distributions, a direct application of the Bayes rule yields
\begin{eqnarray}\label{pdfbY}
     \mathcal{p}(\bz) = \mathcal{p}(\bar{\bz}=\bz)  \frac{\Pr(- \bar{\bz}_0 \leq \0 \mid \bar{\bz} = \bz)}{\Pr(- \bar{\bz}_0 \leq \0)}=  f_{m}(\bz - \bxi; \bOm, g^{(m)})  \frac{\Pr(- \bar{\bz}_0 \leq \0 \mid \bar{\bz} = \bz)}{\Pr(- \bar{\bz}_0 \leq \0)},  \qquad {\bz} \in \R^m.
\end{eqnarray}
Letting {$\bom \bOm^{-1}=\bar\bOm^{-1}\bom^{-1} $} and {$\bom \bOm^{-1}\bom=\bar\bOm^{-1}$}, and recalling again the closure properties of elliptical distributions~\citep[][]{fang90,azzalini14}, it follows that $- \bar{\bz}_0    \sim    \EC_{q}(  -\btau, \bar\bGam,g^{(q)})$ and {$(- \bar{\bz}_0 \mid \bar{\bz} = \bz)  \sim    \EC_{q}(-\btau-  \bDel^\top \bom \bOm^{-1}({\bz}-\bxi),\bar\bGam- \bDel^\top \bom \bOm^{-1} \bom \bDel,g_{Q(\bz)}^{(q)})$}; see also \citet{arellano06b}. As a consequence, the two probabilities at the denominator and numerator of  \eqref{pdfbY} coincide with the cumulative distribution functions, evaluated at $\btau$ and $\btau+  \bDel^\top  \bar\bOm^{-1}\bom^{-1}({\bz}-\bxi)$, of the centered elliptical distributions    $\EC_{q}(  \0, \bar\bGam,g^{(q)})$ and   $\EC_{q}(\0,\bar\bGam- \bDel^\top \bar\bOm^{-1}  \bDel,{g_{Q(\bz)}^{(q)}})$, respectively, thereby allowing to recover \eqref{pdfSUE}.

Besides providing additional insights on the quantities defining the joint density $\mathcal{p}(\bz)$ in  \eqref{pdfSUE}, the selection representation in \eqref{distY}--\eqref{pdfbY}  is also useful to derive the cumulative distribution function $\mathcal{P}(\bz)$ of the SUE distribution having density as in   \eqref{pdfSUE}. More specifically, since the SUE random vector  $ \bz \sim \SUE_{m,q}(\bxi, \bOm, \bDel, \btau, \bar\bGam, g^{(m+q)})$ admits the equivalent selection representation in \eqref{distY}, it follows that 
\begin{eqnarray*}
\mathcal{P}(\bz)=\Pr(\bar{\bz} \leq \bz, -\bar{\bz}_0 \leq \0)/\Pr(- \bar{\bz}_0 \leq \0)=\Pr(\bar{\bz} -\bxi \leq \bz -\bxi, -\bar{\bz}_0 +\btau \leq \btau)/F_q(\btau; \bar\bGam, g^{(q)}).
\end{eqnarray*}
 Therefore, leveraging again the  closure under linear combinations of elliptical distributions  \citep[][]{fang90,azzalini14}, it is possible to derive the following closed-form expression for the cumulative distribution function
\begin{equation}\label{cdfSUE}
\mathcal{P}(\bz) = 
 \frac{F_{m+q}\left(
 \begin{bmatrix}
     \bz - \bxi\\
     \btau
 \end{bmatrix};
 \begin{bmatrix}
     \bOm & - \bom \bDel\\
     -\bDel^\top \bom & \bar\bGam
 \end{bmatrix}
 , g^{(m+q)}\right)}{F_q\left(\btau; \bar\bGam, g^{(q)}\right)}.
\end{equation}
Notice that the vector $\bar\bz_0$ is often called the latent part of the distribution, whereas $q$ is the latent dimension.

Before discussing important SUE examples, we shall emphasize that, as a result of the closure under linear combinations of elliptical and SUE distributions \citep{fang90,arellano10b}, an alternative to representation  \eqref{distY} is
  \begin{eqnarray}\label{distZ}
        \bz \stackrel{d}{=} \bxi + \bom \bz^{\star}, \qquad \qquad \bz^{\star} \stackrel{d}{=} (\tilde\bz \mid \tilde\bz_0 + \btau > \0), \qquad
    \begin{bmatrix}
        \tilde\bz\\
        \tilde\bz_0
    \end{bmatrix}
    \sim \EC_{m+q} \left(
        \begin{bmatrix}
            \0\\
            \0
        \end{bmatrix},
        \begin{bmatrix}
            \bar{\bOm} &  \bDel\\
            \bDel^\top & \bar\bGam
        \end{bmatrix},
        g^{(m+q)}
    \right).
    \end{eqnarray}
Such a representation is the one adopted by \cite{arellano06a} for the sub-family of SUN distributions and is particularly convenient for deriving the mean vector and covariance matrix of $\bz$.  More specifically, leveraging \eqref{distZ}, the law of total expectation, and the previously-discussed closure properties of elliptical distributions, we have that
  \begin{eqnarray}\label{meanSUE}
\E({\bz}) =\bxi + \bom \E (\tilde\bz \mid \tilde\bz_0 + \btau > \0)=\bxi + \bom \bDel  \bar\bGam^{-1}\E (\tilde\bz_0 \mid \tilde\bz_0 + \btau > \0).
    \end{eqnarray}
Recalling \citet{arellano10b}, a related reasoning yields the following covariance matrix   \begin{eqnarray}\label{varSUE}
\mbox{var}({\bz}) =\psi\bOm+ \bom \bDel\left[ \bar\bGam^{-1} \mbox{var}(\tilde\bz_0 \mid \tilde\bz_0 + \btau > \0)\bar\bGam^{-1} -\psi \bar{\bGam}^{-1}\right]\bDel^\top \bom,
    \end{eqnarray}
where $\psi$ is a scalar whose specific form can be obtained from the derivations in     \citet{arellano10b}. The above expressions clarify that moments of SUE random vectors can be directly obtained from those of multivariate truncated elliptical distributions \citep[e.g.,][]{arismendi2017multivariate,moran2021new,galarza22,morales2022moments,valeriano2023moments}. In addition, \eqref{varSUE} shows that to enforce a lack of correlation among the entries in $\bz$, it is not sufficient to impose suitable diagonal or block-diagonal structures within $\bOm$. Rather, these constraints should be combined with additional ones, e.g., on the shape matrix $\bDel$. Such a result is useful in Section~\ref{sec:3} to derive examples of practically-meaningful priors and likelihoods inducing SUE posterior distributions. To this end, Sections~\ref{sec:3.1}, \ref{sec:3.2} and \ref{sec:3.3} state general conjugacy properties and then specialize these results to the two most widely-implemented examples of  SUE distributions that are presented in detail in Section~\ref{sec:2.2} below, namely, multivariate unified skew-normals (SUN) \citep{arellano06a}  and multivariate unified skew-$t$ (SUT) \citep{wang24}. The results for  SUN  clarify that the conjugacy properties derived by \citet{durante19,fasano2022class} and \citet{anceschi23} can be obtained as a special case, and under a different proof technique, of the more general  SUE framework introduced in the present article. Conversely, the conjugacy results stated for SUT are a novel contribution that extends to a broader class of models of potential practical interest the findings in \citet{song2016bayesian} and \citet{zhang23} on specific Student's $t$ linear regressions and multivariate probit formulations based on skew-elliptical link functions, respectively.


\subsection{Relevant sub-classes of multivariate unified skew-elliptical distributions}\label{sec:2.2}

The SUN and SUT families arise from  skewed perturbations of multivariate Gaussians and Student's $t$ densities, respectively. As such, these formulations are arguably the most relevant and practically-impactful sub-classes in the SUE family. In addition, recalling  \citet{arellano06a},  \citet{arellano10b}, and \citet{wang24}, both  SUN and SUT  admit additive stochastic representations which allow for i.i.d.\ sampling under posterior distributions belonging to these two sub-classes, thus facilitating Bayesian inference; see also \citet{yin2024stochastic} for a recent extension of these stochastic representations to more general multivariate skew-elliptical distributions, beyond  SUN and SUT. Sections~\ref{sec:2.21}--\ref{sec:2.22} provide a concise overview of  SUN and SUT sub-classes, respectively. A more extensive treatment can be found in, e.g., \citep{arellano06a,wang24}.


\subsubsection{Multivariate unified skew-normal distributions}\label{sec:2.21}

The SUN family has been introduced by  \citet{arellano06a} to provide a single class of distributions capable of unifying several extensions of the original multivariate skew-normal \citep{azzalini1996multivariate}. Relevant examples of representations that belong to such a wide class are extended multivariate skew-normals \citep{arnold2002skewed,arnold2000hidden} and closed skew-normals \citep{gonzalez2004additive,gupta2004multivariate}, among others. Recalling   \citet{arellano06a}, a random vector $\bz \in \R^m$ has multivariate unified skew-normal (SUN) distribution, i.e., $\bz \sim \SUN_{m,q}(\bxi, \bOm, \bDel, \btau, \bar\bGam)$, if its density $\mathcal{p}(\bz)$ can be expressed as
\begin{eqnarray}\label{pdfSUN}
\mathcal{p}(\bz)=    \phi_{m}(\bz - \bxi; \bOm) \frac{\Phi_q[\btau + \bDel^\top \bar\bOm^{-1}\bom^{-1} (\bz - \bxi); \bar\bGam - \bDel^\top \bar\bOm^{-1} \bDel]}{\Phi_q(\btau; \bar\bGam)}, \qquad \bz \in \R^m,
\end{eqnarray}
where $\bxi$, $ \bOm$, $\btau$,  $\bDel$ and $\bar\bGam$ have the same role and interpretation of the corresponding quantities within the general SUE representation in~\eqref{pdfSUE}, whereas   $\phi_{m}(\bz - \bxi; \bOm)$,  $\Phi_q[\btau + \bDel^\top \bar\bOm^{-1}\bom^{-1} (\bz - \bxi); \bar\bGam - \bDel^\top \bar\bOm^{-1} \bDel]$ and $\Phi_q(\btau; \bar\bGam)$ denote the density and cumulative distribution functions, evaluated at $\bz - \bxi$, $\btau + \bDel^\top \bar\bOm^{-1}\bom^{-1} (\bz - \bxi)$ and $\btau$, respectively, of the centered multivariate Gaussians with covariance matrices $ \bOm \in \R^{m \times m}$, $ \bar\bGam - \bDel^\top \bar\bOm^{-1} \bDel  \in \R^{q \times q}$ and  $\bar\bGam  \in \R^{q \times q}$. Notice that, when $\bDel=\0$, the above density coincides with that of the multivariate Gaussian $\mbox{N}_m(\bxi,\bOm)$, which can be therefore recovered as a special case, irrespectively of the value of $\btau$ and $\bar{\bGam}$.

Although the above representation is originally derived in \citet{arellano06a} under a selection representation similar to \eqref{distZ} applied to an underlying Gaussian  vector, the density in \eqref{pdfSUN} can also be derived directly from~\eqref{pdfSUE} under a suitable choice of the density generators. In particular, for $u \geq 0$, define  $g^{(m)} = \phi^{(m)}(u)=(2\pi)^{-m/2} \exp(-u/2) $, $g^{(q)} = \phi^{(q)}(u)=(2\pi)^{-q/2} \exp(-u/2) $ and recall also that, in the particular Gaussian setting, the conditional generator {$\phi^{(q)}_{Q(\by)}$} coincides with the unconditional one {$\phi^{(q)}$}. Then, recalling related derivations in \citet{arellano10b}, and replacing these density generators in  \eqref{pdfSUE}, directly yields expression \eqref{pdfSUN}, thus clarifying that SUNs are special cases of SUE distributions.  

Recalling the discussion in Section~\ref{sec:1}, the   SUN family has been at the basis of important recent advancements in conjugate Bayesian inference for a broad class of routinely-implemented representations relying on fully-observed, discretized or partially-discretized Gaussian and multivariate skew-normal models  \citep{durante19,fasano2022class,anceschi23}. As discussed in Section~\ref{sec:3}, these conjugacy properties can be recovered as a special case of those we derive for the general  SUE family, which, in turn, allow us to extend such results to larger classes, including the SUT one introduced in Section~\ref{sec:2.22}.


\subsubsection{Multivariate unified skew-$t$ distributions}\label{sec:2.22}
The success of the SUN family has motivated several extensions aimed at deriving similar skewed representations for other sub-classes of the elliptical family. A noticeable and natural generalization is provided by the class of SUT distributions \citep[e.g.,][]{wang24} which can be obtained by replacing Gaussians with suitably-defined multivariate Student's $t$ within the original selection representation of SUNs. Recalling, e.g., \citet{wang24}, this yields a density for the multivariate unified skew-$t$ (SUT) random vector $\bz \sim \SUT_{m,q}(\bxi, \bOm, \bDel, \btau, \bar\bGam, \nu)$, defined as
\begin{eqnarray}\label{pdfSUT}
\mathcal{p}(\bz)=    t_m(\bz - \bxi; \bOm, \nu) \frac{T_q [\alpha_{\nu, Q(\bz)}^{-1/2} \{\btau + \bDel^\top \bar\bOm^{-1}\bom^{-1} (\bz - \bxi)\}; \bar\bGam - \bDel^\top \bar\bOm^{-1} \bDel, \nu+m]}{T_q(\btau; \bar\bGam, \nu)}, \qquad \bz \in \R^m,
\end{eqnarray}
where $\alpha_{\nu, Q(\bz)} = [\nu + Q(\bz)]/{(\nu + m)}$, $Q(\bz) = (\bz - \bxi)^\top \bOm^{-1} (\bz - \bxi)$, and $\nu >0$ are the degrees of freedom. The remaining parameters have, likewise, similar roles and related interpretations to those in \eqref{pdfSUE}. Analogously, $ t_m(\bz - \bxi; \bOm, \nu) $, $T_q [{\alpha_{\nu, Q(\bz)}^{-1/2}} \{\btau + \bDel^\top \bar\bOm^{-1}\bom^{-1} (\bz - \bxi)\}; \bar\bGam - \bDel^\top \bar\bOm^{-1} \bDel, \nu+m]$ and $T_q(\btau; \bar\bGam, \nu)$ denote to the density and cumulative distribution functions, evaluated at the corresponding first argument, respectively, of the centered multivariate Student's $t$ distributions with scale matrices $ \bOm \in \R^{m \times m}$, $ \bar\bGam - \bDel^\top \bar\bOm^{-1} \bDel  \in \R^{q \times q}$ and  $\bar\bGam  \in \R^{q \times q}$, and degrees of freedom $\nu$, $\nu+m$ and $\nu$, respectively. Notice that, for $q = 1$, one retrieves the multivariate extended skew-$t$  in \citep{arellano10a}. Moreover, when $\btau = \0$ and $\bDel = \0$, the numerator and denominator in \eqref{pdfSUT} coincide, and, therefore, the density reduces to that of a multivariate  Student's $t$ distribution $\mathcal{T}_m(\bxi,\bOm, \nu)$ with location $\bxi$, scale $\bOm$, and degrees of freedom $\nu$. This implies that the multivariate Student's $t$ is obtained as a special case of SUT.

As for the SUN, also the SUT density within \eqref{pdfSUT} can be derived from~\eqref{pdfSUE} under a suitable choice for the density generators. In particular, let $g^{(m)} = t_v^{(m)}(u)=c(\nu,m)[1+u/\nu]^{-(\nu+m)/2}$, and $g^{(q)} = t_v^{(q)}(u)=c(\nu,q)[1+u/\nu]^{-(\nu+q)/2}$, where the generic quantity $c(a,b)$ is defined as $c(a,b)=\Gamma[(a+b)/2]/[\Gamma(a/2)(\pi a)^{b/2}]$, while $\Gamma(\cdot)$ denotes the usual gamma function. Then, recalling, for example,  \citet{arellano10b}, by replacing the generic density generators in~\eqref{pdfSUE} with those defined above and with the induced conditional density generator, yields, as a result of straightforward calculations, the SUT density in \eqref{pdfSUT}. Interestingly, as clarified in, e.g., \citet{wang24}, when $\nu \rightarrow \infty$, this density reduces to \eqref{pdfSUN}, thereby establishing a direct connection between SUT and SUN distributions. This  suggests that the conjugacy properties derived in  \citet{durante19,fasano2022class} and \citet{anceschi23} for SUN might extend to SUT, and, more generally, to SUE distributions. Two promising results in this direction have been derived in  \citet{song2016bayesian} and  in \citet{zhang23}, but only with a focus on Student's $t$ linear regression and on specific multivariate binary data settings. Leveraging the SUE properties in Section~\ref{sec:2.3}, we prove in Section~\ref{sec:3} that  SUE conjugacy holds for a substantially larger class of models which further embraces formulations of potential interest in practice.


\subsection{Properties of unified skew-elliptical distributions}\label{sec:2.3}

Lemmas \ref{linSUE}--\ref{lemRedundant} below state several central properties of SUE distributions that are at the core of the novel conjugacy results derived in Section~\ref{sec:3}. More specifically, Lemmas \ref{linSUE}--\ref{condSUE} establish closure under linear combinations, marginalization and conditioning of SUE distributions. All these properties have appeared also in  \citet{arellano10b},  but under a different parameterization. Conversely, Lemmas \ref{suecond}--\ref{lemRedundant} state novel results that are useful to study  SUE conjugacy under broad classes of models and to derive special cases of potential interest in practical contexts.

\begin{lemma}\label{linSUE}
Let $\bz \sim {\normalfont \SUE}_{m,q}(\bxi, \bOm, \bDel, \btau, \bar \bGam, g^{(m+q)})$.  In addition, denote with $\bA \in \R^{r \times m}$ a matrix with rank $r \leq m$, and let  $\bm{b} \in \R^r$ be a vector of constants. Then 
\begin{eqnarray*}
\bA \bz + \bm{b} \sim {\normalfont \SUE}_{r,q}(\bA \bxi + \bm{b}, \bA \bOm \bA^\top,  \bDel_{\bA}, \btau, \bar\bGam, g^{(r+q)}),
\end{eqnarray*}
where $ \bDel_\bA=\bom_\bA^{-1} \bA\bom \bDel$ with $\bom_\bA=\mbox{\normalfont diag}(\bA \bOm \bA^\top)^{1/2}$. Moreover, let $\bz_{C}\in \R^{|C|}$ be a generic sub-vector comprising the entries in $\bz$ with indexes in $C \subset \{1, \ldots, m\}$, and denote with $\bxi_C \in \R^{|C|}$, $\bOm_{CC} \in \R^{|C|\times |C|}$ and $\bDel_{C\cdot} \in \R^{|C| \times q}$ the associated location sub-vector, dispersion sub-matrix and shape sub-matrix, respectively. Then
\begin{eqnarray*}
 \bz_C  \sim {\normalfont \SUE}_{|C|,q}(\bxi_C ,  \bOm_{CC}, \bDel_{C\cdot}, \btau, \bar\bGam, g^{(|C|+q)}),
 \end{eqnarray*}
 for any $C \subset \{1, \ldots, m\}$.
\end{lemma}
\begin{proof}
The proof adapts related derivations in  \citet{arellano10b} to the parameterization considered in the present article. More specifically, by applying to \eqref{distY} the linearity properties of elliptical distributions, it follows that
\begin{eqnarray*}
    \begin{bmatrix}
        \bA \bar{\bz}\\
        -\bar{\bz}_0
    \end{bmatrix}
    \sim 
    \EC_{r+q}\left(
    \begin{bmatrix}
        \bA \bxi\\
        -\btau
    \end{bmatrix},
    \begin{bmatrix}
        \bA \bOm \bA^\top & - \bom_\bA \bDel_\bA\\
        -\bDel_\bA^\top \bom_\bA & \bar\bGam
    \end{bmatrix},
    g^{(r+q)}
    \right),
\end{eqnarray*}
where $\bom_\bA=\mbox{\normalfont diag}(\bA \bOm \bA^\top)^{1/2}$ and $ \bDel_\bA=\bom_\bA^{-1} \bA\bom \bDel$. Now, as a direct consequence of the selection representation in~\eqref{distY}, we have that $\bA \bz + \bm{b}$ is distributed as $(\bA \bar{\bz} + \bm{b}\, |\, \bar{\bz}_0 > \0)$. Therefore, leveraging the results and discussions in Section~\ref{sec:2.1}, the cumulative distribution function of $\bA {\bz} + \bm{b}$ can be expressed as
\begin{eqnarray*}
    \Pr(\bA \bz + \bm{b} \leq {\bz}^*) = \frac{\Pr(\bA \bar{\bz} \leq  {\bz}^* - \bm{b}, -\bar{\bz}_0 \leq \0)}{\Pr(-\bar{\bz}_0 \leq \0)} = \frac{F_{r+q}\left(\begin{bmatrix}
        \bz^* - \bA \bxi - \bm{b}\\
        \btau
    \end{bmatrix};
    \begin{bmatrix}
        \bA \bOm \bA^\top & - \bom_\bA \bDel_\bA\\
        -\bDel_\bA^\top \bom_\bA & \bar\bGam
    \end{bmatrix}, 
    g^{(r+q)}
    \right)}{F_q(\btau; \bar\bGam, g^{(q)})},
\end{eqnarray*}
for any ${\bz}^*\in \R^r$, thus obtaining the cumulative distribution function of the ${\normalfont \SUE}_{r,q}(\bA \bxi + \bm{b}, \bA \bOm \bA^\top, \bDel_\bA, \btau, \bar\bGam, g^{(r+q)})$. The result for the marginals follows as a direct consequence by letting $\bm{b}=\0$ and setting  $\bA$ equal to a suitably-defined binary selection matrix such that $ \bA \bz=\bz_C$.
\end{proof}

Lemma \ref{linSUE} ensures that linear combinations and marginals of multivariate SUE distributions are still within the same class, and the associated parameters can be derived in closed form via tractable analytical calculations. Lemma \ref{condSUE} below clarifies that a related result holds for the conditional distributions.

\begin{lemma}\label{condSUE}
Let $\bz = (\bz_1^\top, \bz_2^\top)^\top \sim{\normalfont \SUE}_{m,q}(\bxi, \bOm, \bDel, \btau, \bar\bGam, g^{(m+q)})$ with $\bz_1 \in \R^{m_1}$, $\bz_2 \in \R^{m_2}$, and parameters partitioned as
 \begin{eqnarray}\label{partSUE}
  \bxi = 
  \begin{bmatrix}
      \bxi_1\\
      \bxi_2
  \end{bmatrix},\quad
  \bOm = 
  \begin{bmatrix}
      \bOm_{11} & \bOm_{12}\\
      \bOm_{21} & \bOm_{22}
  \end{bmatrix},\quad
  \bom = 
  \begin{bmatrix}
      \bom_1 & \0\\
      \0 & \bom_2
  \end{bmatrix},\quad
  \bar\bOm = 
  \begin{bmatrix}
      \bar\bOm_{11} & \bar\bOm_{12}\\
      \bar\bOm_{21} & \bar\bOm_{22}
  \end{bmatrix},\quad
  \bDel = 
  \begin{bmatrix}
      \bDel_1 \\
      \bDel_2
  \end{bmatrix},
\end{eqnarray}
with $m_1+m_2=m$, $ \bar\bOm_{21} =\bom^{-1}_{2}\bOm_{21}\bom^{-1}_{1}$ and $ \bar\bOm_{12} =\bom^{-1}_{1}\bOm_{12}\bom^{-1}_{2}$.

Then, for $i , j \in \{1, 2\}$ and $j \neq i$, we have 
    \begin{eqnarray}\label{YiYjcond}
    (\bz_i \mid \bz_j) \sim {\normalfont \SUE}_{m_i, q}(\bxi_{i \mid j}, \bOm_{i \mid j}, \bDel_{i \mid j}, \btau_{i \mid j}, \bar\bGam_{i \mid j}, g_{Q_j(\bz_j)}^{(m_i+q)}), \qquad \bz_j \in \R^{m_j},
\end{eqnarray}
with parameters defined as
 \begin{equation}\label{partcondSUE}
 \begin{split}
&\bxi_{i \mid j} = \bxi_i + \bOm_{ij} \bOm_{jj}^{-1} (\bz_j - \bxi_j), \quad \bOm_{i \mid j} = \bOm_{ii} - \bOm_{ij} \bOm_{jj}^{-1} \bOm_{ji}, \quad \bom_{i \mid j} = \mbox{\normalfont diag}(\bOm_{i \mid j})^{1/2}, \quad \bgam_{i \mid j} = \mbox{\normalfont diag}(\bar\bGam - \bDel_j^\top \bar\bOm_{jj}^{-1} \bDel_j)^{1/2}, \\
&\bDel_{i\mid j} = \bom_{i \mid j}^{-1} (\bom_i \bDel_i - \bOm_{ij} \bOm_{jj}^{-1} \bom_j \bDel_j) \bgam_{i \mid j}^{-1}, \quad \btau_{i \mid j} = \bgam_{i \mid j}^{-1}[\btau +\bDel_j^\top \bar\bOm_{jj}^{-1} \bom_j^{-1} (\bz_j - \bxi_j)], \quad \bar{\bGam}_{i \mid j} =  \bgam_{i \mid j}^{-1}(\bar\bGam - \bDel_j^\top \bar\bOm_{jj}^{-1} \bDel_j) \bgam_{i \mid j}^{-1}, 
\end{split}
\end{equation}
and conditional density generator ${g^{(m_i + q)}_{Q_j(\bz_j)}(u)} = g^{(m + q)}[Q_j(\bz_j)+u]/g^{(m_j)}[Q_j(\bz_j)]$, with $Q_j(\bz_j)=(\bz_j - \bxi_j)^\top \bOm_{jj}^{-1} (\bz_j - \bxi_j)$.
\end{lemma}

\begin{proof}
To prove Lemma~\ref{condSUE}, let us leverage again \eqref{distY}. To this end, consider the following elliptical distribution 
\begin{eqnarray}\label{partition}
    \begin{bmatrix}
        	\bar{\bz}_1\\
       	\bar{\bz}_2\\
       	\bar{\bz}_0
    \end{bmatrix}
    \sim \EC_{m_1 + m_2 + q} \left(
    \begin{bmatrix}
        \bxi_1\\
        \bxi_2\\
        \btau
    \end{bmatrix},
    \begin{bmatrix}
        \bOm_{11} & \bOm_{12} & \bom_1 \bDel_1\\
        \bOm_{21} & \bOm_{22} & \bom_2 \bDel_2\\
        \bDel_1^\top \bom_1 & \bDel_2^\top \bom_2 & \bar \bGam
    \end{bmatrix},
    g^{(m_1 + m_2 + q)}
    \right).
\end{eqnarray}
Then, by the closure under linear combinations and conditioning of elliptical distributions  \citep[e.g.,][]{fang90}, we have that
\begin{eqnarray*}
([\bar{\bz}_i^\top,(-\bgam_{i \mid j}^{-1} \bar{\bz}_0)^\top]^\top \mid \bar{\bz}_j)
  \sim \EC_{m_i + q}\left(
    \begin{bmatrix}
        \bxi_{i \mid j}\\
        -\btau_{i \mid j} 
    \end{bmatrix}, 
    \begin{bmatrix}
        \bOm_{i \mid j} & -\bom_{i \mid j} \bDel_{i \mid j}\\
        -\bDel_{i \mid j}^\top \bom_{i \mid j} & \bar\bGam_{i \mid j}
    \end{bmatrix},
    g^{(m_i + q)}_{Q_j(\bar{\bz}_j)}
    \right), \qquad   Q_j(\bar{\bz}_j) = (\bar{\bz}_j - \bxi_j)^\top \bOm_{jj}^{-1} (\bar{\bz}_j - \bxi_j),
\end{eqnarray*}
which also implies ${(-\bgam_{i \mid j}^{-1} \bar{\bz}_0 \mid \bar{\bz}_j) \sim \EC_q(-\btau_{i \mid j}, \bar\bGam_{i \mid j}, g^{(q)}_{Q_j(\bar{\bz}_j)})}$, as a direct consequence of the closure under marginalization. Combining these results with the selection representation in \eqref{distY}, and noticing that the event $-\bar{\bz}_0 \leq \0$  is equivalent to {$-\bgam_{i \mid j}^{-1}\bar{\bz}_0  \leq \0$ (since $\bgam_{i \mid j}^{-1}$} is diagonal with positive entries), it follows that
\begin{eqnarray}
  	\mathcal{P}(\bz_i \mid \bz_j) =  \frac{\Pr(\bar{\bz}_i \leq \bz_i , -\bgam_{i \mid j}^{-1}\bar{\bz}_0  \leq \0 \, \mid \bar{\bz}_j = \bz_j)}{\Pr( -\bgam_{i \mid j}^{-1}\bar{\bz}_0  \leq \0 \, \mid \bar{\bz}_j = \bz_j)} =\frac{F_{m_i + q}\left(
    \begin{bmatrix}
        \bz_i - \bxi_{i \mid j}\\
        \btau_{i \mid j}
    \end{bmatrix};
    \begin{bmatrix}
        \bOm_{i \mid j} & - \bom_{i \mid j} \bDel_{i \mid j}\\
        -\bDel_{i \mid j}^\top \bom_{i \mid j} & \bar\bGam_{i \mid j}
    \end{bmatrix},
    g^{(m_i + q)}_{Q_j(\bz_j)}
    \right)}{F_q(\btau_{i \mid j}; \bar\bGam_{i \mid j}, g^{(q)}_{Q_j(\bz_j)})}, \label{YiYjcdf}
\end{eqnarray}
which coincides with the cumulative distribution function of the SUE in \eqref{YiYjcond} having parameters as in \eqref{partcondSUE}; see \eqref{cdfSUE} for the expression of the cumulative distribution function of a generic $\SUE_{m,q}(\bxi, \bOm, \bDel, \btau, \bar\bGam, g^{(m+q)})$. Notice that the first equality in~\eqref{YiYjcdf} follows from the fact that $\mathcal{p}(\bz_i \mid \bz_j)=\mathcal{p}(\bz_i,\bz_j)/\mathcal{p}(\bz_j)$, where $(\bz^\top_i,\bz^\top_j)^\top=\bz$ is distributed as a SUE and, by Lemma~\ref{linSUE}, the same holds for $\bz_j$. Hence, from the selection representation in \eqref{distY}, it follows that
\begin{eqnarray}
  	\mathcal{p}(\bz_i \mid \bz_j) =  \frac{\mathcal{p}(\bar{\bz}_i=\bz_i,\bar{\bz}_j=\bz_j)}{\mathcal{p}(\bar{\bz}_j=\bz_j)}\frac{\Pr(\bar{\bz}_0 \geq \0 \mid \bar{\bz}_i=\bz_i,\bar{\bz}_j=\bz_j)}{\Pr(\bar{\bz}_0 \geq \0 \mid \bar{\bz}_j=\bz_j)}=\mathcal{p}(\bar{\bz}_i=\bz_i\mid \bar{\bz}_j=\bz_j)\frac{\Pr(\bar{\bz}_0 \geq \0 \mid \bar{\bz}_i=\bz_i,\bar{\bz}_j=\bz_j)}{\Pr(\bar{\bz}_0 \geq \0 \mid \bar{\bz}_j=\bz_j)}, \label{margin_cond}
\end{eqnarray}
and, hence, $\mathcal{P}(\bz_i \mid \bz_j)=\Pr(\bar{\bz}_i \leq \bz_i , -\bgam_{i \mid j}^{-1}\bar{\bz}_0  \leq \0 \, \mid \bar{\bz}_j = \bz_j)/\Pr( -\bgam_{i \mid j}^{-1}\bar{\bz}_0  \leq \0 \, \mid \bar{\bz}_j = \bz_j)$.
\end{proof}

Lemma \ref{condSUE} guarantees that when  $\bz=(\bz^\top_i,\bz^\top_j)^\top$ is distributed as a SUE, then also the conditional distribution for a generic sub-vector $\bz_i$ belongs to the same family. Such a result conditions on a given realization $\bz_j$ for the remaining entries in $\bz$. In this respect, Lemma~\ref{suecond} states a novel finding which clarifies that the closure properties established in Lemma \ref{condSUE} can  be preserved also when conditioning on a truncation event. Similar results can be found in \cite{arellano22} and  \cite{wang24}, but with a focus on SUN and SUT sub-classes, respectively. Here, this property is proved for the whole SUE family.

\begin{lemma}\label{suecond}
 Let $\bz = (\bz_1^\top, \bz_2^\top)^\top \sim {\normalfont \SUE}_{m,q}(\bxi, \bOm, \bDel, \btau, \bar\bGam, g^{(m+q)})$ with parameters partitioned as in \eqref{partSUE}. Then 
\begin{eqnarray}\label{condtrunc1}
    (\bz_i\, |\, \bz_j > \0) \sim  {\normalfont \SUE}_{m_i, m_j + q}(\bxi_i, \bOm_{ii}, \tilde{\bDel}_{i \mid j}, \tilde{\btau}_{i \mid j}, \bar{\tilde{\bGam}}_{i \mid j}, g^{(m + q)}),    
\end{eqnarray}
where the quantities $ \tilde{\bDel}_{i \mid j}$, $\tilde{\btau}_{i \mid j},$ and $\bar{\tilde{\bGam}}_{i \mid j}$ are defined as
\begin{eqnarray}
    \tilde{\bDel}_{i\mid j} = 
    \begin{bmatrix}
        \bar\bOm_{ij} & \bDel_i
    \end{bmatrix}, \qquad   \tilde{\btau}_{i \mid j} = 
    \begin{bmatrix}
        \bom_j^{-1} \bxi_j\\
        \btau
    \end{bmatrix}, \qquad
    \bar{\tilde{\bGam}}_{i \mid j} = 
    \begin{bmatrix}
        \bar\bOm_{jj} & \bDel_j \\
        \bDel_j^\top & \bar\bGam
    \end{bmatrix},
\end{eqnarray}
for  every $i , j \in \{1, 2\}$, with $j \neq i$ and $   \bar\bOm_{ij}=\bom^{-1}_{i}\bOm_{ij}\bom^{-1}_{j}$.
\end{lemma}

\begin{proof}
Consider again the selection representation in  \eqref{distY} based on the underlying elliptical distribution in \eqref{partition}. Leveraging derivations and arguments similar to those considered in the proof of Lemma~\ref{condSUE}, we have that
\begin{eqnarray}
    \Pr(\bz_j> \0 \mid \bz_i) =  \frac{\Pr(\bar{\bz}_j > \0 , \bar{\bz}_0 > \0 \mid \bar{\bz}_i = \bz_i)}{\Pr(\bar{\bz}_0 > \0  \mid \bar{\bz}_i = \bz_i)}, \qquad  \Pr(\bz_j > \0) =  \frac{\Pr(\bar{\bz}_j > \0, \bar{\bz}_0 > \0)}{\Pr(\bar{\bz}_0 > \0)} = \frac{\Pr(\bom_j^{-1}  \bar{\bz}_j > \0,  \bar{\bz}_0 > \0)}{\Pr( \bar{\bz}_0 > \0)}. \label{14}
\end{eqnarray}
Moreover, recall  that by representation  \eqref{distY}  and the closure  properties of SUE, the marginal density for $\bz_i$ is defined as $\mathcal{p}({\bz}_i)=\mathcal{p}(\bar{\bz}_i=\bz_i)\Pr(\bar{\bz}_0 > \0 \mid \bar{\bz}_i = \bz_i)/\Pr(\bar{\bz}_0 > \0)$. Combining such an expression with those in \eqref{14} leads to
\begin{align}
   \mathcal{p}(\bz_i \mid \bz_j>0) &=\mathcal{p}(\bz_i) \frac{ \Pr(\bz_j> \0 \mid \bz_i) }{ \Pr(\bz_j> \0) }=\mathcal{p}(\bar{\bz}_i=\bz_i)\frac{\Pr(\bar{\bz}_0 > \0 \mid \bar{\bz}_i = \bz_i)}{\Pr(\bar{\bz}_0 > \0)}  \frac{\Pr(\bar{\bz}_j > \0 , \bar{\bz}_0 > \0 \mid \bar{\bz}_i = \bz_i)}{\Pr(\bar{\bz}_0 > \0  \mid \bar{\bz}_i = \bz_i)}\frac{\Pr( \bar{\bz}_0 > \0)}{\Pr(\bom_j^{-1}  \bar{\bz}_j > \0,  \bar{\bz}_0 > \0)} \nonumber\\
&=\mathcal{p}(\bar{\bz}_i=\bz_i)\frac{\Pr(\bar{\bz}_j > \0 , \bar{\bz}_0 > \0 \mid \bar{\bz}_i = \bz_i)}{\Pr(\bom_j^{-1}  \bar{\bz}_j > \0,  \bar{\bz}_0 > \0)}=\mathcal{p}(\bar{\bz}_i=\bz_i)\frac{\Pr(-\bom_j^{-1}  \bar{\bz}_j \leq \0 , -\bar{\bz}_0 \leq \0 \mid \bar{\bz}_i = \bz_i)}{\Pr(-\bom_j^{-1}  \bar{\bz}_j \leq \0,  -\bar{\bz}_0 \leq \0)}. \label{15}
\end{align}
Leveraging again the closure under linear combinations, marginalization and conditioning of elliptical distributions  \citep[e.g.,][]{fang90}, we have that
\begin{eqnarray*}
   &&     [(-\bom_j^{-1} \bar{\bz}_j)^\top,-\bar{\bz}^\top_0]^{\top}    \sim {\normalfont \EC}_{m_j + q}(
    -\tilde{\btau}_{i \mid j},
    \bar{\tilde{\bGam}}_{i \mid j},
    g^{(m_j + q)}
  ),\\
&&     ([(-\bom_j^{-1} \bar{\bz}_j)^\top,-\bar{\bz}^\top_0]^{\top} \mid {\bar{\bz}}_{i}=\bz_i)    \sim \EC_{m_j + q}(
    -\tilde{\btau}_{i \mid j} - \tilde{\bDel}_{i \mid j}^\top \bar\bOm_{ii}^{-1} \bom_i^{-1} (\bz_i - \bxi_i),
    \bar{\tilde{\bGam}}_{i \mid j} - \tilde{\bDel}_{i \mid j}^\top \bar\bOm_{ii}^{-1} \tilde{\bDel}_{i\mid j}, g^{(m_j + q)}_{Q_i(\bz_i)}),
\end{eqnarray*}
where $Q_i(\bz_i) = (\bz_i - \bxi_i)^\top \bOm_{ii}^{-1} (\bz_i - \bxi_i)$. Hence, combining the above results with expression~\eqref{15}, and recalling that $\mathcal{p}(\bar{\bz}_i=\bz_i)$ coincides with the density of the elliptical distribution $\mbox{EC}_{m_i}(\bxi_i, \bOm_{ii}, g^{(m_i)})$, we obtain
\begin{eqnarray*}
     \mathcal{p}(\bz_i \mid \bz_j>\0) = f_{m_i}(\bz_i - \bxi_i; \bOm_{ii}, g^{(m_i)}) \, \frac{F_{m_j+q}(\tilde{\btau}_{i \mid j} + \tilde{\bDel}_{i \mid j}^\top \bar\bOm_{ii}^{-1} \bom_i^{-1} (\bz_i - \bxi_i); \bar{\tilde{\bGam}}_{i \mid  j} - \tilde{\bDel}_{i\mid  j}^\top \bar\bOm_{ii}^{-1} \tilde{\bDel}_{i \mid  j}, g^{(m_j + q)}_{Q_i(\bz_i)})}{F_{m_j + q}(\tilde{\btau}_{i \mid  j}; \bar{\tilde{\bGam}}_{i\mid  j}, g^{(m_j + q)})},
\end{eqnarray*}
which coincides with the density of the SUE in Lemma~\ref{suecond}. 
\end{proof}

\begin{rmk}
\normalfont Under a similar argument and derivations, it easily follows that also  $(\bz_i \mid \bz_j < \0)$ is a SUE distribution.
\end{rmk}

Lemma~\ref{stand} below is useful for converting a SUE distribution parameterized by a matrix $\bGam$ that is not in the form of a Pearson correlation matrix to a standard SUE meeting such a constraint.

\begin{lemma}\label{stand}
    Let $\bGam$ be a positive-definite matrix, then ${\normalfont \SUE}_{m,q}(\bxi, \bOm, \bDel, \btau, \bGam, g^{(m+q)}) \stackrel{d}{=} {\normalfont \SUE}_{m,q}(\bxi, \bOm, \bDel \bgam^{-1}, \bgam^{-1} \btau, \bar\bGam, g^{(m+q)})$, where $\bar\bGam$ is a Pearson correlation matrix defined as  $\bar\bGam = \bgam^{-1} \bGam \bgam^{-1}$, with $\bgam = {\normalfont \mbox{diag}}(\bGam)^{1/2}$.
    \end{lemma}

\begin{proof}
    Let $\bz \sim \SUE_{m,q}(\bxi, \bOm, \bDel, \btau, \bGam, g^{(m+q)})$. Then, according to the selection representation  in \eqref{distY},  $\bz \stackrel{d}{=} (\bar{\bz} \mid \bar{\bz}_0 > \0)$, and its density function is given by \eqref{pdfbY}. Notice that, since $\bgam$ is a diagonal matrix with non-negative entries, then the numerator and denominator in \eqref{pdfbY} can be alternatively re-written as $\Pr(-\bar{\bz}_0 \leq \0 \mid \bar{\bz} = \bz) =\Pr(-\bgam^{-1}\bar{\bz}_0 \leq \0 \mid \bar{\bz} = \bz) $ and $\Pr(-\bar{\bz}_0 \leq \0) =\Pr(-\bgam^{-1}\bar{\bz}_0 \leq \0) $. Therefore, the proof follows directly from  \eqref{distY}--\eqref{pdfbY} and by the closure under linear combinations and conditioning of elliptical distributions.
\end{proof}

Lemma~\ref{lemRedundant} concludes this section by presenting particular cases of SUE distributions obtained under specific constraints on the associated parameters. These results are useful for detecting redundant latent dimensions and identifying interesting examples of constrained representations yielding specific models of interest under the conjugacy results derived in Section~\ref{sec:3}.

\begin{lemma}\label{lemRedundant}
    Let $\bz \sim {\normalfont \SUE}_{m,q} (\bxi, \bOm, \bDel, \btau ,\bar\bGam, g^{(m+q)})$ with parameters $\bDel$, $\btau$, and $\bar\bGam$ partitioned as
    \begin{eqnarray*}
        \bDel = \begin{bmatrix}
            \bDel_1 & \bDel_2
        \end{bmatrix}, \quad
        \btau = \begin{bmatrix}
            \btau_1\\
            \btau_2
        \end{bmatrix}, \quad
        \bar\bGam = \begin{bmatrix}
            \bar\bGam_{11} & \bar\bGam_{12} \\
            \bar\bGam_{21} & \bar\bGam_{22}
        \end{bmatrix},
    \end{eqnarray*}
    with $\bDel_i \in \R^{m \times q_i}$, $\btau_i \in \R^{q_i}$, and $\bar\bGam_{i j} \in \R^{q_i \times q_j}$, for every $i,j \in \{1, 2\}$, such that $q_1 + q_2 = q$. Then, {\normalfont (i)} if $\bDel = \0$ and $\btau = \0$, it follows that $\bz \sim {\normalfont \EC}_m(\bxi, \bOm, g^{(m)})$.   Additionally, {\normalfont  (ii)} if $\bDel_i = \0$, $\btau_i = \0$, $\bar\bGam_{i j} = \0$, for $i$ and $j$ fixed, with $j \neq i$, then $\bz \sim {\normalfont \SUE}_{m, q_j}(\bxi, \bOm, \bDel_j, \btau_j, \bar\bGam_{jj}, g^{(m+q_j)})$. Finally, {\normalfont (iii)} $F_{m+q}\{ [(\bz-\bxi)^\top,\0^\top]^\top; {\normalfont{\mbox{diag}}}(\bOm, \bar\bGam),g^{(m+q)}\}= F_m(\bz - \bxi; \bOm, g^{(m)}) \cdot F_q(\0; \bar\bGam, g^{(q)})$, where $ {\normalfont{\mbox{diag}}}(\bOm, \bar\bGam)$ denotes a block-diagonal matrix with blocks $\bOm$ and $ \bar\bGam$, respectively. 
      \end{lemma}
      
\begin{proof}
        To prove (i) note that, due to the invariance of orthant probabilities under centered elliptical distributions  \citep[e.g.,][]{fang90}, we have $F_q(\0; \bar\bGam, g^{(q)}_{Q(\bz)}) = F_q(\0; \bar\bGam, g^{(q)})$. Hence, including the constraints $\bDel = \0$ and $\btau = \0$ in \eqref{pdfSUE}, yields
        \begin{eqnarray*}
            \mathcal{p}(\bz) = f_m(\bz - \bxi; \bOm, g^{(m)}) \frac{F_q(\0; \bar\bGam, g^{(q)}_{Q(\bz)})}{F_q(\0; \bar\bGam, g^{(q)})} = f_m(\bz - \bxi; \bOm, g^{(m)}),
        \end{eqnarray*}
        which coincides with the density of the elliptical distribution $\EC_m(\bxi, \bOm, g^{(m)})$. This result allows us to prove also (iii). In particular, since $\bz \sim {\normalfont \EC}_m(\bxi, \bOm, g^{(m)})$ when both $\bDel = \0$ and $\btau = \0$, then $\mathcal{P}(\bz)=F_{m}(\bz-\bxi; \bOm, g^{(m)})$. Conversely, by  including the constraints $\bDel = \0$ and $\btau = \0$ within the general expression for the SUE cumulative distribution function in \eqref{cdfSUE} yields $\mathcal{P}(\bz)=F_{m+q}\{ [(\bz-\bxi)^\top,\0^\top]^\top; {\normalfont{\mbox{diag}}}(\bOm, \bar\bGam),g^{(m+q)}\}/F_q(\0; \bar{\bGam},g^{(q)})$. Therefore $\mathcal{P}(\bz)=F_{m}(\bz-\bxi; \bOm, g^{(m)})=F_{m+q}\{ [(\bz-\bxi)^\top,\0^\top]^\top; {\normalfont{\mbox{diag}}}(\bOm, \bar\bGam),g^{(m+q)}\}/F_q(\0; \bar{\bGam},g^{(q)})$ which implies the result stated in point (iii).
         
        Finally, to prove (ii), assume for the sake of simplicity that $\bDel_2 = \0$, $\btau_2 = \0$, and $\bar\bGam_{21} = \bar\bGam^\top_{12} = \0$, i.e., $i = 2$ and $j = 1$ (the proof for $i=1$ and $j=2$ is analogous). Then, applying (iii) to both the numerator and the denominator of the expression for $\mathcal{P}(\bz)$ in  \eqref{cdfSUE}, evaluated under the constrained parameters, leads to  
        \begin{align*}
        \mathcal{P}(\bz) &= \frac{F_{m+q_1+q_2} \left(\begin{bmatrix}
                \bz - \bxi\\
                \btau_1\\
                \0
            \end{bmatrix};
            \begin{bmatrix}
                \bOm & -\bom \bDel_1 & \0\\
                -\bDel_1^\top \bom & \bar\bGam_{11} & \0\\
                \0 & \0 & \bar\bGam_{22}
            \end{bmatrix}, g^{(m + q_1 + q_2)} \right)}{F_{q_1+q_2} \left(\begin{bmatrix}
                \btau_1\\
                \0
            \end{bmatrix};
            \begin{bmatrix}
                \bar\bGam_{11} & \0\\
                \0 & \bar\bGam_{22}
            \end{bmatrix}, g^{(q_1 + q_2)} \right)}= \frac{F_{m+q_1} \left(\begin{bmatrix}
                \bz - \bxi\\
                \btau_1
            \end{bmatrix};
            \begin{bmatrix}
                \bOm & -\bom \bDel_1\\
                -\bDel_1^\top \bom & \bar\bGam_{11}
            \end{bmatrix}, g^{(m + q_1)}  \right)}{F_{q_1}(\btau_1; \bar\bGam_{11}, g^{(q_1)})},
        \end{align*}
    where the last equality follows from the fact that $F_{q_2}(\0; \bar\bGam_{22}, g^{(q_2)})$ at the numerator and the denominator simplifies. To conclude the proof, it suffices to notice that the above cumulative distribution function coincides with the one of a $\SUE_{m,q_1}(\bxi, \bOm, \bDel_1, \btau_1, \bar\bGam_{11},g^{(m+q_1)})$.       
\end{proof}

\begin{rmk}\label{rem_sun_ii}
  \normalfont  In the particular case of a SUN distribution,  it can be shown that the restrictions $\btau = \0$ or $\btau_i = \0$ are unnecessary in Lemma \ref{lemRedundant} since there are results analogous to (i)--(iii) that follow directly by the specific properties of  Gaussian cumulative distribution functions. In particular, adapting the above proof to the SUN sub-family, it can be easily shown, for example, that if $\bDel_i = \0$ and $\bar\bGam_{i j} =\bar\bGam^\top_{j i} = \0$, $j \neq i$, then $\bz \sim {\normalfont \SUN}_{m, q_j}(\bxi, \bOm, \bDel_j, \btau_j, \bar\bGam_{jj})$, even if $\btau_i \neq \0$.
\end{rmk}

Leveraging Lemmas \ref{linSUE}--\ref{lemRedundant}, Section~\ref{sec:3} derives the novel conjugacy properties of  SUE distributions.

\section{Conjugacy properties of multivariate unified skew-elliptical (SUE)  distributions\label{sec:3}}

Sections~\ref{sec:3.1}--\ref{sec:3.3} present the new results on the SUE conjugacy properties  under a broad class of regression models for fully-observed, censored or dichotomized realizations from elliptical or skew-elliptical variables. As anticipated in Section~\ref{sec:1}, the technical derivation of these results is based on specifying a general joint SUE distribution for the parameters $\bbeta$ and the noise vector $\beps$ underlying the response $\by$. This allows to  leverage the closure properties in Section~\ref{sec:2.3}, to obtain closed-form SUE priors $\mathcal{p}(\bbeta)$  and meaningful likelihoods $\mathcal{p}(\by \mid \bbeta)$ whose combination, under the standard Bayes rule, yields posterior distributions $\mathcal{p}(\bbeta \mid \by) \propto \mathcal{p}(\bbeta)\mathcal{p}(\by \mid \bbeta)$ that still belong to the SUE class.

The above technical focus on the joint distribution $\mathcal{p}(\bbeta, \by)$ for the data $\by$ and the parameters $\bbeta$ is motivated by the fact that the results in Section~\ref{sec:3.1}--\ref{sec:3.3} crucially clarify that not all the models arising from elliptical or skew-elliptical noise vectors admit conjugate SUE priors. For this property to hold generally within the SUE family, it is necessary to consider a form of dependence between $\bbeta$ and $\beps$ due to the specific properties of the density generator. Notice that such a dependence is often weak. In particular, it allows to account for meaningful priors and models having $\bbeta$ and $\beps$ uncorrelated, while reducing to full independence under the density generators of the multivariate Gaussians and unified skew-normals. Nonetheless, such a weak dependence combined with a technical focus on  $\mathcal{p}(\bbeta, \by)$ allows for a more comprehensive investigation of  SUE conjugacy properties that would not be as immediate to prove theoretically under a direct specification of the prior $\mathcal{p}(\bbeta)$ and the likelihood  $\mathcal{p}(\by \mid \bbeta)$.

As a simple illustrative example that clarifies the above arguments, consider a univariate setting with $\mbox{Cauchy}(0,1)$ prior for $\beta$ and a $\mbox{Cauchy}(\beta,1)$ likelihood for $(y \mid \beta)$. By application of the Bayes rule, we obtain $\mathcal{p}(\beta \mid y) \propto \mathcal{p}(\beta)\mathcal{p}(y \mid \beta)$ where $\mathcal{p}(\beta)\mathcal{p}(y \mid \beta)=1/[\pi^2(1+\beta^2)(1+(y-\beta)^2)]$ is not proportional to the kernel of a Cauchy density.  Clearly, in this simple example and in general situations where conjugacy lacks, Bayesian inference can still proceed via routinely-implemented MCMC methods or deterministic approximations of the target posterior. Nonetheless, as clarified in Sections~\ref{sec:3.1}--\ref{sec:3.3},  SUE conjugacy can be still achieved under certain likelihoods induced by elliptical or skew-elliptical error terms (including instances of potential practical interest), thereby facilitating posterior inference.


\subsection{Conjugacy properties of SUE distributions in multivariate linear models}\label{sec:3.1}
Let us first study the SUE conjugacy properties under general multivariate linear models of the form
\begin{eqnarray}\label{linreg}
    \by = \bX \bbeta + \beps,
\end{eqnarray}
where $\by = (y_1, \ldots, y_n)^\top \in \R^n$ is the response vector, $\bX \in \R^{n \times p}$ corresponds to a known design matrix, $\bbeta \in \R^p$ denotes a vector of unknown parameters, often referred to as the regression coefficients, and $\beps \in \R^n$ is the error vector. Current results in Bayesian inference under the above model have established conjugacy of Gaussian or SUN priors for $\bbeta$ when combined with Gaussian or SUN noise vectors $\beps$ \citep{anceschi23}. Although these advancements cover a broad range of models, in practice, it is  of interest to consider alternative representations within the wider elliptical or skew-elliptical family, which account for heavier tails and ensure increased robustness. However, conjugacy remains unexplored in these larger classes, undermining advancements in tractable Bayesian inference. Proposition~\ref{Prop1} below covers such a gap. 
 \begin{prop}\label{Prop1}
  Assume that $(\bbeta^\top, \beps^\top)^\top \sim  {\normalfont \SUE}_{p+n,q} (\bxi, \bOm, \bDel, \btau, \bar\bGam, g^{(p+n+q)})$ with parameters partitioned as
  \begin{eqnarray}\label{beSUEpart}
    \bxi =
    \begin{bmatrix}
        \bxi_\bbeta\\ 
        \bxi_\beps
    \end{bmatrix},\qquad
    \bOm = \begin{bmatrix}
            \bOm_\bbeta & \bOm_{\bbeta \beps} \\ 
            \bOm_{\beps \bbeta} & \bOm_{\beps}
            \end{bmatrix},\qquad
    \bDel = \begin{bmatrix}
            \bDel_\bbeta \\ 
            \bDel_\beps
            \end{bmatrix}.
  \end{eqnarray}
  Then, when $ \by$ is defined as in \eqref{linreg}, it follows that $(\bbeta^\top, \by^\top)^\top \sim {\normalfont  \SUE}_{p+n,q}(\bxi^\dag, \bOm^\dag, \bDel^\dag, \btau, \bar\bGam, g^{(p+n+q)}),$ with
\begin{equation}\label{ybeta}
\begin{split}
    &\bxi^\dag = \begin{bmatrix}
        \bxi_\bbeta\\
        \bX \bxi_\bbeta + \bxi_\beps    
    \end{bmatrix} =:
    \begin{bmatrix}
        \bxi_\bbeta\\
        \bxi_\by   
    \end{bmatrix}
    , \quad \qquad
    \bDel^\dag= \begin{bmatrix}
        \bDel_\bbeta\\
        \bom_{\by}^{-1} (\bX \bom_\bbeta \bDel_\bbeta + \bom_\beps \bDel_\beps)
    \end{bmatrix} = :
    \begin{bmatrix}
        \bDel_\bbeta\\
        \bDel_\by    
    \end{bmatrix},\\
    & \bOm^\dag = \begin{bmatrix}
        \bOm_\bbeta & \bOm_\bbeta \bX^\top + \bOm_{\bbeta \beps}\\
         \bX \bOm_\bbeta + \bOm_{\beps \bbeta} & \bX \bOm_\bbeta \bX^\top + \bOm_{\beps \bbeta} \bX^\top + \bX \bOm_{\bbeta \beps} + \bOm_\beps
    \end{bmatrix} =:
    \begin{bmatrix}
        \bOm_\bbeta & \bOm_{\bbeta \by}\\
        \bOm_{\by \bbeta} & \bOm_{\by}
    \end{bmatrix},
    \end{split}
\end{equation}   
where $\bom_\beps = {\normalfont \mbox{diag}}(\bOm_\beps)^{1/2}$, $\bom_\bbeta = {\normalfont \mbox{diag}}(\bOm_\bbeta)^{1/2}$, and $\bom_\by = {\normalfont \mbox{diag}}(\bOm_\by)^{1/2}$. Moreover    
\begin{enumerate}[label=(\alph*)]
    \item{Prior distribution}. $\bbeta \sim {\normalfont \SUE}_{p,q}(\bxi_\bbeta, \bOm_\bbeta, \bDel_\bbeta, \btau, \bar\bGam, g^{(p+q)})$.
    \item{Likelihood}. $(\by \mid \bbeta) \sim  {\normalfont \SUE}_{n,q}(\bxi_{\by \mid \bbeta}, \bOm_{\by  \mid \bbeta}, \bDel_{\by  \mid \bbeta}, \btau_{\by  \mid \bbeta}, \bar\bGam_{\by \mid \bbeta}, g^{(n+q)}_{Q_\bbeta(\bm{\bbeta})}),$
    with parameters
 \begin{eqnarray*}
 \begin{split}
&\bxi_{\by \mid \bbeta} = \bxi_\by + \bOm_{\by\bbeta} \bOm_{\bbeta}^{-1} (\bbeta - \bxi_\bbeta), \quad \bOm_{\by \mid \bbeta} = \bOm_{\by} - \bOm_{\by \bbeta} \bOm_{\bbeta}^{-1} \bOm_{\bbeta \by}, \quad  \bDel_{\by\mid \bbeta} = \bom_{\by \mid \bbeta}^{-1} (\bom_\by \bDel_\by - \bOm_{\by \bbeta} \bOm_{\bbeta}^{-1} \bom_\bbeta \bDel_\bbeta) \bgam_{\by \mid \bbeta}^{-1}, \\
& \btau_{\by \mid \bbeta} = \bgam_{\by \mid \bbeta}^{-1}[\btau +\bDel_\bbeta^\top \bar\bOm_{\bbeta}^{-1} \bom_\bbeta^{-1} (\bbeta- \bxi_\bbeta)], \quad \bar{\bGam}_{\by \mid \bbeta} =  \bgam_{\by \mid \bbeta}^{-1}(\bar\bGam - \bDel_\bbeta^\top \bar\bOm_{\bbeta}^{-1} \bDel_\bbeta) \bgam_{\by \mid \bbeta}^{-1}, \quad Q_\bbeta(\bbeta) = (\bbeta - \bxi_\bbeta)^\top \bOm_\bbeta^{-1} (\bbeta - \bxi_\bbeta),  
\end{split}
\end{eqnarray*}         
     where $\bom_{\by \mid \bbeta} = {\normalfont \mbox{diag}}(\bOm_{\by \mid \bbeta})^{1/2}$ and $\bgam_{\by \mid \bbeta} = {\normalfont \mbox{diag}}(\bar\bGam - \bDel_\bbeta^\top \bar\bOm_{\bbeta}^{-1} \bDel_\bbeta)^{1/2}$.
     \item{Posterior distribution}.   $(\bbeta \mid \by) \sim {\normalfont \SUE}_{p,q}(\bxi_{\bbeta\mid \by}, \bOm_{\bbeta \mid \by}, \bDel_{\bbeta \mid \by}, \btau_{\bbeta \mid \by}, \bar\bGam_{\bbeta \mid \by}, g^{(p+q)}_{Q_\by(\by)}),$
     with parameters 
 \begin{eqnarray*}
 \begin{split}
&\bxi_{\bbeta \mid \by} = \bxi_\bbeta + \bOm_{\bbeta\by} \bOm_{\by}^{-1} (\by - \bxi_\by), \quad \bOm_{\bbeta \mid \by} = \bOm_{\bbeta} - \bOm_{ \bbeta \by} \bOm_{\by}^{-1} \bOm_{\by \bbeta}, \quad  \bDel_{\bbeta \mid \by} = \bom_{\bbeta \mid \by}^{-1} (\bom_\bbeta \bDel_\bbeta - \bOm_{\bbeta \by} \bOm_{\by}^{-1} \bom_\by \bDel_\by) \bgam_{\bbeta \mid \by}^{-1}, \\
& \btau_{\bbeta \mid \by} = \bgam_{\bbeta \mid \by}^{-1}[\btau +\bDel_\by^\top \bar\bOm_{\by}^{-1} \bom_\by^{-1} (\by- \bxi_\by)], \quad \bar{\bGam}_{\bbeta \mid \by} =  \bgam_{\bbeta \mid \by}^{-1}(\bar\bGam - \bDel_\by^\top \bar\bOm_{\by}^{-1} \bDel_\by) \bgam_{\bbeta \mid \by}^{-1}, \quad Q_\by(\by) = (\by - \bxi_\by)^\top \bOm_\by^{-1} (\by - \bxi_\by),
\end{split}
\end{eqnarray*}    
where $\bom_{\bbeta \mid \by} = {\normalfont \mbox{diag}}(\bOm_{\bbeta \mid \by})^{1/2}$ and $\bgam_{\bbeta \mid \by} = {\normalfont \mbox{diag}}(\bar\bGam - \bDel_\by^\top \bar\bOm_{\by}^{-1} \bDel_\by)^{1/2}$.
  \end{enumerate} 
\end{prop}

\begin{rmk}\label{rem_tot}
\normalfont Before proving Proposition~\ref{Prop1}, it shall be emphasized that the above results, along with those provided in  Propositions~\ref{Prop2} and \ref{Prop3}, are purposely stated in a highly-general form in order to derive a comprehensive conjugacy theory for SUE distributions that is of broader and independent interest in expanding the theoretical analysis of such a family. As clarified in Examples~\ref{exsun}--\ref{exsut3}, priors and likelihoods of potential interest in practice are only a subset of the general results in Propositions~\ref{Prop1}--\ref{Prop3}. More specifically, setting $\bOm_{\beps \bbeta}=\bOm^\top_{\bbeta \beps }=\0$, $\btau=\0$, $\bxi_{\beps}=\0$, and either $\bDel_\bbeta=\0$ or  $\bDel_\beps=\0$, would be sufficient to recover most of the priors and likelihoods of direct interest in applications.
\end{rmk}

\begin{proof}
To prove \eqref{ybeta} in Proposition~\ref{Prop1}, first notice that $(\bbeta^\top, \by^\top)^\top=\bA(\bbeta^\top, \beps^\top)^\top$, where $\bA$ denotes a known  matrix of dimension $(p+n) \times (p+n)$ with blocks $\bA_{11}=  \bI_p$, $\bA_{12}=  \0$, $\bA_{21}=  \bX$ and $\bA_{22}=  \bI_n$. Combining such a representation with the closure under linear combination properties of SUE distributions presented in Lemma \ref{linSUE}, we have that $(\bbeta^\top, \by^\top)^\top \sim {\normalfont  \SUE}_{p+n,q}(\bA \bxi, \bA \bOm \bA^\top, \bDel_\bA, \btau, \bar\bGam, g^{(p+n+q)}),$ where $\bA \bxi=\bxi^\dag$, $\bA \bOm \bA^\top=\bOm^\dag$ and $\bDel_\bA=\bDel^\dag$. As a result, the prior distribution for $\bbeta$ in Proposition~\ref{Prop1} follows directly from the closure under marginalization of the SUE family outlined in Lemma \ref{linSUE}. Similarly, the likelihood $(\by \mid \bbeta)$ and posterior  $(\bbeta \mid \by)$  in Proposition~\ref{Prop1} can be readily derived by applying the closure under conditioning properties in Lemma \ref{condSUE} to the joint SUE distribution for $(\bbeta^\top, \by^\top)^\top$ presented in Proposition~\ref{Prop1}, with parameters $\bxi^\dag$, $\bOm^\dag$ and $\bDel^\dag$ partitioned as in \eqref{ybeta}.
\end{proof}

As anticipated within Section~\ref{sec:3}, the results in Proposition \ref{Prop1} clarify that the joint distribution for  $\bbeta$ and $\beps$ requires some form of dependence to guarantee conjugacy. In this respect, notice that even when $\bxi=\0$, $\bOm_{\beps \bbeta}=\bOm^\top_{\bbeta \beps }=\0$ and $\bDel=\0$, by the closure under conditioning properties of the unified skew-elliptical family, it  follows that {$(\beps \mid \bbeta) \sim {\normalfont \SUE}_{n,q}(\0, \bOm_\beps, \0, \btau, \bar\bGam, g_{Q_\bbeta(\bbeta)}^{(n+q)})$}, which clarifies that a weak form of dependence persists in the conditional density generator. Nonetheless, such a form of higher-level dependence still allows to include within the results in Proposition~\ref{Prop1}  interesting models with uncorrelated $\bbeta$ and $\beps$ vectors. Recalling the expression for the SUE covariance matrix in \eqref{varSUE}, a sufficient condition to retrieve these uncorrelated representations is to assume either $\bDel_\bbeta=\0$ or  $\bDel_\beps=\0$, and set $\bOm_{\beps \bbeta}=\bOm^\top_{\bbeta \beps }=\0$. When both $\bDel_\bbeta$ and  $\bDel_\beps$ are $\0$, and also $\btau=\0$, by point (i) in Lemma~\ref{lemRedundant}, {$(\bbeta^\top, \beps^\top)^\top$} reduces to an elliptical distribution. As such, conjugacy under this latter class can be obtained as a special case of Proposition \ref{Prop1}.

The above discussion clarifies that full independence between  $\bbeta$ and $\beps$ cannot be generally enforced if the objective is to obtain broad conjugacy results as in Proposition \ref{Prop1} that hold for the whole SUE family. Nonetheless, in the specific setting of Gaussian density generators, which leads to the sub-class of SUN distributions, such a full independence can be enforced without undermining conjugacy. As discussed in Section~\ref{sec:2.21}, under this specific choice, the conditional density generator coincides with the unconditional one, thus allowing to enforce independence between  $\bbeta$ and $\beps$ while preserving conjugacy. This is clear from the results in \citet{anceschi23}, that establish SUN conjugacy via a classical Bayes rule perspective, without requiring to specify a joint distribution for $\bbeta$ and $\beps$ or, alternatively, $\bbeta$ and $\by$. As illustrated in Example~\ref{exsun} below, these conjugacy results can be obtained as a particular case of those  in Proposition \ref{Prop1}.

\begin{exa}[\bf Multivariate unified skew-normal (SUN) conjugacy] \label{exsun}
\normalfont The supplementary materials of \citet{anceschi23} present an example based on a classical linear regression with skew-normal errors, i.e.,  $(y_i \mid \bbeta) \sim \SN(\bx_i^\top \bbeta, \sigma^2, \alpha)$, independently for  $i \in \{1, \ldots, n\}$, and, consistent with our notation, $\mbox{SUN}_{p,q}(\bxi_{\bbeta}, \bOm_{\bbeta}, \bDel_{\bbeta}, \btau_{\bbeta}, \bar{\bGam}_{\bbeta})$ prior for $\bbeta$. This model yields a likelihood $\mathcal{p}(\by \mid \bbeta) \propto \phi_n(\by - \bX \bbeta; \sigma^2 \bI_n) \, \Phi_n(\alpha \by - \alpha \bX \bbeta; \sigma^2 \bI_n)$, proportional to a $\mbox{SUN}_{n,n}(\bX \bbeta, \sigma^2 \bI_n, \alpha \sigma \bI_n, \0, (1+\alpha^2) \sigma^2 \bI_n)$ density. Leveraging Lemma \ref{stand}, such a SUN is equivalent to $(\by \mid \bbeta) \sim  \mbox{SUN}_{n,n}(\bX \bbeta, \sigma^2 \bI_n, [\alpha/(1+\alpha^2)^{1/2}]  \bI_n, \0, \bI_n)$. Before showing  that this Bayesian formulation is a special case of the broader family of models and priors in Proposition \ref{Prop1}, it shall be emphasized that this construction also comprises classical multivariate Gaussian priors for $\bbeta$, when $\bDel_{\bbeta}=\0$, and Gaussian linear regression for $\by$ if $\alpha=0$. Replacing $\sigma^2 \bI_n$ with a full covariance matrix also leads to general multivariate versions of such models. This yields an important class of routinely-implemented formulations.

To rephrase the above Bayesian formulation within those covered by Proposition \ref{Prop1}, consider the case $(\bbeta^\top, \beps^\top)^\top \sim  {\normalfont \SUN}_{p+n,q+n} (\bxi, \bOm, \bDel, \btau, \bar\bGam)$ with parameters partitioned as
\begin{eqnarray*}
    \bxi = \begin{bmatrix}
        \bxi_\bbeta\\ 
        \0
        \end{bmatrix},\qquad
    \bOm = \begin{bmatrix}
        \bOm_\bbeta & \0 \\ 
        \0 & \sigma^2 \bI_n
        \end{bmatrix},\qquad
    \bDel = \begin{bmatrix}
        \bDel_{\bbeta} & \0\\
        \0 & \bar{\alpha} \bI_n
        \end{bmatrix},\qquad
    \btau = \begin{bmatrix}
        \btau_\bbeta\\
        \0
    \end{bmatrix},\qquad      
    \bar\bGam = \begin{bmatrix}
        \bar\bGam_{\bbeta} & \0\\
        \0 & \bI_n
    \end{bmatrix},
\end{eqnarray*}
where $\bar{\alpha}=\alpha/(1+\alpha^2)^{1/2}$. Adapting Proposition \ref{Prop1} to this setting, yields $(\bbeta^\top, \by^\top)^\top \sim {{\normalfont  \SUN}_{p+n,q+n}(\bxi^\dag, \bOm^\dag, \bDel^\dag, \btau, \bar\bGam)},$ with
\begin{eqnarray*}
    \bxi^\dag = \begin{bmatrix}
        \bxi_\bbeta\\ 
        \bX \bxi_\bbeta
        \end{bmatrix},\quad
    \bOm^\dag = \begin{bmatrix}
        \bOm_\bbeta &   \bOm_\bbeta \bX^\top \\ 
        \bX   \bOm_\bbeta &   \bX   \bOm_\bbeta  \bX^\top  +\sigma^2 \bI_n
        \end{bmatrix},\quad
    \bDel^\dag = \begin{bmatrix}
        \bDel_{\bbeta} & \0\\
        \bom_\by^{-1} \bX \bom_\bbeta \bDel_\bbeta &   \bom_\by^{-1} \sigma \bar{\alpha} \bI_n
        \end{bmatrix},\quad
    \btau = \begin{bmatrix}
        \btau_\bbeta\\
        \0
    \end{bmatrix},\quad      
    \bar\bGam = \begin{bmatrix}
        \bar\bGam_{\bbeta} & \0\\
        \0 & \bI_n
    \end{bmatrix}.
\end{eqnarray*}
As a direct consequence of the closure under linear combinations of SUE, and hence SUN, distributions, together with point (ii) in Lemma~\ref{lemRedundant} (see also Remark~\ref{rem_sun_ii}), the above formulation implies 
\begin{eqnarray*}
\bbeta \sim \mbox{SUN}_{p, q}(\bxi_{\bbeta}, \bOm_{\bbeta}, \bDel_{\bbeta}, \btau_{\bbeta}, \bar{\bGam}_{\bbeta}), \qquad \beps \sim \mbox{SUN}_{n, n}(\0, \sigma^2 \bI_n, \bar{\alpha}\bI_n, \0, \bI_n),
\end{eqnarray*}
 where the marginal for $\bbeta$ coincides with the prior considered in  \citet{anceschi23}, whereas $\beps$ corresponds to a noise vector comprising $n$ independent skew-normals.  Similarly, by applying to the above parameters the expressions for those of $(\by \mid \bbeta)$ in  Proposition \ref{Prop1}, and recalling Remark~\ref{rem_sun_ii}, yields after standard calculations the following  likelihood  
\begin{eqnarray*}
(\by \mid \bbeta) \sim  \mbox{SUN}_{n,n}(\bX \bbeta, \sigma^2 \bI_n, [\alpha/(1+\alpha^2)^{1/2}]  \bI_n, \0, \bI_n).
\end{eqnarray*}
 Such a likelihood is proportional to the one considered in  \citet{anceschi23} for the general skew-normal regression setting, which includes the Gaussian as a special case and can be readily extended to more general multivariate models. As such, Proposition \ref{Prop1} also covers  SUN conjugacy properties in commonly-implemented linear models.
\end{exa}

\vspace{2pt}
\begin{exa}[\bf Multivariate unified skew-{\boldsymbol{$t$}} (SUT) conjugacy]\label{exsut}
\normalfont As stated in Corollary~\ref{cor1}, by specializing  Proposition \ref{Prop1} to the SUT sub-family presented in Section~\ref{sec:2.22}, it is possible to derive novel conjugacy results not yet explored in the current literature, along with specific examples of potential interest in applications.
\begin{corollary}
  Consider the model in \eqref{linreg}, and assume  $ (\bbeta^\top, \beps^\top)^\top \sim {\normalfont \SUT}_{p+n,q} (\bxi, \bOm, \bDel, \btau, \bar\bGam, \nu),$ with  $\bxi, \bOm, \bDel$ partitioned as in \eqref{beSUEpart}. Then, the prior distribution for $\bbeta$ is $\bbeta \sim {\normalfont \SUT}_{p,q}(\bxi_\bbeta, \bOm_\bbeta, \bDel_\bbeta, \btau, \bar\bGam, \nu)$, whereas the likelihood and posterior are {$(\by \mid \bbeta) \sim {\normalfont \SUT}_{n,q} (\bxi_{\by \mid\bbeta}, \alpha_\bbeta \bOm_{\by \mid\bbeta}, \bDel_{\by \mid\bbeta}, \alpha_\bbeta^{-1/2} \btau_{\by \mid\bbeta}, {\bar \bGam}_{\by \mid\bbeta}, \nu + p)$} and {$(\bbeta \mid \by) \sim {\normalfont \SUT}_{p,q} (\bxi_{\bbeta \mid \by}, \alpha_\by \bOm_{\bbeta \mid \by}, \bDel_{\bbeta \mid \by}, \alpha_\by^{-1/2} \btau_{\bbeta \mid \by}, {\bar \bGam}_{\bbeta \mid \by}, \nu + n)$}, respectively. In these expressions, $\alpha_\bbeta = [\nu + Q_\bbeta(\bbeta)]/(\nu + p)$ and $\alpha_\by = [\nu + Q_\by(\by)]/(\nu + n)$. The remaining parameters, along with $Q_\bbeta(\bbeta)$ and $Q_\by(\by)$, are defined as in Proposition \ref{Prop1}. 
  \label{cor1}
\end{corollary}

\begin{proof}
The proofs follow directly by replacing the generic density generators within Proposition \ref{Prop1} with those of the Student's $t$ distribution presented in Section~\ref{sec:2.22}. Alternatively, it is possible to prove the statement leveraging the specific properties of unified skew-$t$ distributions in Proposition 11 of \citet{wang24}.
\end{proof}

Corollary~\ref{cor1} states a general conjugacy result which further includes those of the SUN as a limiting case, provided that SUT distributions converge to SUNs when $\nu \to \infty$ \citep[e.g.,][]{wang24}. In addition, suitable constraints on the parameters of the joint SUT distribution for $ (\bbeta^\top, \beps^\top)^\top$ in Corollary~\ref{cor1}, yield priors and likelihoods of potential practical interest. In particular, consider $ (\bbeta^\top, \beps^\top)^\top \sim {\normalfont \SUT}_{p+n,q} (\bxi, \bOm, \bDel, \btau, \bar\bGam, \nu),$ with 
\begin{eqnarray*}
    \bxi = \begin{bmatrix}
        \bxi_\bbeta\\ 
        \0
        \end{bmatrix},\qquad
    \bOm = \begin{bmatrix}
        \bOm_\bbeta & \0 \\ 
        \0 & \bOm_\beps 
        \end{bmatrix},\qquad
    \bDel = \begin{bmatrix}
        \0\\
       \bDel_\beps
        \end{bmatrix},
\end{eqnarray*}
and $\btau=\0$. Then, by the closure under marginalization of SUE distributions combined with Lemma~\ref{lemRedundant} and Corollary~\ref{cor1}, we have that $\bbeta$ and $\beps$ are uncorrelated and have marginals
\begin{eqnarray*}
\bbeta \sim \mathcal{T}_{p}(\bxi_\bbeta, \bOm_\bbeta, \nu), \qquad \beps \sim {\normalfont \SUT}_{n,q} (\0,\bOm_{\beps}, \bDel_{\beps}, \0, {\bar \bGam}, \nu),
\end{eqnarray*}
which yield a Bayesian multivariate regression model $\by= \bX \bbeta + \beps$ with Student's $t$ prior on $\bbeta$ and unified skew-$t$ residuals in $\beps$,  uncorrelated with $\bbeta$. In the above expression, $\mathcal{T}_{p}(\bxi_\bbeta, \bOm_\bbeta, \nu)$ denotes the $p$-variate Student's $t$ distribution with location $\bxi_\bbeta$, scale $\bOm_\bbeta$ and degrees of freedom $\nu$. By Corollary~\ref{cor1}, this implies the SUT likelihood  
\begin{eqnarray*}
 (\by \mid \bbeta) \sim {\normalfont \SUT}_{n,q} (\bX\bbeta, \alpha_\bbeta \bOm_{\beps}, \bDel_{\beps}, \0, {\bar \bGam}, \nu + p).
 \end{eqnarray*}
  Including the additional constrain $ \bDel_{\beps}=\0$ within such a formulation, and recalling again Lemma~\ref{lemRedundant}, it is possible to obtain the Bayesian Student's $t$ regression with prior and likelihood given by
 \begin{eqnarray*}
\bbeta \sim \mathcal{T}_{p}(\bxi_\bbeta, \bOm_\bbeta, \nu), \qquad  (\by \mid \bbeta) \sim \mathcal{T}_{n}(\bX\bbeta, \alpha_\bbeta \bOm_{\beps}, \nu + p),
 \end{eqnarray*}
which yields, by Corollary~\ref{cor1},  a $p$-variate Student's $t$ posterior for $ \bbeta$. This result provides an important finding which clarifies that, in specific contexts of potential practical interest, Student's $t$ -- Student's $t$ conjugacy can be attained, thereby expanding some earlier findings in \citet{song2016bayesian} on a simpler formulation. As is clear from the expression of the likelihood, for this property to hold it is necessary to incorporate the classical location dependence on $\bbeta$ via $\bX\bbeta$, together with a weak form of additional dependence induced by the scaling term  $\alpha_\bbeta = [\nu + (\bbeta - \bxi_\bbeta)^\top \bOm_\bbeta^{-1} (\bbeta - \bxi_\bbeta)]/(\nu + p)$. Recalling \citet{zhang23}, under weakly informative Student's $t$ priors employed in practice for $\bbeta$, such an effect tends to be small, and most of the dependence between $\bbeta$ and $\by$ is through the classical linear predictor $\bX\bbeta$. In addition, notice that, when $\bbeta \sim \mathcal{T}_{p}(\bxi_\bbeta, \bOm_\bbeta, \nu)$, then $(\bbeta - \bxi_\bbeta)^\top \bOm_\bbeta^{-1} (\bbeta - \bxi_\bbeta)/p$ has $F$ distribution with degrees of freedom $p$ and $\nu$, which implies that for moderate $p$ and $\nu$, the  term {$(\bbeta - \bxi_\bbeta)^\top \bOm_\bbeta^{-1} (\bbeta - \bxi_\bbeta)/p$} and, hence $\alpha_\bbeta$, shrink around $1$. 
\end{exa}

Section~\ref{sec:3.2}, shows that the conjugacy properties of SUE distributions derived in Proposition \ref{Prop1} extend even beyond multivariate linear models for continuous response vectors, to cover, in particular, also generalizations of multivariate probit and multinomial probit under elliptical or skew-elliptical link functions.


\subsection{Conjugacy properties of SUE distributions in multivariate binary models}\label{sec:3.2}
When the focus is on Bayesian modeling of multivariate binary data $\by \in \{0,1\}^n$, a natural strategy, which extends classical probit, multivariate probit and multinomial probit formulations \citep{albert1993bayesian,chib1998analysis}, is to adapt the class of models studied in Section~\ref{sec:3.1} to such a setting, by assuming
\begin{eqnarray}\label{probreg}
\by=[1(\bar{y}_1 >0), \ldots, 1(\bar{y}_n >0) ]^\top, \qquad  \qquad    \bar{\by} = \bX \bbeta + \beps,
\end{eqnarray}
where $1(\cdot)$ is the indicator function, $ \bar{\by}  = (\bar{y}_1, \ldots, \bar{y}_n)^\top \in \R^n$ and $(\bbeta^\top, \beps^\top)^\top \sim  {\normalfont \SUE}_{p+n,q} (\bxi, \bOm, \bDel, \btau, \bar\bGam, g^{(p+n+q)})$. Proposition~\ref{Prop2} clarifies from a general perspective that  SUE conjugacy can be established also in these contexts. This unifies and extends  contributions by   \citet{durante19}, \citet{fasano2022class}, \citet{anceschi23} and \citet{zhang23} on particular  SUE sub-classes; i.e., SUNs and specific skew-elliptical distributions in the  SUE family; see also Remark~\ref{rem_tot}.

\begin{prop}\label{Prop2}
Consider the binary random vector $\by \in \{0,1\}^n$, defined as in \eqref{probreg} and let $\bD_\by = \mbox{\normalfont diag}(2 y_1 - 1, \ldots, 2 y_n - 1)$. Moreover, assume again $(\bbeta^\top, \beps^\top)^\top \sim  {\normalfont \SUE}_{p+n,q} (\bxi, \bOm, \bDel, \btau, \bar\bGam, g^{(p+n+q)})$ with parameters partitioned as in \eqref{beSUEpart}. Then, $(\bbeta^\top, \bar{\by}^\top)^\top$ is a SUE with dimensions $(p+n, q)$ and parameters defined as in \eqref{ybeta}. In addition
\begin{enumerate}[label=(\alph*)]
\item{Prior distribution}. $\bbeta \sim {\normalfont \SUE}_{p,q}(\bxi_\bbeta, \bOm_\bbeta, \bDel_\bbeta, \btau, \bar\bGam, g^{(p+q)})$.
    \item{Likelihood}. $(\by \mid \bbeta)$ is a multivariate Bernoulli with probability ${\boldsymbol \Pi}_{\by \mid \bbeta}$ for the generic configuration $\by$ defined as
    \begin{eqnarray*}
    {\boldsymbol \Pi}_{\by \mid \bbeta}=\Pr(\by \mid \bbeta) = \frac{\displaystyle F_{n+q}\left(
    \begin{bmatrix}
        \bxi_{\by \mid \bbeta}\\
        \btau_{\by \mid \bbeta}
    \end{bmatrix};
    \begin{bmatrix}
        \bOm_{\by \mid \bbeta} & \bom_{\by \mid \bbeta} \bDel_{\by \mid \bbeta}\\
         \bDel_{\by \mid \bbeta}^\top \bom_{\by \mid \bbeta} & \bar\bGam_{\by \mid \bbeta}
    \end{bmatrix}, g^{(n+q)}_{Q_\bbeta(\bbeta)} \right)}
    {\displaystyle F_q(\btau_{\by \mid \bbeta}; \bar\bGam_{\by \mid \bbeta}, g^{(q)}_{Q_\bbeta(\bbeta)})}, 
    \end{eqnarray*}    
     for all $\by \in \{0,1\}^n$, and parameters available in closed form according to the following equations
    \begin{eqnarray*}
    \begin{split}
       &\bxi_{\by \mid \bbeta} = \bD_\by[\bxi_\by + \bOm_{\by \bbeta} \bOm_\bbeta^{-1}( \bbeta - \bxi_\bbeta)], \ \ \bOm_{\by \mid \bbeta} = \bD_\by(\bOm_\by - \bOm_{\by \bbeta} \bOm_\bbeta^{-1} \bOm_{\bbeta\by}) \bD_\by, \ \ \bDel_{\by \mid \bbeta} =  \bom_{\by \mid \bbeta}^{-1} \bD_\by (\bom_\by \bDel_\by - \bOm_{\by \bbeta} \bOm_{\bbeta}^{-1} \bom_\bbeta \bDel_\bbeta) \bgam_{\by \mid \bbeta}^{-1}, \\
       &
       \btau_{\by \mid \bbeta} = \bgam_{\by \mid \bbeta}^{-1} [\btau + \bDel_\bbeta^\top \bar \bOm_\bbeta^{-1} \bom_\bbeta^{-1} (\bbeta- \bxi_\bbeta)], \quad  \bar{\bGam}_{\by \mid \bbeta} = \bgam_{\by \mid \bbeta}^{-1} (\bar\bGam - \bDel_\bbeta^\top \bar\bOm_\bbeta^{-1} \bDel_\bbeta)\bgam_{\by \mid \bbeta}^{-1} ,  \quad Q_\bbeta(\bbeta) = (\bbeta - \bxi_\bbeta)^\top \bOm_\bbeta^{-1} (\bbeta - \bxi_\bbeta),
       \end{split}
    \end{eqnarray*}
    where  $\bom_{\by \mid \bbeta} = \mbox{\normalfont diag}(\bOm_{\by \mid \bbeta})^{1/2}$ and $\bgam_{\by \mid \bbeta} = \mbox{\normalfont diag}(\bar\bGam - \bDel_\bbeta^\top \bar\bOm_\bbeta^{-1} \bDel_\bbeta)^{1/2}$, while $\bxi_\by$, $\bOm_\by$, $\bom_\by$, $\bOm_{\by \bbeta}$, $\bOm_{\bbeta \by}$, and $\bDel_\by$ are as in~\eqref{ybeta}. 
    \item{Posterior distribution}. $(\bbeta \mid \by) \sim {\normalfont \SUE}_{p, n+q}(\bxi_{\bbeta \mid \by}, \bOm_{\bbeta \mid \by}, \bDel{_{\bbeta \mid \by}}, \btau{_{\bbeta \mid \by}}, \bar{{\bGam}}{_{\bbeta \mid \by}}, g^{(p+n+q)}),$ 
    \begin{eqnarray*}
        &\bxi_{\bbeta \mid \by}=\bxi_\bbeta, \quad \bOm_{\bbeta \mid \by}= \bOm_\bbeta, \quad {\bDel}_{\bbeta \mid \by} = \begin{bmatrix}
            \bar\bOm_{\bbeta \by} \bD_\by & \bDel_\bbeta
        \end{bmatrix}, \quad 
        {\btau}_{\bbeta \mid \by} = \begin{bmatrix}
            \bom_\by^{-1} \bD_\by \bxi_\by\\
            \btau
        \end{bmatrix}, \quad
        \bar{{\bGam}}_{\bbeta \mid \by} =     \begin{bmatrix}
              	\bD_\by  \bar\bOm_\by 	\bD_\by  & \bD_\by \bDel_\by\\
                \bDel_\by^\top \bD_\by & \bar\bGam
            \end{bmatrix},
    \end{eqnarray*}
    with $ \bar\bOm_{\bbeta \by}  = \bom_\bbeta^{-1} \bOm_{\bbeta \by} \bom_\by^{-1}$ and $ \bar\bOm_{\by}  = \bom_\by^{-1} \bOm_{\by} \bom_\by^{-1}$,  while $\bxi_\by$, $\bOm_\by$, $\bom_\by$, $\bOm_{\by \bbeta}$, $\bOm_{\bbeta \by}$, and $\bDel_\by$ are defined as in~\eqref{ybeta}. 
  \end{enumerate}
\end{prop}

\begin{proof}
  To prove Proposition~\ref{Prop2}, first notice that under model \eqref{probreg}, the probability of observing a given configuration $\by$ coincides with that of the event $\bD_\by \bar{\by}>\0$. Let $  \bD_\by \bar{\by}=: \bar{\by}_{\bD_\by}$, then 
\begin{eqnarray}\label{addi_bary_proof}
    \begin{bmatrix}
        \bbeta\\
        \bar{\by}_{\bD_\by}
    \end{bmatrix}
    = \bA_\by \begin{bmatrix}
        \bbeta\\
        \bar{\by}
    \end{bmatrix},
    \qquad \qquad
    \bA_\by = \begin{bmatrix}
        \bI_p & \0 \\
        \0  & \bD_\by 
    \end{bmatrix}, 
\end{eqnarray}
where $(\bbeta^\top, \bar{\by}^\top)^\top \sim {\normalfont  \SUE}_{p+n,q}(\bxi^\dag, \bOm^\dag, \bDel^\dag, \btau, \bar\bGam, g^{(p+n+q)}),$  with parameters as in \eqref{ybeta}. Therefore, by Lemma~\ref{linSUE},  we have that {$(\bbeta^\top, \bar{\by}_{\bD_\by}^\top)^\top \sim {\normalfont  \SUE}_{p+n,q}( \bA_\by \bxi^\dag, \bA_\by  \bOm^\dag \bA^\top_\by , \bDel_{\bA_\by} ^\dag, \btau, \bar\bGam, g^{(p+n+q)}),$}  with 
\begin{eqnarray*}
\begin{split}
    &\bA_\by\bxi^\dag = 
    \begin{bmatrix}
        \bxi_\bbeta\\
        {\bD_\by}\bxi_\by   
    \end{bmatrix}
    , \quad \bA_\by\bOm^\dag\bA^\top_\by =
    \begin{bmatrix}
        \bOm_\bbeta & \bOm_{\bbeta \by} {\bD_\by}\\
         {\bD_\by}\bOm_{\by \bbeta} &  {\bD_\by}\bOm_{\by} {\bD_\by}
    \end{bmatrix}, \quad
    \bDel_{\bA_\by}^\dag=     \begin{bmatrix}
        \bDel_\bbeta\\
          {\bD_\by}\bDel_\by    
    \end{bmatrix}. \    \end{split}
\end{eqnarray*}   
Under the above construction, the prior for $\bbeta$ follows directly by the SUE closure under marginalization. 

As for the likelihood of $\by$, recall that $\Pr(\by \mid \bbeta)=\Pr(\bar{\by}_{\bD_\by}>\0 \mid \bbeta)$. Moreover, by applying the results in Proposition~\ref{Prop1} to the random vector  {$(\bbeta^\top, \bar{\by}_{\bD_\by}^\top)^\top$}, we have {$(\bar{\by}_{\bD_\by} \mid \bbeta) \sim  {\normalfont \SUE}_{n,q}(\bxi_{\by \mid \bbeta}, \bOm_{\by  \mid \bbeta}, \bDel_{\by  \mid \bbeta}, \btau_{\by  \mid \bbeta}, \bar\bGam_{\by \mid \bbeta}, g^{(n+q)}_{Q_\bbeta(\bm{\bbeta})}),$} with parameters defined as in Proposition~\ref{Prop2}. Therefore,  {$\Pr(\by \mid \bbeta)$} coincides with the cumulative distribution function, evaluated at $\0$, of the SUE random variable {$(-\bar{\by}_{\bD_\by} \mid \bbeta) \sim  {\normalfont \SUE}_{n,q}(-\bxi_{\by \mid \bbeta}, \bOm_{\by  \mid \bbeta},- \bDel_{\by  \mid \bbeta}, \btau_{\by  \mid \bbeta}, \bar\bGam_{\by \mid \bbeta}, g^{(n+q)}_{Q_\bbeta(\bm{\bbeta})})$}. As a consequence, by applying \eqref{cdfSUE} to such a SUE yields
    \begin{eqnarray*}   
\Pr(\by \mid \bbeta) = \frac{\displaystyle F_{n+q}\left(
    \begin{bmatrix}
        \bxi_{\by \mid \bbeta}\\
        \btau_{\by \mid \bbeta}
    \end{bmatrix};
    \begin{bmatrix}
        \bOm_{\by \mid \bbeta} & \bom_{\by \mid \bbeta} \bDel_{\by \mid \bbeta}\\
         \bDel_{\by \mid \bbeta}^\top \bom_{\by \mid \bbeta} & \bar\bGam_{\by \mid \bbeta}
    \end{bmatrix}, g^{(n+q)}_{Q_\bbeta(\bbeta)} \right)}
    {\displaystyle F_q(\btau_{\by \mid \bbeta}; \bar\bGam_{\by \mid \bbeta}, g^{(q)}_{Q_\bbeta(\bbeta)})},     \end{eqnarray*}   
    for all $\by \in \{0,1\}^n$, as in Proposition~\ref{Prop2}. To conclude the proof, note that $(\bbeta \mid \by)$ is distributed as $(\bbeta \mid \bar{\by}_{\bD_\by}>\0)$. As a result, the posterior distribution follows directly by applying Lemma~\ref{suecond} to the SUE random vector $(\bbeta^\top, \bar{\by}_{\bD_\by}^\top)^\top$.
\end{proof}

Proposition~\ref{Prop2} clarifies that SUE distributions possess fundamental conjugacy properties also when combined with specific models for multivariate binary data. This result extends the one recently derived by \citet{zhang23} under model \eqref{probreg} with a specific focus on a skew-elliptical joint distribution for $(\bbeta^\top, \beps^\top)^\top$ which enforces lack of correlation between $\bbeta$ and $\beps$ while inducing an elliptical prior for $\bbeta$. Such a  construction can be derived, under simple linear algebra operations, as a particular case of the general SUE assumption for $(\bbeta^\top, \beps^\top)^\top$ in Proposition~\ref{Prop2},  which crucially allows to recover more general Bayesian formulations, including priors beyond the symmetric elliptical family. This connection with the contribution by \citet{zhang23} is helpful to showcase the practical impact of extending conjugacy to broader classes of models beyond classical multivariate and multinomial probit. Examples~\ref{exsun3}--\ref{exsut1} further stress this aspect with a focus on SUN and SUT distributions.


\begin{exa}[\bf Multivariate unified skew-normal (SUN) conjugacy] \label{exsun3}
\normalfont A direct and natural strategy to adapt the model studied in Example~\ref{exsun} within the binary data context, is to consider $y_i= 1(\bar{y}_i>0)$ with $(\bar{y}_i \mid \bbeta) \sim \SN(\bx_i^\top \bbeta, \sigma^2, \alpha)$, independently for $i \in \{1, \ldots, n\}$, and $\bbeta \sim \mbox{SUN}_{p, q}(\bxi_{\bbeta}, \bOm_{\bbeta}, \bDel_{\bbeta}, \btau_{\bbeta}, \bar{\bGam}_{\bbeta})$. Such a model is studied in the supplementary materials of \citet{anceschi23} as a broad extension of classical probit models to skewed link functions, further facilitating generalizations to multivariate and multinomial binary responses. Leveraging standard properties of multivariate Gaussian cumulative distribution functions, the resulting likelihood in \citet{anceschi23} can be alternatively re-expressed as proportional to  $\Phi_{2n}([(\bD_\by \bX \bbeta)^\top, \0^\top]^\top; \bSig)$, where $\bSig$ is a block matrix partitioned as $\bSig_{11}= \sigma^2 \bI_n$, $\bSig_{21}=\bSig^\top_{12}=\bD_\by \bar{\alpha} \sigma \bI_n$ and $\bSig_{22}=\bI_n$.

To recast the above Bayesian formulation within those studied in Proposition \ref{Prop2}, consider the setting $(\bbeta^\top, \beps^\top)^\top \sim  {\normalfont \SUN}_{p+n,q+n} (\bxi, \bOm, \bDel, \btau, \bar\bGam)$ with parameters partitioned as in Example~\ref{exsun}. This assumption, combined with model \eqref{probreg} and the proof of  Proposition \ref{Prop2}, implies {$(\bbeta^\top, \bar{\by}_{\bD_\by}^\top)^\top \sim {\normalfont  \SUN}_{p+n,q+n}( \bA_\by \bxi^\dag, \bA_\by  \bOm^\dag \bA^\top_\by , \bDel_{\bA_\by} ^\dag, \btau, \bar\bGam),$} with 
\begin{eqnarray*}
   \bA_\by \bxi^\dag = \begin{bmatrix}
        \bxi_\bbeta\\ 
     \bD_\by   \bX \bxi_\bbeta
        \end{bmatrix},\quad
  \bA_\by  \bOm^\dag\bA^\top_\by = \begin{bmatrix}
        \bOm_\bbeta &   \bOm_\bbeta \bX^\top   \bD_\by \\ 
        \bD_\by   \bX   \bOm_\bbeta &     \bD_\by (\bX   \bOm_\bbeta  \bX^\top  {+}\sigma^2 \bI_n)  \bD_\by 
        \end{bmatrix},\quad
    \bDel_{\bA_\by}^\dag = \begin{bmatrix}
        \bDel_{\bbeta} & \0\\
        \bom_\by^{-1} \bD_\by\bX \bom_\bbeta \bDel_\bbeta &   \bom_\by^{-1}\bD_\by \sigma \bar{\alpha} \bI_n
        \end{bmatrix},
\end{eqnarray*}
where  $  \bar{\by}_{\bD_\by}=\bD_\by \bar{\by}$, $\bar{\alpha}=\alpha/(1+\alpha^2)^{1/2}$, and $\bA_\by $ is defined in \eqref{addi_bary_proof}. The above representations, together with the closure properties of SUNs and point (ii) in Lemma~\ref{lemRedundant}, yield
 \begin{eqnarray*}
\bbeta \sim \mbox{SUN}_{p, q}(\bxi_{\bbeta}, \bOm_{\bbeta}, \bDel_{\bbeta}, \btau_{\bbeta}, \bar{\bGam}_{\bbeta}), \qquad \beps \sim \mbox{SUN}_{n, n}(\0, \sigma^2 \bI_n, \bar{\alpha}\bI_n, \0, \bI_n),
 \end{eqnarray*}
thereby recovering the SUN prior and skew-normal noise vector considered in  \citet{anceschi23}. Moreover, by Proposition~\ref{condSUE} and Remark~\ref{rem_sun_ii}, we have $(\bar{\by}_{\bD_\by} \mid \bbeta) \sim  \mbox{SUN}_{n,n}(\bD_\by\bX \bbeta, \sigma^2 \bI_n, \bD_\by \bar{\alpha} \bI_n, \0, \bI_n)$ which implies
 \begin{eqnarray*}
\Pr(\by \mid \bbeta) \propto \displaystyle \Phi_{2n}\left(
    \begin{bmatrix}
        \bD_\by\bX \bbeta\\
        \0
    \end{bmatrix};
    \begin{bmatrix}
        \sigma^2 \bI_n &  \bD_\by \bar{\alpha}\sigma \bI_n\\
         \bD_\by \bar{\alpha}\sigma \bI_n & \bI_n
    \end{bmatrix}\right), 
     \end{eqnarray*}
for all $\by \in \{0,1\}^n$, which coincides again with the likelihood in  \citet{anceschi23}. As discussed above, such a formulation includes several models of direct interest in practice. For instance, setting $\alpha=0$ yields classical probit regression, whereas replacing $\sigma^2 \bI_n$ with a full covariance matrix allows to recover multivariate probit and, for a suitable specification of $\bX$, multinomial probit, under both skewed and non-skewed link functions.  

As discussed in Example~\ref{exsut1} below, these classes of models can be further extended to specific formulations relying on Student's $t$ and skew-$t$ link functions while preserving conjugacy. As such, the practical impact of Proposition \ref{Prop2} goes beyond the  broad class of models studied in  \citet{anceschi23}.
\end{exa}

\begin{exa}[\bf Multivariate unified skew-{\boldsymbol{$t$}} (SUT) conjugacy] \label{exsut1}
\normalfont Corollary~\ref{cor2}  specializes the conjugacy properties derived in Proposition \ref{Prop2} to the specific context of the SUT sub-family presented in Section~\ref{sec:2.22}. 
\begin{corollary}
  Consider model \eqref{probreg}, with  $ (\bbeta^\top, \beps^\top)^\top \sim {\normalfont \SUT}_{p+n,q} (\bxi, \bOm, \bDel, \btau, \bar\bGam, \nu),$ and  parameters $\bxi, \bOm, \bDel$ partitioned as in \eqref{beSUEpart}. Then, the induced prior distribution is $\bbeta \sim {\normalfont \SUT}_{p,q}(\bxi_\bbeta, \bOm_\bbeta, \bDel_\bbeta, \btau, \bar\bGam, \nu)$, whereas the likelihood is equal to 
 \begin{eqnarray*}
\Pr(\by \mid \bbeta) = \frac{\displaystyle T_{n+q}\left(
    \alpha_\bbeta^{-1/2}  \begin{bmatrix}
       \bxi_{\by \mid \bbeta}\\
        \btau_{\by \mid \bbeta}
    \end{bmatrix};
    \begin{bmatrix}
        \bOm_{\by \mid \bbeta} & \bom_{\by \mid \bbeta} \bDel_{\by \mid \bbeta}\\
         \bDel_{\by \mid \bbeta}^\top \bom_{\by \mid \bbeta} & \bar\bGam_{\by \mid \bbeta}
    \end{bmatrix}, \nu+p \right)}
    {\displaystyle T_q(\alpha_\bbeta^{-1/2}  \btau_{\by \mid \bbeta}; \bar\bGam_{\by \mid \bbeta},  \nu+p )}, 
     \end{eqnarray*} 
for all $\by \in \{0,1\}^n,$ with  $\alpha_\bbeta = [\nu + (\bbeta - \bxi_\bbeta)^\top \bOm_\bbeta^{-1} (\bbeta - \bxi_\bbeta)]/(\nu + p)$, and  $ \bxi_{\by \mid \bbeta}$, $ \btau_{\by \mid \bbeta}$, $ \bOm_{\by \mid \bbeta} $, $ \bom_{\by \mid \bbeta} $, $\bDel_{\by \mid \bbeta}$, $\bar\bGam_{\by \mid \bbeta}$ as  in Proposition~\ref{Prop2}. The resulting posterior  is {$(\bbeta \mid \by) \sim {\normalfont \SUT}_{p,n+q} (\bxi_{\bbeta \mid \by},  \bOm_{\bbeta \mid \by}, \bDel_{\bbeta \mid \by}, \btau_{\bbeta \mid \by}, {\bar \bGam}_{\bbeta \mid \by}{,} \nu)$}, with parameters as in  Proposition~\ref{Prop2}. 
  \label{cor2}
\end{corollary}

\begin{proof}
To prove Corollary~\ref{cor2}, it suffices to replace the generic density generators in Proposition \ref{Prop2} with those of the Student's $t$ distribution provided in Section~\ref{sec:2.22}.
\end{proof}

As for the continuous setting in Example~\ref{exsut}, let us consider special cases of potential practical interest that arise from Corollary~\ref{cor2} under suitable constraints. In particular, define $ (\bbeta^\top, \beps^\top)^\top \sim {\normalfont \SUT}_{p+n,q} (\bxi, \bOm, \bDel, \btau, \bar\bGam, \nu),$ with 
\begin{eqnarray*}
    \bxi = \begin{bmatrix}
        \bxi_\bbeta\\ 
        \0
        \end{bmatrix},\qquad
    \bOm = \begin{bmatrix}
        \bOm_\bbeta & \0 \\ 
        \0 & \bOm_\beps 
        \end{bmatrix},\qquad
    \bDel = \begin{bmatrix}
        \0\\
       \bDel_\beps
        \end{bmatrix},
\end{eqnarray*}
and $\btau=\0$. Recalling Example~\ref{exsut}, such a formulation implies
\begin{eqnarray*}
\bbeta \sim \mathcal{T}_{p}(\bxi_\bbeta, \bOm_\bbeta, \nu), \qquad \beps \sim {\normalfont \SUT}_{n,q} (\0,\bOm_{\beps}, \bDel_{\beps}, \0, {\bar \bGam}, \nu),
\end{eqnarray*}
and hence, under \eqref{probreg}, the resulting model  for $\by$ coincides with a multivariate binary regression having unified skew-$t$ link function, and Student's $t$ prior for $\bbeta$ uncorrelated with the underlying noise vector $\beps$. As a direct consequence of  Corollary~\ref{cor2}, such a formulation yields the likelihood
\begin{eqnarray*}
\Pr(\by \mid \bbeta) \propto \displaystyle T_{n+q}\left(
    \begin{bmatrix}
      \alpha_\bbeta^{-1/2}   \bD_\by\bX \bbeta\\
        \0
    \end{bmatrix};
    \begin{bmatrix}
          \bD_\by \bOm_\beps  \bD_\by & \bom_\beps \bD_\by   \bDel_\beps  \\
       \bDel^\top_\beps \bD_\by \bom_\beps & \bar{\bGam}
    \end{bmatrix}, \nu+p\right), 
    \end{eqnarray*}
for all $\by \in \{0,1\}^n$, which provides a natural extension to more general settings of classical binary regression with $t$ link function; recall the discussion in  Example~\ref{exsut} on the impact of $\alpha_\bbeta^{-1/2}$, relative to the standard linear dependence through $\bX \bbeta$. By setting $ \bDel_{\beps}=\0$ within the above formulation, and recalling again Example~\ref{exsut}, leads to 
\begin{eqnarray*}
\bbeta \sim \mathcal{T}_{p}(\bxi_\bbeta, \bOm_\bbeta, \nu), \qquad  \Pr(\by \mid \bbeta) \propto T_n( \alpha_\bbeta^{-1/2}   \bD_\by\bX \bbeta;  \bD_\by \bOm_\beps  \bD_\by, \nu+p),
\end{eqnarray*}
which yields a closed-form SUT posterior, while further clarifying the direct link with  models implemented in practice.
 \end{exa}

Section~\ref{sec:3.3} concludes our analysis by studying SUE conjugacy in regression models for random vectors comprising both fully-observed and dichotomized data. Such a class combines results in Sections~\ref{sec:3.1}--\ref{sec:3.2} to explore a general set of formulations that extends classical tobit models in both multivariate and skew-elliptical contexts.


\subsection{Conjugacy properties of SUE distributions in multivariate censored models}\label{sec:3.3}
The classes of models studied in Sections \ref{sec:3.1}--\ref{sec:3.2} are designed for data that are either all continuous or all discretized. However, in practice, it is also possible to observe vectors comprising a combination of these two types of data. This is the case, for example, when a continuous variable is fully observed only if its value exceeds a certain threshold. Such a form of censoring is common in several applications and is typically addressed via tobit models and related extensions \citep{amemiya84,chib1992bayes}. Although common implementations rely on Gaussian noise vectors, such a class can be naturally extended to the broader unified skew-elliptical family via the following formulation
\begin{eqnarray}\label{tobreg}
\by=[\bar{y}_1 1(\bar{y}_1 >0), \ldots, \bar{y}_n 1(\bar{y}_n >0) ]^\top, \qquad  \qquad    \bar{\by} = \bX \bbeta + \beps,
\end{eqnarray}
where $1(\cdot)$ is the indicator function, $\bar{\by} = (\bar{y}_1, \ldots, \bar{y}_n)^\top \in \R^n$ and $(\bbeta^\top, \beps^\top)^\top \sim  {\normalfont \SUE}_{p+n,q} (\bxi, \bOm, \bDel, \btau, \bar\bGam, g^{(p+n+q)})$.  Recent research on such a class of Bayesian models \citep{anceschi23} has shown that when $\bbeta$ and $\beps$ have independent SUN distributions, also the posterior $(\bbeta \mid \by)$ is SUN. Proposition~\ref{Prop3} below clarifies that similar conjugacy results can be obtained when the focus is on the whole SUE family; see also Remark~\ref{rem_tot}.

 \begin{prop}\label{Prop3}
Let $\by=(\by^\top_1,\by^\top_0)^\top$ denote a generic realization from model \eqref{tobreg}, where $\by_1 \in \R_+^{n_1}$ corresponds to the vector of fully-observed data and $\by_0 =\0$ comprises the $n_0$ censored ones, with $n_1+n_0=n$. Moreover, assume again $(\bbeta^\top, \beps^\top)^\top \sim  {\normalfont \SUE}_{p+n,q} (\bxi, \bOm, \bDel, \btau, \bar\bGam, g^{(p+n+q)})$ and consider the following partition of the parameters   
\begin{eqnarray}\label{beSUEpart1}
    \bxi = \begin{bmatrix}
        \bxi_\bbeta\\ 
        \bxi_{\beps}
    \end{bmatrix}=
    \begin{bmatrix}
        \bxi_\bbeta\\ 
        \bxi_{\beps_1}\\
        \bxi_{\beps_0}
    \end{bmatrix},\qquad
    \bOm = \begin{bmatrix}
            \bOm_\bbeta & \bOm_{\bbeta \beps} \\ 
            \bOm_{\beps \bbeta} & \bOm_{\beps} 
            \end{bmatrix}= \begin{bmatrix}
            \bOm_\bbeta & \bOm_{\bbeta \beps_1}& \bOm_{\bbeta \beps_0} \\ 
            \bOm_{\beps_1 \bbeta} & \bOm_{\beps_1}  & \bOm_{\beps_1 \beps_0}\\
                        \bOm_{\beps_0 \bbeta} & \bOm_{\beps_0 \beps_1}  & \bOm_{\beps_0}
            \end{bmatrix},\qquad
    \bDel = \begin{bmatrix}
            \bDel_\bbeta \\ 
            \bDel_{\beps}
            \end{bmatrix}
=\begin{bmatrix}
            \bDel_\bbeta \\ 
            \bDel_{\beps_1}\\
                    \bDel_{\beps_0}
            \end{bmatrix}, 
  \end{eqnarray}
where $\beps_1$ and $\beps_0$ comprise the noise terms associated with the two vectors $\by_1$ and $\by_0$ in which the generic realization $\by$ is partitioned.  Then, when $\bar{\by}$ is defined as in \eqref{tobreg}, we have  $(\bbeta^\top, \bar\by^\top)^\top \sim {\normalfont  \SUE}_{p+n_1+n_0,q}(\bxi^\dag, \bOm^\dag, \bDel^\dag, \btau, \bar\bGam, g^{(p+n_1+n_0+q)})$ with
\begin{eqnarray} \label{eqtobitsplit}
 \begin{split}
 \bxi^\dag &= \begin{bmatrix}
        \bxi_\bbeta\\
        \bX_1 \bxi_\bbeta + \bxi_{\beps_1}    \\
        \bX_0 \bxi_\bbeta + \bxi_{\beps_0}    
    \end{bmatrix} =:
 \begin{bmatrix}
        \bxi_\bbeta\\
        \bxi_{\by_1}\\   
        \bxi_{\by_0}
    \end{bmatrix} = \begin{bmatrix}
               \bxi_{-\by_0}\\   
        \bxi_{\by_0}
    \end{bmatrix} 
    , \qquad \qquad 
    \bDel^\dag = \begin{bmatrix}
        \bDel_\bbeta\\
        \bom_{\by_1}^{-1} (\bX_1 \bom_\bbeta \bDel_\bbeta + \bom_{\beps_1} \bDel_{\beps_1})\\
           \bom_{\by_0}^{-1} (\bX_0 \bom_\bbeta \bDel_\bbeta + \bom_{\beps_0} \bDel_{\beps_0})
    \end{bmatrix} = :
  \begin{bmatrix}
        \bDel_\bbeta\\
        \bDel_{\by_1}\\
          \bDel_{\by_0}   
    \end{bmatrix}= \begin{bmatrix}
        \bDel_{-\by_0}\\
          \bDel_{\by_0}   
    \end{bmatrix},\\
     \bOm^\dag& = \begin{bmatrix}
        \bOm_\bbeta & \bOm_\bbeta \bX_1^\top + \bOm_{\bbeta \beps_1} &  \bOm_\bbeta \bX_0^\top + \bOm_{\bbeta \beps_0}\\
         \bX_1 \bOm_\bbeta + \bOm_{\beps_1 \bbeta} & \bX_1 \bOm_\bbeta \bX_1^\top + \bOm_{\beps_1 \bbeta} \bX_1^\top + \bX_1 \bOm_{\bbeta \beps_1} + \bOm_{\beps_1}& \bX_1 \bOm_\bbeta \bX_0^\top + \bOm_{\beps_1 \bbeta} \bX_0^\top + \bX_1 \bOm_{\bbeta \beps_0} + \bOm_{\beps_1 \beps_0}\\
 \bX_0 \bOm_\bbeta + \bOm_{\beps_0 \bbeta} & \bX_0 \bOm_\bbeta \bX_1^\top + \bOm_{\beps_0 \bbeta} \bX_1^\top + \bX_0 \bOm_{\bbeta \beps_1} + \bOm_{\beps_0 \beps_1}& \bX_0 \bOm_\bbeta \bX_0^\top + \bOm_{\beps_0 \bbeta} \bX_0^\top + \bX_0 \bOm_{\bbeta \beps_0} + \bOm_{\beps_0}      
    \end{bmatrix}    \\   &  =:
 \begin{bmatrix}
        \bOm_\bbeta & \bOm_{\bbeta \by_1} & \bOm_{\bbeta \by_0}\\
        \bOm_{\by_1 \bbeta} & \bOm_{\by_1} &   \bOm_{\by_1 \by_0}\\
         \bOm_{\by_0 \bbeta} & \bOm_{\by_0 \by_1} &   \bOm_{\by_0}
    \end{bmatrix} =   \begin{bmatrix}  \bOm_{-\by_0} &   \bOm_{\cdot \by_0 }  \\
     \bOm_{\by_0 \cdot} & \bOm_{\by_0} 
    \end{bmatrix},  \quad \begin{matrix} \bom_{\beps_0} = {\normalfont \mbox{diag}}(\bOm_{\beps_0})^{1/2}, & \bom_{\beps_1} = {\normalfont \mbox{diag}}(\bOm_{\beps_1})^{1/2}, & \bom_\bbeta = {\normalfont \mbox{diag}}(\bOm_\bbeta)^{1/2}, \\ \bom_{\by_0} = {\normalfont \mbox{diag}}(\bOm_{\by_0})^{1/2},
 &   \bom_{\by_1} = {\normalfont \mbox{diag}}(\bOm_{\by_1})^{1/2},
    \end{matrix}
    \end{split}
\end{eqnarray}
where $\bX_1 \in \R^{n_1 \times p}$ and $\bX_0 \in \R^{n_0 \times p}$ denote the two design matrices associated with the two sub-vectors $\bar{\by}_1$ and $\bar{\by}_0$  of $\bar{\by}=~(\bar{\by}^\top_1,\bar{\by}^\top_0)^\top$ in \eqref{tobreg}, which in turn correspond to the partition $\by=(\by^\top_1,\by^\top_0)^\top$. Moreover    
\begin{enumerate}[label=(\alph*)]
    \item {Prior distribution}. $\bbeta \sim {\normalfont \SUE}_{p,q}(\bxi_\bbeta, \bOm_\bbeta, \bDel_\bbeta, \btau, \bar\bGam, g^{(p+q)})$.
    \item {Likelihood}. Let $\beeta_{-\by_0}:=(\bbeta^\top, \by^\top_1)^\top$, then $(\by \mid \bbeta)$ is a multivariate random vector whose density is equal to 
     \begin{eqnarray*}
    \mathcal{p}(\by \mid \bbeta)=\mathcal{p}(\by_1, \by_0 \mid \bbeta)=\mathcal{p}(\bar{\by}_1=\by_1 \mid \bbeta) \cdot \Pr(\bar{\by}_0 \leq \0 \mid \bar{\by}_1=\by_1, \bbeta),
    \end{eqnarray*}         
     where  $\mathcal{p}(\bar{\by}_1=\by_1 \mid \bbeta) $ is the density of the  ${\normalfont \SUE}_{n_1,q}(\bxi_{\by_1 \mid \bbeta}, \bOm_{\by_1  \mid \bbeta}, \bDel_{\by_1  \mid \bbeta}, \btau_{\by_1  \mid \bbeta}, \bar\bGam_{\by_1 \mid \bbeta}, g^{(n_1+q)}_{Q_\bbeta(\bm{\bbeta})}),$ having parameters
 \begin{eqnarray*}
 \begin{split}
&\bxi_{\by_1 \mid \bbeta} = \bxi_{\by_1} + \bOm_{\by_1\bbeta} \bOm_{\bbeta}^{-1} (\bbeta - \bxi_\bbeta), \quad \bOm_{\by_1 \mid \bbeta} = \bOm_{\by_1} - \bOm_{\by_1 \bbeta} \bOm_{\bbeta}^{-1} \bOm_{\bbeta \by_1}, \quad  \bDel_{\by_1\mid \bbeta} = \bom_{\by_1 \mid \bbeta}^{-1} (\bom_{\by_1} \bDel_{\by_1} - \bOm_{\by_1 \bbeta} \bOm_{\bbeta}^{-1} \bom_\bbeta \bDel_\bbeta) \bgam_{\by_1 \mid \bbeta}^{-1}, \\
& \btau_{\by_1 \mid \bbeta} = \bgam_{\by_1 \mid \bbeta}^{-1}[\btau +\bDel_\bbeta^\top \bar\bOm_{\bbeta}^{-1} \bom_\bbeta^{-1} (\bbeta- \bxi_\bbeta)], \quad \bar{\bGam}_{\by_1 \mid \bbeta} =  \bgam_{\by_1 \mid \bbeta}^{-1}(\bar\bGam - \bDel_\bbeta^\top \bar\bOm_{\bbeta}^{-1} \bDel_\bbeta) \bgam_{\by_1 \mid \bbeta}^{-1}, \quad Q_\bbeta(\bbeta) = (\bbeta - \bxi_\bbeta)^\top \bOm_\bbeta^{-1} (\bbeta - \bxi_\bbeta),  
\end{split}
\end{eqnarray*}         
     with $\bom_{\by_1 \mid \bbeta} = {\normalfont \mbox{diag}}(\bOm_{\by_1 \mid \bbeta})^{1/2}$ and $\bgam_{\by_1 \mid \bbeta} = {\normalfont \mbox{diag}}(\bar\bGam - \bDel_\bbeta^\top \bar\bOm_{\bbeta}^{-1} \bDel_\bbeta)^{1/2}$, whereas $ \Pr(\bar{\by}_0 \leq \0 \mid \bar{\by}_1=\by_1, \bbeta)$ corresponds to the cumulative distribution function, evaluated at $\0$, of the {${\normalfont \SUE}_{n_0,q}(\bxi_{\by_0 \mid \cdot}, \bOm_{\by_0 \mid \cdot}, \bDel_{\by_0 \mid \cdot}, \btau_{\by_0 \mid \cdot}, \bar\bGam_{\by_0 \mid \cdot}, g^{(n_0+q)}_{Q_{\beeta_{-\by_0}}(\beeta_{-\by_0})})$} with
     \begin{eqnarray*}
 \begin{split}
&\bxi_{\by_0 \mid \cdot} {=} \ \bxi_{\by_0} + \bOm_{\by_0 \cdot} \bOm_{-\by_0}^{-1} (\beeta_{-\by_0} - \bxi_{-\by_0}), \ \ \ \bOm_{\by_0 \mid \cdot} {=} \ \bOm_{\by_0} - \bOm_{\by_0 \cdot} \bOm_{-\by_0}^{-1} \bOm_{\cdot\by_0}, \ \ \ \bDel_{\by_0 \mid \cdot} {=} \ \bom_{\by_0 \mid \cdot}^{-1} (\bom_{\by_0} \bDel_{\by_0} - \bOm_{\by_0 \cdot } \bOm_{-\by_0}^{-1} \bom_{-\by_0} \bDel_{-\by_0}) \bgam_{\by_0 \mid \cdot}^{-1}, \\
& \btau_{\by_0 \mid \cdot} {=} \ \bgam_{\by_0 \mid \cdot}^{-1}[\btau +\bDel_{-\by_0}^\top \bar\bOm_{-\by_0}^{-1} \bom_{-\by_0}^{-1} (\beeta_{-\by_0} - \bxi_{-\by_0})], \quad \bar{\bGam}_{\by_0 \mid \cdot} {=} \  \bgam_{\by_0 \mid \cdot}^{-1}(\bar\bGam - \bDel_{-\by_0}^\top \bar\bOm_{-\by_0}^{-1} \bDel_{-\by_0}) \bgam_{\by_0 \mid \cdot}^{-1}, 
\end{split}
\end{eqnarray*} 
with $Q_{\beeta_{-\by_0}}(\beeta_{-\by_0}) = (\beeta_{-\by_0} - \bxi_{-\by_0})^\top \bOm_{-\by_0}^{-1} (\beeta_{-\by_0} - \bxi_{-\by_0})$,      $\bom_{\by_0 \mid \cdot} = {\normalfont \mbox{diag}}(\bOm_{\by_0 \mid \cdot})^{1/2}$ and $\bgam_{\by_0 \mid \cdot} = {\normalfont \mbox{diag}}(\bar\bGam - \bDel_{-\by_0}^\top \bar\bOm_{-\by_0}^{-1} \bDel_{-\by_0})^{1/2}$.
     \item {Posterior distribution}.   $(\bbeta \mid \by) \sim {\normalfont \SUE}_{p,n_0+q}(\bxi_{\bbeta\mid \by}, \bOm_{\bbeta \mid \by}, \bDel_{\bbeta \mid \by}, \btau_{\bbeta \mid \by}, \bar\bGam_{\bbeta \mid \by}, g^{(p+n_0+q)}_{Q_{\by_1}(\by_1)}),$
     with the following parameters 
 \begin{eqnarray*}
 \begin{split}
&\bxi_{\bbeta\mid \by}=  \bxi_\bbeta+\bOm_{\bbeta \by_1} \bOm^{-1}_{\by_1}(\by_1-\bxi_{\by_1}), \qquad  \bOm_{\bbeta \mid \by} = \bOm_{\beta} -       \bOm_{\bbeta \by_1}\bOm^{-1}_{\by_1} \bOm_{\by_1 \bbeta}, \qquad \ \bom_{\bbeta \mid \by} = {\normalfont \mbox{diag}}(\bOm_{\bbeta \mid \by})^{1/2}, \\
&  \bDel_{\bbeta \mid \by} = \bom_{\bbeta \mid \by}^{-1} \left[-\bOm_{\bbeta \by_0}+ \bOm_{\bbeta \by_1}\bOm^{-1}_{\by_1}\bOm_{\by_1 \by_0} \  \ \ \ \bom_\bbeta \bDel_{\bbeta}- \bOm_{\bbeta \by_1} \bOm^{-1}_{\by_1} \bom_{\by_1} \bDel_{\by_1}\right] \bgam_{\bbeta \mid \by}^{-1}, \quad Q_{\by_1}(\by_1) = (\by_1 - \bxi_{\by_1})^\top \bOm_{\by_1}^{-1} (\by_1 - \bxi_{\by_1}),  \\
& \btau_{\bbeta \mid \by} =\bgam_{\bbeta \mid \by}^{-1} 
[(- \bxi_{\by_0}- \bOm_{\by_0 \by_1} \bOm^{-1}_{\by_1}(\by_1-\bxi_{\by_1}))^{\intercal} \  \ \ 
 (\btau+\bDel^\top_{\by_1}\bar{\bOm}^{-1}_{\by_1}\bom^{-1}_{\by_1}(\by_1-\bxi_{\by_1}))^{\intercal}]^{\intercal},\\
&  \bar{\bGam}_{\bbeta \mid \by} =\bgam_{\bbeta \mid \by}^{-1}\begin{bmatrix}
\bOm_{\by_0} -       \bOm_{\by_0 \by_1}\bOm^{-1}_{\by_1} \bOm_{\by_1 \by_0 }  & -\bom_{\by_0} \bDel_{\by_0}+ \bOm_{\by_0 \by_1} \bOm^{-1}_{\by_1} \bom_{\by_1} \bDel_{\by_1}\\
(-\bom_{\by_0} \bDel_{\by_0}+ \bOm_{\by_0 \by_1} \bOm^{-1}_{\by_1} \bom_{\by_1} \bDel_{\by_1})^\top &\bar{\bGam}-\bDel^\top_{\by_1}\bar{\bOm}_{\by_1}^{-1}\bDel_{\by_1}
 \end{bmatrix} \bgam_{\bbeta \mid \by}^{-1},
\end{split}
\end{eqnarray*}    
where $\bgam_{\bbeta \mid \by}$ is a block diagonal matrix with blocks $ {\normalfont \mbox{diag}}(\bOm_{\by_0} -       \bOm_{\by_0 \by_1}\bOm^{-1}_{\by_1} \bOm_{\by_1 \by_0 })^{1/2}$ and $ {\normalfont \mbox{diag}}(\bar{\bGam}-\bDel^\top_{\by_1}\bar{\bOm}_{\by_1}^{-1}\bDel_{\by_1})^{1/2}$.
  \end{enumerate} 
\end{prop}

\begin{proof}
The proof of \eqref{eqtobitsplit} follows directly from the closure under linear combinations of SUE derived in Lemma~\ref{linSUE}, after noticing that  $(\bbeta^\top, \bar\by^\top)^\top =(\bbeta^\top, \bar\by_1^\top,\bar\by_0^\top)^\top=\bA(\bbeta^\top, \beps^\top)^\top= \bA(\bbeta^\top, \beps_1^\top,\beps_0^\top)^\top$ with $\bA$ a block matrix having row blocks $\bA_{1 \cdot}=[\bI_p \ \0 \ \ \0]$, $\bA_{2 \cdot}=[\bX_1 \ \ \bI_{n_1} \ \ \0]$ and $\bA_{3 \cdot}=[\bX_0 \ \ \0 \ \ \bI_{n_0}]$. Under~\eqref{eqtobitsplit}, the SUE prior distribution for $\bbeta$ in Proposition~\ref{Prop3} is a direct consequence of the closure under marginalization stated in Lemma~\ref{linSUE} for the SUE class. 

As for the likelihood, notice that  $\mathcal{p}(\by \mid \bbeta)=\mathcal{p}(\by_1, \by_0 \mid \bbeta)=\mathcal{p}(\by_1 \mid \bbeta)\mathcal{p}(\by_0 \mid \by_1, \bbeta)$. Under model \eqref{tobreg}, $\mathcal{p}(\by_1 \mid \bbeta)$ is equal to $\mathcal{p}(\bar{\by}_1=\by_1 \mid \bbeta)$, which in turn coincides with the density, evaluated at $\by_1$, of $(\bar{\by}_1 \mid \bbeta)$. Therefore, by applying to the SUE random vector $(\bbeta^\top, \bar\by^\top)^\top$ --- with parameters as in~\eqref{eqtobitsplit}  --- the closure under marginalization and conditioning presented in Lemmas \ref{linSUE}--\ref{condSUE},  it directly follows that $(\bar{\by}_1 \mid \bbeta)$ is  a SUE having parameters  as in Proposition~\ref{Prop3}. For what concerns the second term $\mathcal{p}(\by_0 \mid \by_1, \bbeta)$, notice that, under model  \eqref{tobreg}, $\mathcal{p}(\by_0 \mid \by_1, \bbeta)=\Pr(\bar{\by}_0 \leq \0 \mid \bar{\by}_1=\by_1, \bbeta)$, where $(\bar{\by}_0 \mid \bar{\by}_1=\by_1, \bbeta)$ is, again, a SUE whose parameters are defined in Proposition~\ref{Prop3}. Such a latter result follows directly from the closure under linear combinations and conditioning properties in Lemmas  \ref{linSUE}--\ref{condSUE}, applied to the SUE random vector $(\bbeta^\top, \bar\by^\top)^\top=(\bbeta^\top, \bar\by_1^\top, \bar\by_0^\top)^\top$ partitioned as $(\beeta^\top_{-\by_0}, \bar\by_0^\top)^\top$.

To conclude the proof, notice that
\begin{eqnarray*}
 \Pr(\bbeta \leq \bb \mid \by)= \Pr(\bbeta \leq \bb \mid \bar{\by}_1=\by_1, \bar{\by}_0 \leq \0)=\Pr(\bbeta \leq \bb,\bar{\by}_0 \leq \0 \mid \bar{\by}_1=\by_1)/\Pr(\bar{\by}_0 \leq \0 \mid \bar{\by}_1=\by_1).
 \end{eqnarray*}
By Lemma~\ref{condSUE}, the numerator in the above expression coincides with the cumulative distribution function, evaluated at $(\bb^\top,\0^\top)^\top$, of the random vector having  ${\normalfont  \SUE}_{p+n_0,q}(\bxi_{\tiny\mbox{nu}}, \bOm_{\tiny\mbox{nu}}, \bDel_{\tiny\mbox{nu}}, \btau_{\tiny\mbox{nu}}, \bar\bGam_{\tiny\mbox{nu}}, g_{Q_{\by_1}(\by_1)}^{(p+n_0+q)})$ distribution, with parameters
\begin{eqnarray*}
\begin{split}
&   \bxi_{\tiny\mbox{nu}} {=} \ \begin{bmatrix}
        \bxi_\bbeta+\bOm_{\bbeta \by_1} \bOm^{-1}_{\by_1}(\by_1-\bxi_{\by_1})\\ 
      \bxi_{\by_0}{+} \bOm_{\by_0 \by_1} \bOm^{-1}_{\by_1}(\by_1-\bxi_{\by_1})
        \end{bmatrix} =:\begin{bmatrix}   \bxi_{\tiny\mbox{nu} \bbeta} \\ \bxi_{\tiny\mbox{nu} \by_0}
        \end{bmatrix},\ \  
 \bOm_{\tiny\mbox{nu}} {=} \ \begin{bmatrix}
 \bOm_{\bbeta} -       \bOm_{\bbeta \by_1}\bOm^{-1}_{\by_1} \bOm_{\by_1 \bbeta} & \bOm_{\bbeta \by_0}- \bOm_{\bbeta \by_1}\bOm^{-1}_{\by_1}\bOm_{\by_1 \by_0}  \\ 
    \bOm_{ \by_0 \bbeta}-     \bOm_{\by_0 \by_1}\bOm^{-1}_{\by_1}\bOm_{\by_1 \bbeta}  &   \bOm_{\by_0} -       \bOm_{\by_0 \by_1}\bOm^{-1}_{\by_1} \bOm_{\by_1 \by_0 } 
        \end{bmatrix} =:\begin{bmatrix}  \bOm_{\tiny\mbox{nu}\bbeta} & \bOm_{\tiny\mbox{nu}\bbeta \by_0} \\
         \bOm_{\tiny\mbox{nu} \by_0 \bbeta} & \bOm_{\tiny\mbox{nu}\by_0}
        \end{bmatrix},\\         
    &\bDel_{\tiny\mbox{nu}} {=} \ \begin{bmatrix}
        \bom^{-1}_{\tiny\mbox{nu} \tiny \bbeta} (\bom_\bbeta \bDel_{\bbeta}- \bOm_{\bbeta \by_1} \bOm^{-1}_{\by_1} \bom_{\by_1} \bDel_{\by_1}) \bgam^{-1}_{\tiny\mbox{nu}} \\
        \bom^{-1}_{\tiny\mbox{nu} \by_0}    ( \bom_{\by_0} \bDel_{\by_0}- \bOm_{\by_0 \by_1} \bOm^{-1}_{\by_1} \bom_{\by_1} \bDel_{\by_1})  \bgam^{-1}_{\tiny\mbox{nu}}
        \end{bmatrix}=:\begin{bmatrix}   \bDel_{\tiny\mbox{nu} \bbeta} \\ \bDel_{\tiny\mbox{nu} \by_0}
        \end{bmatrix}, \quad \begin{matrix}  \qquad \btau_{\tiny\mbox{nu}}= \bgam^{-1}_{\tiny\mbox{nu}}[\btau+\bDel^\top_{\by_1}\bar{\bOm}^{-1}_{\by_1}\bom^{-1}_{\by_1}(\by_1-\bxi_{\by_1})], \\ \bar\bGam_{\tiny\mbox{nu}}= \bgam^{-1}_{\tiny\mbox{nu}}(\bar{\bGam}-\bDel^\top_{\by_1}\bar{\bOm}_{\by_1}^{-1}\bDel_{\by_1})  \bgam^{-1}_{\tiny\mbox{nu}}, \end{matrix} 
        \end{split} 
\end{eqnarray*}
where $Q_{\by_1}(\by_1)=(\by_1-\bxi_{\by_1})^\top \bOm^{-1}_{\by_1}(\by_1-\bxi_{\by_1})$, $\bom_{\tiny\mbox{nu} \bbeta}   = {\normalfont \mbox{diag}}(  \bOm_{\tiny\mbox{nu}\bbeta} )^{1/2}$,  $\bom_{\tiny\mbox{nu} \by_0}   = {\normalfont \mbox{diag}}(  \bOm_{\tiny\mbox{nu} \by_0} )^{1/2}$ and~$\bgam_{\tiny\mbox{nu}}  = {\normalfont \mbox{diag}}(\bar{\bGam}-\bDel^\top_{\by_1}\bar{\bOm}_{\by_1}^{-1}\bDel_{\by_1})^{1/2}$. Similarly, the denominator in the expression for $\Pr(\bbeta \leq \bb \mid \by)$ coincides with the cumulative distribution function, evaluated at $\0$, of a SUE random vector. By the closure under marginalization of SUE variables, the distribution of this vector can be directly derived from the one above to obtain a  {${\normalfont  \SUE}_{n_0,q}(\bxi_{\tiny\mbox{de}}, \bOm_{\tiny\mbox{de}}, \bDel_{\tiny\mbox{de}}, \btau_{\tiny\mbox{de}}, \bar\bGam_{\tiny\mbox{de}}, g_{Q_{\by_1}(\by_1)}^{(n_0+q)})$} with parameters
\begin{eqnarray*}
&   \bxi_{\tiny\mbox{de}} =   \bxi_{\tiny\mbox{nu} \by_0},\quad 
 \bOm_{\tiny\mbox{de}} =    \bOm_{\tiny\mbox{nu}\by_0},\quad 
    \bDel_{\tiny\mbox{de}} =    \bDel_{\tiny\mbox{nu} \by_0}, \quad 
     \btau_{\tiny\mbox{de}}= \btau_{\tiny\mbox{nu}}, \quad \bar\bGam_{\tiny\mbox{de}}= \bar\bGam_{\tiny\mbox{nu}}, \quad Q_{\by_1}(\by_1)=(\by_1-\bxi_{\by_1})^\top \bOm^{-1}_{\by_1}(\by_1-\bxi_{\by_1}).
\end{eqnarray*}
Combining the above results and recalling the expression for the SUE cumulative distribution function in \eqref{cdfSUE}, we have
\begin{eqnarray*}
\begin{split}
 \Pr(\bbeta \leq \bb \mid \by)&=  \frac{F_{p+n_0+q}\left(
 \begin{bmatrix}
     \bb - \bxi_{\tiny\mbox{nu} \bbeta}\\
      - \bxi_{\tiny\mbox{nu} \by_0} \\
     \btau_{\tiny\mbox{nu}}
 \end{bmatrix};
 \begin{bmatrix}
     \bOm_{\tiny\mbox{nu}\bbeta} &   \bOm_{\tiny\mbox{nu}\bbeta \by_0}&  - \bom_{\tiny\mbox{nu}\bbeta}  \bDel_{\tiny\mbox{nu}\bbeta} \\
      \bOm_{\tiny\mbox{nu} \by_0\bbeta} & \bOm_{\tiny\mbox{nu}\by_0}  &- \bom_{\tiny\mbox{nu}\by_0}  \bDel_{\tiny\mbox{nu}\by_0}\\
       -\bDel_{\tiny\mbox{nu}\bbeta} ^\top \bom_{\tiny\mbox{nu}\bbeta}  &  -\bDel_{\tiny\mbox{nu}\by_0} ^\top \bom_{\tiny\mbox{nu}\by_0}    &\bar\bGam_{\tiny\mbox{nu}}
 \end{bmatrix}
 , g_{Q_{\by_1}(\by_1)}^{(p+n_0+q)}\right)F_q\left(\btau_{\tiny\mbox{nu}}; \bar\bGam_{\tiny\mbox{nu}}, g_{Q_{\by_1}(\by_1)}^{(q)}\right)}{F_q\left(\btau_{\tiny\mbox{nu}}; \bar\bGam_{\tiny\mbox{nu}}, g_{Q_{\by_1}(\by_1)}^{(q)}\right)F_{n_0+q}\left(
 \begin{bmatrix}
      - \bxi_{\tiny\mbox{nu} \by_0}\\
      \btau_{\tiny\mbox{nu}}\\
 \end{bmatrix};
 \begin{bmatrix}
     \bOm_{\tiny\mbox{nu}\by_0}  & - \bom_{\tiny\mbox{nu}\by_0}  \bDel_{\tiny\mbox{nu}\by_0}\\
     -\bDel_{\tiny\mbox{nu}\by_0} ^\top \bom_{\tiny\mbox{nu}\by_0}    & \bar\bGam_{\tiny\mbox{nu}} \\
 \end{bmatrix}
 , g_{Q_{\by_1}(\by_1)}^{(n_0+q)}\right)}\\
 &=\frac{F_{p+n_0+q}\left(
 \begin{bmatrix}
     \bb - \bxi_{\bbeta \mid \by}\\
     \btau_{\bbeta \mid \by}
 \end{bmatrix};
 \begin{bmatrix}
     \bOm_{\bbeta \mid \by} &   - \bom_{\bbeta \mid \by}  \bDel_{\bbeta \mid \by} \\
       -\bDel_{\bbeta \mid \by}^\top \bom_{\bbeta \mid \by}    &\bar\bGam_{\bbeta \mid \by}
 \end{bmatrix}
 , g_{Q_{\by_1}(\by_1)}^{(p+n_0+q)}\right)}{F_{n_0+q}\left(
\btau_{\bbeta \mid \by};\bar\bGam_{\bbeta \mid \by}
 , g_{Q_{\by_1}(\by_1)}^{(n_0+q)}\right)},
 \end{split}
\end{eqnarray*}
which coincides with the cumulative distribution functions of the {${\normalfont \SUE}_{p,n_0+q}(\bxi_{\bbeta\mid \by}, \bOm_{\bbeta \mid \by}, \bDel_{\bbeta \mid \by}, \btau_{\bbeta \mid \by}, \bar\bGam_{\bbeta \mid \by}, g^{(p+n_0+q)}_{Q_{\by_1}(\by_1)})$} posterior for $\bbeta$ whose parameters, after suitable standardizations based on Lemma~\ref{stand}, are defined as in  Proposition~\ref{Prop3}.
\end{proof}

Proposition~\ref{Prop3} states a general result that establishes SUE conjugacy for a broad class of models whose likelihood factorizes as the product of multivariate elliptical densities and cumulative distribution functions. These likelihoods substantially extend classical tobit representations to multivariate and skewed contexts while covering a broader family of noise terms beyond the Gaussian ones. As clarified in Examples \ref{exsun5}--\ref{exsut3}, albeit general, such a result allows the recovery of Bayesian formulations of potential interest in practice while ensuring conjugacy under these representations.

\begin{exa}[\bf Multivariate unified skew-normal (SUN) conjugacy] \label{exsun5}
\normalfont Classical tobit models consider $y_i=\bar{y}_i 1(\bar{y}_i>0)$ with $(\bar{y}_i \mid \bbeta) \sim \mbox{N}(\bx_i^\top \bbeta, \sigma^2)$, independently for $i \in \{1, \ldots, n\}$. A natural extension that incorporates skewness within these representations replaces $\mbox{N}(\bx_i^\top \bbeta, \sigma^2)$ with $\SN(\bx_i^\top \bbeta, \sigma^2, \alpha)$ \citep{hutton11}. Under this setting, which includes the classical and routinely-implemented tobit formulation when $\alpha=0$, \citet{anceschi23} have shown that SUN priors for $\bbeta$, i.e., $\bbeta \sim \mbox{SUN}_{p, q}(\bxi_{\bbeta}, \bOm_{\bbeta}, \bDel_{\bbeta}, \btau_{\bbeta}, \bar{\bGam}_{\bbeta})$, yield posterior distributions within the same class. Such a result can be derived as a very special case of Proposition~\ref{Prop3} under a Gaussian density generator and suitable constraints on the parameters.

To clarify the above point, assume again the case $(\bbeta^\top, \beps^\top)^\top \sim  {\normalfont \SUN}_{p+n,q+n} (\bxi, \bOm, \bDel, \btau, \bar\bGam)$ with parameters partitioned as in Example~\ref{exsun}. Moreover, consider the  partitioning  {$(\by^\top_1,\by^\top_0)^\top$} defined in Proposition \ref{Prop3} for a generic realization $\by$ from model \eqref{tobreg}.  This construction, combined with the results in  \eqref{eqtobitsplit} and the proof of Proposition \ref{Prop3}, implies that {$(\bbeta^\top, \bar\by^\top)^\top \sim {\normalfont  \SUN}_{p+n_1+n_0,q+n_1+n_0}(\bxi^\dag, \bOm^\dag, \bDel^\dag, \btau, \bar\bGam),$} with
\begin{eqnarray*}
 \begin{split}
 \bxi^\dag = \begin{bmatrix}
        \bxi_\bbeta\\
        \bX_1 \bxi_\bbeta    \\
        \bX_0 \bxi_\bbeta   
    \end{bmatrix},  \ \
    \bDel^\dag = \begin{bmatrix}
        \bDel_\bbeta & \0 & \0\\
        \bom_{\by_1}^{-1} \bX_1 \bom_\bbeta \bDel_\bbeta &   \bom_{\by_1}^{-1} \sigma \bar{\alpha} \bI_{n_1} & \0\\\
           \bom_{\by_0}^{-1} \bX_0 \bom_\bbeta \bDel_\bbeta & \0 &  \bom_{\by_0}^{-1} \sigma \bar{\alpha} \bI_{n_0}
    \end{bmatrix}, \ \
     \bOm^\dag = \begin{bmatrix}
        \bOm_\bbeta & \bOm_\bbeta \bX_1^\top &  \bOm_\bbeta \bX_0^\top\\
         \bX_1 \bOm_\bbeta& \bX_1 \bOm_\bbeta \bX_1^\top  + \sigma^2 \bI_{n_1}& \bX_1 \bOm_\bbeta \bX_0^\top\\
 \bX_0 \bOm_\bbeta & \bX_0 \bOm_\bbeta \bX_1^\top & \bX_0 \bOm_\bbeta \bX_0^\top + \sigma^2 \bI_{n_0}     
    \end{bmatrix},       \end{split}
\end{eqnarray*}
where  $\bar{\alpha}=\alpha/(1+\alpha^2)^{1/2}$,  $\bom_\bbeta=\mbox{diag}(\bOm_\bbeta)^{1/2}$, $\bom_{\by_1}=\mbox{diag}( \bX_1 \bOm_\bbeta \bX_1^\top  + \sigma^2 \bI_{n_1})^{1/2}$ and $\bom_{\by_0}=\mbox{diag}( \bX_0 \bOm_\bbeta \bX_0^\top  + \sigma^2 \bI_{n_0})^{1/2}$. These results, combined with the closure properties of SUNs and point (ii) in Lemma~\ref{lemRedundant}, yield
 \begin{eqnarray*}
\bbeta \sim \mbox{SUN}_{p, q}(\bxi_{\bbeta}, \bOm_{\bbeta}, \bDel_{\bbeta}, \btau_{\bbeta}, \bar{\bGam}_{\bbeta}), \qquad \beps \sim \mbox{SUN}_{n, n}(\0, \sigma^2 \bI_n, \bar{\alpha}\bI_n, \0, \bI_n),
 \end{eqnarray*}
that coincide with the SUN prior and skew-normal noise vector for the extension of the tobit model analyzed in the supplementary materials of  \citet{anceschi23}. In addition, leveraging Lemma~\ref{condSUE} and Remark~\ref{rem_sun_ii}, we have $(\bar{\by}_{1} \mid \bbeta) \sim  \mbox{SUN}_{n_1,n_1}(\bX_1 \bbeta, \sigma^2 \bI_{n_1}, \bar{\alpha} \bI_{n_1}, \0, \bI_{n_1})$. Similarly, by the properties of Gaussian density generators, it follows that under the above constraints for the SUN parameters, $(\bar{\by}_{0} \perp \bar{\by}_{1}  \mid \bbeta)$. Therefore, $\mathcal{p}(\bar{\by}_{0} \mid \bar{\by}_{1}=\by_1,  \bbeta)=\mathcal{p}(\bar{\by}_{0} \mid  \bbeta)$, and hence, by the same derivations that led to the  SUN for $(\bar{\by}_{1} \mid \bbeta) $, we obtain $(\bar{\by}_{0} {\mid} \bbeta) \sim  \mbox{SUN}_{n_0,n_0}(\bX_0 \bbeta, \sigma^2 \bI_{n_0}, \bar{\alpha} \bI_{n_0}, \0, \bI_{n_0})$. Combining these results with Proposition \ref{Prop3}, Lemma~\ref{condSUE} and Remark~\ref{rem_sun_ii}, and recalling the expression for the SUN density and cumulative distribution function (see e.g., \citet{arellano22}), leads to
 \begin{eqnarray*}
\mathcal{p}(\by_1, \by_0 \mid \bbeta) \propto \phi_{n_1}(\by_1-\bX_1 \bbeta; \sigma^2 \bI_{n_1})\Phi_{n_1+2n_0}\left(
    \begin{bmatrix}
    \alpha(\by_1-\bX_1 \bbeta)\\
  -      \bX_0 \bbeta\\
        \0
    \end{bmatrix};
    \begin{bmatrix}
            \sigma^2 \bI_{n_1} &  \0 & \0\\
      \0 & \sigma^2 \bI_{n_0} &-   \bar{\alpha}\sigma \bI_{n_0}\\
        \0 &- \bar{\alpha}\sigma \bI_{n_0} & \bI_{n_0}
    \end{bmatrix}\right),
 \end{eqnarray*}
which coincides again with the likelihood in  \citet{anceschi23}. As for the models explored in Examples \ref{exsun} and \ref{exsun3}, the above representation also includes several formulations of direct interest in practice. In particular, setting $\alpha= 0$ yields classical tobit regression, whereas replacing $\sigma^2 \bI_n$ with a full covariance matrix allows to recover multivariate extensions of tobit models, including those based on skewed link functions.  

Example~\ref{exsut3} concludes our analysis by clarifying that similar, but yet-unexplored, conjugacy properties can be established also when the focus is on models for Student's $t$ or skew-$t$ censored observations.
\end{exa}

\begin{exa}[\bf Multivariate unified skew-{\boldsymbol{$t$}} (SUT) conjugacy] \label{exsut3}
\normalfont Conjugacy properties for generalizations of tobit models relying on Student's $t$ or skew-$t$ censored observations are currently lacking. As stated within Corollary~\ref{cor3},  these properties can be derived as special cases of Proposition \ref{Prop3} under  Student's $t$ density generators. 
\begin{corollary}
  Consider model \eqref{tobreg}, with  $ (\bbeta^\top, \beps^\top)^\top \sim {\normalfont \SUT}_{p+n,q} (\bxi, \bOm, \bDel, \btau, \bar\bGam, \nu),$ and  parameters $\bxi, \bOm, \bDel$ partitioned as in \eqref{beSUEpart1}. Then, the induced prior distribution is $\bbeta \sim {\normalfont \SUT}_{p,q}(\bxi_\bbeta, \bOm_\bbeta, \bDel_\bbeta, \btau, \bar\bGam, \nu)$, whereas the likelihood is equal to
 \begin{eqnarray*}
\mathcal{p}(\by \mid \bbeta)=\mathcal{p}(\by_1, \by_0 \mid \bbeta)=\mathcal{p}(\bar{\by}_1=\by_1 \mid \bbeta) \cdot \Pr(\bar{\by}_0 \leq \0 \mid \bar{\by}_1=\by_1, \bbeta),
   \end{eqnarray*}
where $\mathcal{p}(\bar{\by}_1=\by_1 \mid \bbeta) $ denotes the density of the  ${\normalfont \SUT}_{n_1,q}(\bxi_{\by_1 \mid \bbeta}, \alpha_{\bbeta}\bOm_{\by_1  \mid \bbeta}, \bDel_{\by_1  \mid \bbeta}, \alpha^{-1/2}_{\bbeta}\btau_{\by_1  \mid \bbeta}, \bar\bGam_{\by_1 \mid \bbeta}, \nu+p),$ with parameters as in Proposition \ref{Prop3}, while  $\Pr(\bar{\by}_0 \leq \0 \mid \bar{\by}_1=\by_1, \bbeta)$ corresponds to the cumulative distribution function, evaluated at $\0$, of the {${\normalfont \SUT}_{n_0,q}(\bxi_{\by_0 \mid \cdot}, \alpha_{\beeta_{-\by_0}}\bOm_{\by_0 \mid \cdot}, \bDel_{\by_0 \mid \cdot}, \alpha^{-1/2}_{\beeta_{-\by_0}}\btau_{\by_0 \mid \cdot}, \bar\bGam_{\by_0 \mid \cdot}, \nu+p+n_1)$} having parameters as in Proposition \ref{Prop3}. In these expressions  $\alpha_{\bbeta}=[\nu+Q_{\bbeta}(\bbeta)]/(\nu+p)$ and $\alpha_{\beeta_{-\by_0}}=[\nu+Q_{\beeta_{-\by_0}}(\beeta_{-\by_0})]/(\nu+n_1+p)$, with $Q_{\bbeta}(\bbeta)$ and $Q_{\beeta_{-\by_0}}(\beeta_{-\by_0})$ defined again in Proposition \ref{Prop3}.  Finally,  the resulting posterior distribution for $\bbeta$ is {$(\bbeta \mid \by) \sim {\normalfont \SUT}_{p,n_0+q} (\bxi_{\bbeta \mid \by}, \alpha_{\by_1} \bOm_{\bbeta \mid \by}, \bDel_{\bbeta \mid \by}, \alpha_{\by_1}^{-1/2} \btau_{\bbeta \mid \by}, {\bar \bGam}_{\bbeta \mid \by}, \nu + n_1)$} with $\alpha_{\by_1}=[\nu + Q_{\by_1}(\by_1)]/(\nu+n_1)$, and the remaining quantities defined as in  Proposition \ref{Prop3}.
  \label{cor3}
\end{corollary}

\begin{proof}
The proof of Corollary~\ref{cor3} requires replacing the generic density generators in Proposition \ref{Prop3} with those of the Student's $t$, and then leveraging the properties of such generators described in  Section~\ref{sec:2.22}.
\end{proof}

Let us conclude by highlighting some special cases of Corollary~\ref{cor3} that yield priors and likelihoods of potential interest in practice. To this end, similarly to Examples~\ref{exsut} and \ref{exsut1}, consider again $ (\bbeta^\top, \beps^\top)^\top \sim {\normalfont \SUT}_{p+n,q} (\bxi, \bOm, \bDel, \btau, \bar\bGam, \nu),$ with $n=n_1+n_0$ and parameters partitioned as
\begin{eqnarray*}
    \bxi = \begin{bmatrix}
        \bxi_\bbeta\\ 
        \0
        \end{bmatrix}=\begin{bmatrix}
        \bxi_\bbeta\\ 
        \0\\
        \0
        \end{bmatrix},\qquad
    \bOm = \begin{bmatrix}
        \bOm_\bbeta & \0 \\ 
        \0 & \bOm_\beps 
        \end{bmatrix}=\begin{bmatrix}
        \bOm_\bbeta & \0 &\0 \\ 
        \0 & \bOm_{\beps_1} &\0\\
        \0 & \0 &\bOm_{\beps_0}
        \end{bmatrix},\qquad
    \bDel = \begin{bmatrix}
        \0\\
       \bDel_\beps
        \end{bmatrix}=\begin{bmatrix}
        \0\\
       \bDel_{\beps_1}\\
        \bDel_{\beps_0}
        \end{bmatrix},
\end{eqnarray*}
and $\btau=\0$. Recalling Examples~\ref{exsut} and \ref{exsut1}, such a construction implies
\begin{eqnarray*}
\bbeta \sim \mathcal{T}_{p}(\bxi_\bbeta, \bOm_\bbeta, \nu), \qquad \beps \sim {\normalfont \SUT}_{n,q} (\0,\bOm_{\beps}, \bDel_{\beps}, \0, {\bar \bGam}, \nu),
\end{eqnarray*}
and therefore, under \eqref{tobreg}, the model underlying a generic observation $\by=(\bar{\by}^\top_1,\bar{\by}^\top_0)^\top$ coincides with a multivariate extension of tobit regression having unified skew-$t$ error terms, and Student's $t$ prior for $\bbeta$ uncorrelated with the noise vector $\beps$. Applying Corollary~\ref{cor3} to such a formulation yields the likelihood
\begin{eqnarray*}
\mathcal{p}(\by \mid \bbeta)=\mathcal{p}(\by_1, \by_0 \mid \bbeta)=\mathcal{p}(\bar{\by}_1=\by_1 \mid \bbeta) \cdot \Pr(\bar{\by}_0 \leq \0 \mid \bar{\by}_1=\by_1, \bbeta),
 \end{eqnarray*}
 where $\mathcal{p}(\bar{\by}_1=\by_1 \mid \bbeta)$ coincides with the density function, computed at $\by_1$, of the ${\normalfont \SUT}_{n_1,q}(\bX_1 \bbeta, \alpha_{\bbeta}\bOm_{\beps_1}, \bDel_{\beps_1},\0, \bar\bGam, \nu+p)$ variable,  whereas the quantity $\Pr(\bar{\by}_0 \leq \0 \mid \bar{\by}_1=\by_1, \bbeta)$ corresponds to the cumulative distribution function, evaluated at $\0$, of the  {${\normalfont \SUT}_{n_0,q}(\bX_0 \bbeta, \alpha_{\beeta_{-\by_0}}\bOm_{\beps_0}, \bDel_{\beps_0} \bgam^{-1}_0, \alpha^{-1/2}_{\beeta_{-\by_0}} \bgam^{-1}_0\bDel_{\beps_1}^\top\bar{\bOm}_{\beps_1}^{-1} \bom_{\beps_1}^{-1}(\by_1-\bX_1\bbeta), \bgam_0^{-1}(\bar\bGam- \bDel_{\beps_1}^\top\bar{\bOm}_{\beps_1}^{-1} \bDel_{\beps_1})\bgam_0^{-1}, \nu+p+n_1)$}, with $\bgam_0$ defined as $\bgam_0=\mbox{diag}(\bar\bGam- \bDel_{\beps_1}^\top\bar{\bOm}_{\beps_1}^{-1} \bDel_{\beps_1})^{1/2}$. This result clarifies that classical multivariate tobit representations admit extensions to suitable skew-$t$ formulations while preserving conjugacy. Imposing additional constraints within such a formulation further highlights the practical potential of our contribution. For example, setting $ \bDel_{\beps}=\0$ in the above formulation, and recalling again Examples~\ref{exsut} and \ref{exsut1}, yields 
 \begin{eqnarray*}
 \bbeta \sim \mathcal{T}_{p}(\bxi_\bbeta, \bOm_\bbeta, \nu), \qquad  \mathcal{p}(\by \mid \bbeta)=\mathcal{p}(\bar{\by}_1=\by_1 \mid \bbeta) \cdot \Pr(\bar{\by}_0 \leq \0 \mid \bar{\by}_1=\by_1, \bbeta),
 \end{eqnarray*}
 with $\mathcal{p}(\bar{\by}_1=\by_1 \mid \bbeta)$ denoting the density of  $\mathcal{T}_{n_1}(\bX_1 \bbeta, \alpha_{\bbeta}\bOm_{\beps_1}, \nu+p)$  evaluated at $\by_1$, whereas  $\Pr(\bar{\by}_0 \leq \0 \mid \bar{\by}_1=\by_1, \bbeta)$ is a {$\mathcal{T}_{n_0}(\bX_0 \bbeta, \alpha_{\beeta_{-\by_0}}\bOm_{\beps_0}, \nu+p+n_1)$} cumulative distribution function computed at $\0$. As a consequence of Corollary~\ref{cor3}, the induced posterior distribution for $\bbeta$ is still within the SUT family.
 \end{exa}
 

\section{Conclusions\label{sec:4}}

This article proves that SUE distributions have important conjugacy properties when combined with broad classes of likelihoods that generalize classical probit, tobit, multinomial probit, and linear models in several directions. These generalizations include multivariate representations based on general elliptical noise terms and allow for asymmetric representation relying on unified skew-elliptical extensions. Our results leverage available and newly-derived closure properties of the SUE family to prove that priors within this class yield again SUE posterior distributions when combined with the likelihood of the models mentioned above, under the classical Bayes rule. Recalling Propositions~\ref{Prop1}--\ref{Prop3}, these results are technically derived by starting from a joint SUE distribution for the parameters and the observed data. Such a proof technique is not meant to provide a different perspective on the standard specification of a prior and a likelihood in Bayesian statistics. Rather, it provides a convenient strategy that facilitates the derivation, within the SUE class, of meaningful priors and likelihoods yielding closed-form SUE posterior distributions. 

More specifically, Examples \ref{exsun}--\ref{exsut3} clarify that our results include Bayesian models of direct interest in practice, such as those based on multivariate Gaussian or Student's $t$ formulations, along with the corresponding skewed extensions. In this respect, an interesting  direction would be to specialize  Propositions~\ref{Prop1}--\ref{Prop3} to other  SUE sub-families, e.g., those based on Cauchy or logistic density generators. This  can be  accomplished by replacing the generic density generators in  Propositions~\ref{Prop1}--\ref{Prop3}, with those yielding the sub-family investigated. These extensions further motivate advancements in the study of other relevant  SUE sub-families to derive results and properties similar to those characterizing SUN \citep[e.g.,][]{arellano22} and SUT \citep[e.g.,][]{wang24} distributions. Particularly impactful, within our context, would be the derivation of additive stochastic representations as those obtained for  SUNs and  SUTs. Advancements along these lines would facilitate i.i.d.\ sampling under any SUE posterior, thus enlarging the class of models and priors that allow for tractable Bayesian inference. The recent additive stochastic representations derived by \citet{yin2024stochastic} for general skew-elliptical distributions provide a promising advancement in this direction, which also suggests that related results could be derived even for the wider SUE family. Similarly, expanding the available strategies for the efficient evaluation of the moments of SUE distributions in, e.g., \eqref{meanSUE}--\eqref{varSUE}, would further facilitate Bayesian inference leveraging the conjugacy results derived in the present article. Current contributions \citep[e.g.,][]{arismendi2017multivariate,moran2021new,galarza22,morales2022moments,valeriano2023moments} provide important results along these lines which motivate future research to showcase the computational advantages and the practical impact of these solutions when the focus is on Bayesian inference under the newly-derived SUE posterior distributions.

Finally, we shall emphasize that the SUE family can be itself rephrased as a particular case of selection elliptical distributions \citep{arellano06b} arising from even  more general conditioning mechanisms. As such, it would be interesting to expand the conjugacy properties derived in this article for SUE distributions to the broader selection elliptical  family. Let us conclude by highlighting that not all the priors and likelihoods implied by the general results in Propositions~\ref{Prop1}--\ref{Prop3} have direct practical applicability. Nonetheless, from a theoretical perspective, also these instances are of interest in expanding the analysis of the probabilistic properties of the SUE family. Moreover, although conjugacy is a desirable property, it is important to emphasize that those priors and likelihoods in the SUE family that do not yield SUE posteriors, can still allow for Bayesian inference leveraging, e.g., MCMC methods and deterministic approximations.
 

\section*{Acknowledgments}
Maicon J. Karling and Marc G. Genton  acknowledge the support by the King Abdullah University of Science and Technology (KAUST) in Saudi Arabia.



\begin{thebibliography}{65}
\expandafter\ifx\csname natexlab\endcsname\relax\def\natexlab#1{#1}\fi
\providecommand{\bibinfo}[2]{#2}
\ifx\xfnm\relax \def\xfnm[#1]{\unskip,\space#1}\fi

\bibitem[{Adcock and Azzalini(2020)}]{adcock20}
\bibinfo{author}{C.~Adcock}, \bibinfo{author}{A.~Azzalini}, \bibinfo{title}{A
  selective overview of skew-elliptical and related distributions and of their
  applications}, \bibinfo{journal}{Symmetry} \bibinfo{volume}{12}
  (\bibinfo{year}{2020}) \bibinfo{pages}{118}.

\bibitem[{Albert and Chib(1993)}]{albert1993bayesian}
\bibinfo{author}{J.~H. Albert}, \bibinfo{author}{S.~Chib},
  \bibinfo{title}{Bayesian analysis of binary and polychotomous response data},
  \bibinfo{journal}{Journal of the American Statistical Association}
  \bibinfo{volume}{88} (\bibinfo{year}{1993}) \bibinfo{pages}{669--679}.

\bibitem[{Amemiya(1984)}]{amemiya84}
\bibinfo{author}{T.~Amemiya}, \bibinfo{title}{Tobit models: A survey},
  \bibinfo{journal}{Journal of Econometrics} \bibinfo{volume}{24}
  (\bibinfo{year}{1984}) \bibinfo{pages}{3--61}.

\bibitem[{Anceschi et~al.(2023)Anceschi, Fasano, Durante and
  Zanella}]{anceschi23}
\bibinfo{author}{N.~Anceschi}, \bibinfo{author}{A.~Fasano},
  \bibinfo{author}{D.~Durante}, \bibinfo{author}{G.~Zanella},
  \bibinfo{title}{Bayesian conjugacy in probit, tobit, multinomial probit and
  extensions: {A} review and new results}, \bibinfo{journal}{Journal of the
  American Statistical Association} \bibinfo{volume}{118}
  (\bibinfo{year}{2023}) \bibinfo{pages}{1451--1469}.

\bibitem[{Arellano-Valle and Azzalini(2006)}]{arellano06a}
\bibinfo{author}{R.~B. Arellano-Valle}, \bibinfo{author}{A.~Azzalini},
  \bibinfo{title}{On the unification of families of skew-normal distributions},
  \bibinfo{journal}{Scandinavian Journal of Statistics} \bibinfo{volume}{33}
  (\bibinfo{year}{2006}) \bibinfo{pages}{561--574}.

\bibitem[{Arellano-Valle and Azzalini(2022)}]{arellano22}
\bibinfo{author}{R.~B. Arellano-Valle}, \bibinfo{author}{A.~Azzalini},
  \bibinfo{title}{Some properties of the unified skew-normal distribution},
  \bibinfo{journal}{Statistical Papers} \bibinfo{volume}{63}
  (\bibinfo{year}{2022}) \bibinfo{pages}{461--487}.

\bibitem[{Arellano-Valle et~al.(2006)Arellano-Valle, Branco and
  Genton}]{arellano06b}
\bibinfo{author}{R.~B. Arellano-Valle}, \bibinfo{author}{M.~D. Branco},
  \bibinfo{author}{M.~G. Genton}, \bibinfo{title}{A unified view on skewed
  distributions arising from selections}, \bibinfo{journal}{The Canadian
  Journal of Statistics} \bibinfo{volume}{34} (\bibinfo{year}{2006})
  \bibinfo{pages}{581--601}.

\bibitem[{Arellano-Valle et~al.(2012)Arellano-Valle, Castro,
  Gonz{\'a}lez-Far{\'\i}as and Mu{\~n}oz-Gajardo}]{arellano2012student}
\bibinfo{author}{R.~B. Arellano-Valle}, \bibinfo{author}{L.~M. Castro},
  \bibinfo{author}{G.~Gonz{\'a}lez-Far{\'\i}as}, \bibinfo{author}{K.~A.
  Mu{\~n}oz-Gajardo}, \bibinfo{title}{Student-t censored regression model:
  Properties and inference}, \bibinfo{journal}{Statistical Methods \&
  Applications} \bibinfo{volume}{21} (\bibinfo{year}{2012})
  \bibinfo{pages}{453--473}.

\bibitem[{Arellano-Valle and Genton(2010{\natexlab{a}})}]{arellano10a}
\bibinfo{author}{R.~B. Arellano-Valle}, \bibinfo{author}{M.~G. Genton},
  \bibinfo{title}{Multivariate extended skew--t distributions and related
  families}, \bibinfo{journal}{Metron} \bibinfo{volume}{68}
  (\bibinfo{year}{2010}{\natexlab{a}}) \bibinfo{pages}{201--234}.

\bibitem[{Arellano-Valle and Genton(2010{\natexlab{b}})}]{arellano10b}
\bibinfo{author}{R.~B. Arellano-Valle}, \bibinfo{author}{M.~G. Genton},
  \bibinfo{title}{Multivariate unified skew-elliptical distributions},
  \bibinfo{journal}{Chilean Journal of Statistics} \bibinfo{volume}{1}
  (\bibinfo{year}{2010}{\natexlab{b}}) \bibinfo{pages}{17--33}.

\bibitem[{Arellano-Valle et~al.(2009)Arellano-Valle, Genton and
  Loschi}]{arellano2009shape}
\bibinfo{author}{R.~B. Arellano-Valle}, \bibinfo{author}{M.~G. Genton},
  \bibinfo{author}{R.~H. Loschi}, \bibinfo{title}{Shape mixtures of
  multivariate skew-normal distributions}, \bibinfo{journal}{Journal of
  Multivariate Analysis} \bibinfo{volume}{100} (\bibinfo{year}{2009})
  \bibinfo{pages}{91--101}.

\bibitem[{Arellano-Valle et~al.(2005)Arellano-Valle, Ozan, Bolfarine and
  Lachos}]{arellano2005skew}
\bibinfo{author}{R.~B. Arellano-Valle}, \bibinfo{author}{S.~Ozan},
  \bibinfo{author}{H.~Bolfarine}, \bibinfo{author}{V.~H.~Lachos},
  \bibinfo{title}{Skew normal measurement error models},
  \bibinfo{journal}{Journal of Multivariate Analysis} \bibinfo{volume}{96}
  (\bibinfo{year}{2005}) \bibinfo{pages}{265--281}.

\bibitem[{Arismendi and Broda(2017)}]{arismendi2017multivariate}
\bibinfo{author}{J.~C. Arismendi}, \bibinfo{author}{S.~Broda},
  \bibinfo{title}{Multivariate elliptical truncated moments},
  \bibinfo{journal}{Journal of Multivariate Analysis} \bibinfo{volume}{157}
  (\bibinfo{year}{2017}) \bibinfo{pages}{29--44}.

\bibitem[{Arnold and Beaver(2000)}]{arnold2000hidden}
\bibinfo{author}{B.~C. Arnold}, \bibinfo{author}{R.~J. Beaver},
  \bibinfo{title}{Hidden truncation models}, \bibinfo{journal}{Sankhy{\=a}: The
  Indian Journal of Statistics, Series A} \bibinfo{volume}{62} (\bibinfo{year}{2000})
  \bibinfo{pages}{23--35}.

\bibitem[{Arnold and Beaver(2002)}]{arnold2002skewed}
\bibinfo{author}{B.~C. Arnold}, \bibinfo{author}{R.~J. Beaver},
  \bibinfo{title}{Skewed multivariate models related to hidden truncation
  and/or selective reporting}, \bibinfo{journal}{Test} \bibinfo{volume}{11}
  (\bibinfo{year}{2002}) \bibinfo{pages}{7--54}.

\bibitem[{Azzalini and Capitanio(1999)}]{azzalini1999statistical}
\bibinfo{author}{A.~Azzalini}, \bibinfo{author}{A.~Capitanio},
  \bibinfo{title}{Statistical applications of the multivariate skew normal
  distribution}, \bibinfo{journal}{Journal of the Royal Statistical Society:
  Series B (Statistical Methodology)} \bibinfo{volume}{61}
  (\bibinfo{year}{1999}) \bibinfo{pages}{579--602}.

\bibitem[{Azzalini and Capitanio(2003)}]{azzalini2003distributions}
\bibinfo{author}{A.~Azzalini}, \bibinfo{author}{A.~Capitanio},
  \bibinfo{title}{Distributions generated by perturbation of symmetry with
  emphasis on a multivariate skew t-distribution}, \bibinfo{journal}{Journal of
  the Royal Statistical Society: Series B (Statistical Methodology)}
  \bibinfo{volume}{65} (\bibinfo{year}{2003}) \bibinfo{pages}{367--389}.

\bibitem[{Azzalini and Capitanio(2014)}]{azzalini14}
\bibinfo{author}{A.~Azzalini}, \bibinfo{author}{A.~Capitanio},
  \bibinfo{title}{{The Skew-Normal and Related Families}},
  \bibinfo{publisher}{Cambridge University Press},
  \bibinfo{address}{Cambridge}, \bibinfo{year}{2014}.

\bibitem[{Azzalini and Dalla~Valle(1996)}]{azzalini1996multivariate}
\bibinfo{author}{A.~Azzalini}, \bibinfo{author}{A.~Dalla~Valle},
  \bibinfo{title}{The multivariate skew-normal distribution},
  \bibinfo{journal}{Biometrika} \bibinfo{volume}{83} (\bibinfo{year}{1996})
  \bibinfo{pages}{715--726}.

\bibitem[{Barros et~al.(2018)Barros, Galea, Leiva and
  Santos-Neto}]{barros2018generalized}
\bibinfo{author}{M.~Barros}, \bibinfo{author}{M.~Galea},
  \bibinfo{author}{V.~Leiva}, \bibinfo{author}{M.~Santos-Neto},
  \bibinfo{title}{Generalized tobit models: Diagnostics and application in
  econometrics}, \bibinfo{journal}{Journal of Applied Statistics}
  \bibinfo{volume}{45} (\bibinfo{year}{2018}) \bibinfo{pages}{145--167}.

\bibitem[{Baz{\'{a}}n et~al.(2010)Baz{\'{a}}n, Bolfarine and Branco}]{bazan10}
\bibinfo{author}{J.~L. Baz{\'{a}}n}, \bibinfo{author}{H.~Bolfarine},
  \bibinfo{author}{M.~D. Branco}, \bibinfo{title}{A framework for skew-probit
  links in binary regression}, \bibinfo{journal}{Communications in Statistics -
  Theory and Methods} \bibinfo{volume}{39} (\bibinfo{year}{2010})
  \bibinfo{pages}{678--697}.

\bibitem[{Benavoli et~al.(2020)Benavoli, Azzimonti and Piga}]{benavoli2020skew}
\bibinfo{author}{A.~Benavoli}, \bibinfo{author}{D.~Azzimonti},
  \bibinfo{author}{D.~Piga}, \bibinfo{title}{Skew {G}aussian processes for
  classification}, \bibinfo{journal}{Machine Learning} \bibinfo{volume}{109}
  (\bibinfo{year}{2020}) \bibinfo{pages}{1877--1902}.

\bibitem[{Bolfarine and Lachos(2007)}]{bolfarine2007skew}
\bibinfo{author}{H.~Bolfarine}, \bibinfo{author}{V.~H. Lachos},
  \bibinfo{title}{Skew-probit measurement error models},
  \bibinfo{journal}{Statistical Methodology} \bibinfo{volume}{4}
  (\bibinfo{year}{2007}) \bibinfo{pages}{1--12}.

\bibitem[{Branco and Dey(2001)}]{branco01}
\bibinfo{author}{M.~D. Branco}, \bibinfo{author}{D.~K. Dey}, \bibinfo{title}{A
  general class of multivariate skew-elliptical distributions},
  \bibinfo{journal}{Journal of Multivariate Analysis} \bibinfo{volume}{79}
  (\bibinfo{year}{2001}) \bibinfo{pages}{99--113}.

\bibitem[{Branco and Dey(2002)}]{branco02}
\bibinfo{author}{M.~D. Branco}, \bibinfo{author}{D.~K. Dey},
  \bibinfo{title}{Regression model under skew elliptical error distribution},
  \bibinfo{journal}{The Journal of Mathematical Sciences} \bibinfo{volume}{1}
  (\bibinfo{year}{2002}) \bibinfo{pages}{151--169}.

\bibitem[{Canale et~al.(2016)Canale, Kenne~Pagui and
  Scarpa}]{canale2016bayesian}
\bibinfo{author}{A.~Canale}, \bibinfo{author}{E.~C. Kenne~Pagui},
  \bibinfo{author}{B.~Scarpa}, \bibinfo{title}{Bayesian modeling of university
  first-year students' grades after placement test}, \bibinfo{journal}{Journal
  of Applied Statistics} \bibinfo{volume}{43} (\bibinfo{year}{2016})
  \bibinfo{pages}{3015--3029}.

\bibitem[{Cao et~al.(2022)Cao, Durante and Genton}]{cao2022scalable}
\bibinfo{author}{J.~Cao}, \bibinfo{author}{D.~Durante}, \bibinfo{author}{M.~G.
  Genton}, \bibinfo{title}{Scalable computation of predictive probabilities in
  probit models with {G}aussian process priors}, \bibinfo{journal}{Journal of
  Computational and Graphical Statistics} \bibinfo{volume}{31}
  (\bibinfo{year}{2022}) \bibinfo{pages}{709--720}.

\bibitem[{Chen et~al.(1999)Chen, Dey and Shao}]{chen1999new}
\bibinfo{author}{M.-H. Chen}, \bibinfo{author}{D.~K. Dey},
  \bibinfo{author}{Q.-M. Shao}, \bibinfo{title}{A new skewed link model for
  dichotomous quantal response data}, \bibinfo{journal}{Journal of the American
  Statistical Association} \bibinfo{volume}{94} (\bibinfo{year}{1999})
  \bibinfo{pages}{1172--1186}.

\bibitem[{Chib(1992)}]{chib1992bayes}
\bibinfo{author}{S.~Chib}, \bibinfo{title}{Bayes inference in the tobit
  censored regression model}, \bibinfo{journal}{Journal of Econometrics}
  \bibinfo{volume}{51} (\bibinfo{year}{1992}) \bibinfo{pages}{79--99}.

\bibitem[{Chib and Greenberg(1998)}]{chib1998analysis}
\bibinfo{author}{S.~Chib}, \bibinfo{author}{E.~Greenberg},
  \bibinfo{title}{Analysis of multivariate probit models},
  \bibinfo{journal}{Biometrika} \bibinfo{volume}{85} (\bibinfo{year}{1998})
  \bibinfo{pages}{347--361}.

\bibitem[{Dagne and Huang(2012)}]{dagne12}
\bibinfo{author}{G.~Dagne}, \bibinfo{author}{Y.~Huang},
  \bibinfo{title}{Bayesian inference for a nonlinear mixed-effects tobit
  model with multivariate skew-t distributions: Application to {AIDS} studies},
  \bibinfo{journal}{The International Journal of Biostatistics}
  \bibinfo{volume}{8} (\bibinfo{year}{2012}) \bibinfo{pages}{27}.

\bibitem[{Durante(2019)}]{durante19}
\bibinfo{author}{D.~Durante}, \bibinfo{title}{Conjugate {B}ayes for probit
  regression via unified skew-normal distributions},
  \bibinfo{journal}{Biometrika} \bibinfo{volume}{106} (\bibinfo{year}{2019})
  \bibinfo{pages}{765--779}.

\bibitem[{Durante et~al.(2024)Durante, Pozza and Szabo}]{durante2023skewed}
\bibinfo{author}{D.~Durante}, \bibinfo{author}{F.~Pozza},
  \bibinfo{author}{B.~Szabo}, \bibinfo{title}{Skewed {B}ernstein-von {M}ises
  theorem and skew-modal approximations}, \bibinfo{journal}{Annals of
  Statistics}  (\bibinfo{year}{2024}) \bibinfo{pages}{In press}.

\bibitem[{Fang(2003)}]{fang03}
\bibinfo{author}{B.~Q.~Fang}, \bibinfo{title}{The skew elliptical distributions
  and their quadratic forms}, \bibinfo{journal}{Journal of Multivariate
  Analysis} \bibinfo{volume}{87} (\bibinfo{year}{2003})
  \bibinfo{pages}{298--314}.

\bibitem[{Fang et~al.(1990)Fang, Kotz and Ng}]{fang90}
\bibinfo{author}{K.-T. Fang}, \bibinfo{author}{S.~Kotz}, \bibinfo{author}{K.-W.
  Ng}, \bibinfo{title}{{Symmetric Multivariate and Related Distributions}},
  \bibinfo{publisher}{Chapman \& Hall}, \bibinfo{address}{London},
  \bibinfo{year}{1990}.

\bibitem[{Fasano and Durante(2022)}]{fasano2022class}
\bibinfo{author}{A.~Fasano}, \bibinfo{author}{D.~Durante}, \bibinfo{title}{A
  class of conjugate priors for multinomial probit models which includes the
  multivariate normal one}, \bibinfo{journal}{The Journal of Machine Learning
  Research} \bibinfo{volume}{23} (\bibinfo{year}{2022})
  \bibinfo{pages}{1358--1383}.

\bibitem[{Fasano et~al.(2022)Fasano, Durante and Zanella}]{fasano2022scalable}
\bibinfo{author}{A.~Fasano}, \bibinfo{author}{D.~Durante},
  \bibinfo{author}{G.~Zanella}, \bibinfo{title}{Scalable and accurate
  variational {B}ayes for high-dimensional binary regression models},
  \bibinfo{journal}{Biometrika} \bibinfo{volume}{109} (\bibinfo{year}{2022})
  \bibinfo{pages}{901--919}.

\bibitem[{Fasano et~al.(2021)Fasano, Rebaudo, Durante and
  Petrone}]{fasano2021closed}
\bibinfo{author}{A.~Fasano}, \bibinfo{author}{G.~Rebaudo},
  \bibinfo{author}{D.~Durante}, \bibinfo{author}{S.~Petrone}, \bibinfo{title}{A
  closed-form filter for binary time series}, \bibinfo{journal}{Statistics and
  Computing} \bibinfo{volume}{31} (\bibinfo{year}{2021}) \bibinfo{pages}{47}.

\bibitem[{Galarza et~al.(2022{\natexlab{a}})Galarza, Matos, Castro and
  Lachos}]{galarza22}
\bibinfo{author}{C.~E. Galarza}, \bibinfo{author}{L.~A. Matos},
  \bibinfo{author}{L.~M. Castro}, \bibinfo{author}{V.~H. Lachos},
  \bibinfo{title}{Moments of the doubly truncated selection elliptical
  distributions with emphasis on the unified multivariate skew-t distribution},
  \bibinfo{journal}{Journal of Multivariate Analysis} \bibinfo{volume}{189}
  (\bibinfo{year}{2022}{\natexlab{a}}) \bibinfo{pages}{104944}.

\bibitem[{Galarza et~al.(2022{\natexlab{b}})Galarza, Matos and
  Lachos}]{galarza2022algorithm}
\bibinfo{author}{C.~E. Galarza}, \bibinfo{author}{L.~A. Matos},
  \bibinfo{author}{V.~H. Lachos}, \bibinfo{title}{An {EM} algorithm for
  estimating the parameters of the multivariate skew-normal distribution with
  censored responses}, \bibinfo{journal}{Metron} \bibinfo{volume}{80}
  (\bibinfo{year}{2022}{\natexlab{b}}) \bibinfo{pages}{231--253}.

\bibitem[{Garay et~al.(2015)Garay, Bolfarine, Lachos and
  Cabral}]{garay2015bayesian}
\bibinfo{author}{A.~M. Garay}, \bibinfo{author}{H.~Bolfarine},
  \bibinfo{author}{V.~H. Lachos}, \bibinfo{author}{C.~R.~B. Cabral},
  \bibinfo{title}{Bayesian analysis of censored linear regression models with
  scale mixtures of normal distributions}, \bibinfo{journal}{Journal of Applied
  Statistics} \bibinfo{volume}{42} (\bibinfo{year}{2015})
  \bibinfo{pages}{2694--2714}.

\bibitem[{Gonzalez-Farias et~al.(2004)Gonzalez-Farias, Dominguez-Molina and
  Gupta}]{gonzalez2004additive}
\bibinfo{author}{G.~Gonz\'alez-Far\'ias}, \bibinfo{author}{J.~A.~Dom\'inguez-Molina},
  \bibinfo{author}{A.~K. Gupta}, \bibinfo{title}{Additive properties of skew
  normal random vectors}, \bibinfo{journal}{Journal of Statistical Planning and
  Inference} \bibinfo{volume}{126} (\bibinfo{year}{2004})
  \bibinfo{pages}{521--534}.

\bibitem[{Gupta et~al.(2004)Gupta, Gonzalez-Farias and
  Dominguez-Molina}]{gupta2004multivariate}
\bibinfo{author}{A.~K. Gupta}, \bibinfo{author}{G.~Gonz\'alez-Far\'ias},
  \bibinfo{author}{J.~A.~Dom\'inguez-Molina}, \bibinfo{title}{A multivariate skew
  normal distribution}, \bibinfo{journal}{Journal of Multivariate Analysis}
  \bibinfo{volume}{89} (\bibinfo{year}{2004}) \bibinfo{pages}{181--190}.

\bibitem[{Hutton and Stanghellini(2011)}]{hutton11}
\bibinfo{author}{J.~L.~Hutton}, \bibinfo{author}{E.~Stanghellini},
  \bibinfo{title}{Modelling bounded health scores with censored skew-normal
  distributions}, \bibinfo{journal}{Statistics in Medicine} \bibinfo{volume}{30}
  (\bibinfo{year}{2011}) \bibinfo{pages}{368--376}.

\bibitem[{Islam et~al.(2014)Islam, Yildirim and Yazici}]{islam14}
\bibinfo{author}{M.~Q. Islam}, \bibinfo{author}{F.~Yildirim},
  \bibinfo{author}{M.~Yazici}, \bibinfo{title}{{Inference in multivariate
  linear regression models with elliptically distributed errors}},
  \bibinfo{journal}{Journal of Applied Statistics} \bibinfo{volume}{41}
  (\bibinfo{year}{2014}) \bibinfo{pages}{1746--1766}.

\bibitem[{Lachos et~al.(2022)Lachos, Baz\'an, Castro and Park}]{lachos22}
\bibinfo{author}{V.~H. Lachos}, \bibinfo{author}{J.~L. Baz\'an},
  \bibinfo{author}{L.~M. Castro}, \bibinfo{author}{J.~Park},
  \bibinfo{title}{The skew-t censored regression model: Parameter estimation
  via an {EM}-type algorithm}, \bibinfo{journal}{Communications for Statistical
  Applications and Methods} \bibinfo{volume}{29} (\bibinfo{year}{2022})
  \bibinfo{pages}{333--351}.

\bibitem[{Lachos et~al.(2007)Lachos, Bolfarine, Arellano-Valle and
  Montenegro}]{lachos07}
\bibinfo{author}{V.~H. Lachos}, \bibinfo{author}{H.~Bolfarine},
  \bibinfo{author}{R.~B. Arellano-Valle}, \bibinfo{author}{L.~C. Montenegro},
  \bibinfo{title}{Likelihood-based inference for multivariate skew-normal
  regression models}, \bibinfo{journal}{Communications in Statistics - Theory
  and Methods} \bibinfo{volume}{36} (\bibinfo{year}{2007})
  \bibinfo{pages}{1769--1786}.

\bibitem[{Lachos et~al.(2010)Lachos, Ghosh and
  Arellano-Valle}]{lachos2010likelihood}
\bibinfo{author}{V.~H. Lachos}, \bibinfo{author}{P.~Ghosh},
  \bibinfo{author}{R.~B. Arellano-Valle}, \bibinfo{title}{Likelihood based
  inference for skew-normal independent linear mixed models},
  \bibinfo{journal}{Statistica Sinica}  (\bibinfo{year}{2010})
  \bibinfo{pages}{303--322}.

\bibitem[{Lachos et~al.(2021)Lachos, Prates and Dey}]{lachos2021heckman}
\bibinfo{author}{V.~H. Lachos}, \bibinfo{author}{M.~O. Prates},
  \bibinfo{author}{D.~K. Dey}, \bibinfo{title}{Heckman selection-t model:
  Parameter estimation via the {EM}-algorithm}, \bibinfo{journal}{Journal of
  Multivariate Analysis} \bibinfo{volume}{184} (\bibinfo{year}{2021})
  \bibinfo{pages}{104737}.

\bibitem[{Lange et~al.(1989)Lange, Little and Taylor}]{lange1989robust}
\bibinfo{author}{K.~L. Lange}, \bibinfo{author}{R.~J.~A. Little},
  \bibinfo{author}{J.~M.~G. Taylor}, \bibinfo{title}{Robust statistical modeling
  using the t distribution}, \bibinfo{journal}{Journal of the American
  Statistical Association} \bibinfo{volume}{84} (\bibinfo{year}{1989})
  \bibinfo{pages}{881--896}.

\bibitem[{Marchenko and Genton(2012)}]{marchenko12}
\bibinfo{author}{Y.~V. Marchenko}, \bibinfo{author}{M.~G. Genton},
  \bibinfo{title}{A {H}eckman selection-t model}, \bibinfo{journal}{Journal of
  the American Statistical Association} \bibinfo{volume}{107}
  (\bibinfo{year}{2012}) \bibinfo{pages}{304--317}.

\bibitem[{Matos et~al.(2018)Matos, Castro, Cabral and
  Lachos}]{matos2018multivariate}
\bibinfo{author}{L.~A. Matos}, \bibinfo{author}{L.~M. Castro},
  \bibinfo{author}{C.~R.~B. Cabral}, \bibinfo{author}{V.~H. Lachos},
  \bibinfo{title}{Multivariate measurement error models based on Student-t
  distribution under censored responses}, \bibinfo{journal}{Statistics}
  \bibinfo{volume}{52} (\bibinfo{year}{2018}) \bibinfo{pages}{1395--1416}.

\bibitem[{Morales et~al.(2022)Morales, Matos, Dey and
  Lachos}]{morales2022moments}
\bibinfo{author}{C.~E.~G. Morales}, \bibinfo{author}{L.~A. Matos},
  \bibinfo{author}{D.~K. Dey}, \bibinfo{author}{V.~H. Lachos},
  \bibinfo{title}{On moments of folded and doubly truncated multivariate
  extended skew-normal distributions}, \bibinfo{journal}{Journal of
  Computational and Graphical Statistics} \bibinfo{volume}{31}
  (\bibinfo{year}{2022}) \bibinfo{pages}{455--465}.

\bibitem[{Mor{\'a}n-V{\'a}squez and Ferrari(2021)}]{moran2021new}
\bibinfo{author}{R.~A. Mor{\'a}n-V{\'a}squez}, \bibinfo{author}{S.~L.~P. Ferrari},
  \bibinfo{title}{New results on truncated elliptical distributions},
  \bibinfo{journal}{Communications in Mathematics and Statistics}
  \bibinfo{volume}{9} (\bibinfo{year}{2021}) \bibinfo{pages}{299--313}.

\bibitem[{Onorati and Liseo(2022)}]{onorati2022extension}
\bibinfo{author}{P.~Onorati}, \bibinfo{author}{B.~Liseo}, \bibinfo{title}{An
  extension of the unified skew-normal family of distributions and application
  to {B}ayesian binary regression}, \bibinfo{journal}{arXiv:2209.03474}
  (\bibinfo{year}{2022}).

\bibitem[{Sahu et~al.(2003)Sahu, Dey and Branco}]{sahu03}
\bibinfo{author}{S.~K. Sahu}, \bibinfo{author}{D.~K. Dey},
  \bibinfo{author}{M.~D. Branco}, \bibinfo{title}{A new class of multivariate
  skew distributions with applications to {B}ayesian regression models},
  \bibinfo{journal}{The Canadian Journal of Statistics} \bibinfo{volume}{31}
  (\bibinfo{year}{2003}) \bibinfo{pages}{129--150}.

\bibitem[{Song and Xia(2016)}]{song2016bayesian}
\bibinfo{author}{C.~Song}, \bibinfo{author}{S.-T. Xia},
  \bibinfo{title}{Bayesian linear regression with Student-t assumptions},
  \bibinfo{journal}{arXiv:1604.04434}  (\bibinfo{year}{2016}).

\bibitem[{Spanos(1994)}]{spanos1994modeling}
\bibinfo{author}{A.~Spanos}, \bibinfo{title}{On modeling heteroskedasticity:
  The Student's t and elliptical linear regression models},
  \bibinfo{journal}{Econometric Theory} \bibinfo{volume}{10}
  (\bibinfo{year}{1994}) \bibinfo{pages}{286--315}.

\bibitem[{Valeriano et~al.(2023{\natexlab{a}})Valeriano, Galarza and
  Matos}]{valeriano2023moments}
\bibinfo{author}{K.~A.~L. Valeriano}, \bibinfo{author}{C.~E. Galarza},
  \bibinfo{author}{L.~A. Matos}, \bibinfo{title}{Moments and random number
  generation for the truncated elliptical family of distributions},
  \bibinfo{journal}{Statistics and Computing} \bibinfo{volume}{33}
  (\bibinfo{year}{2023}{\natexlab{a}}) \bibinfo{pages}{32}.

\bibitem[{Valeriano et~al.(2023{\natexlab{b}})Valeriano, Galarza, Matos and
  Lachos}]{valeriano23}
\bibinfo{author}{K.~A.~L. Valeriano}, \bibinfo{author}{C.~E. Galarza},
  \bibinfo{author}{L.~A. Matos}, \bibinfo{author}{V.~H. Lachos},
  \bibinfo{title}{Likelihood-based inference for the multivariate skew-t
  regression with censored or missing responses}, \bibinfo{journal}{Journal of
  Multivariate Analysis} \bibinfo{volume}{196}
  (\bibinfo{year}{2023}{\natexlab{b}}) \bibinfo{pages}{105174}.

\bibitem[{Vieira et~al.(2015)Vieira, Loschi and
  Duarte}]{vieira2015nonparametric}
\bibinfo{author}{C.~C. Vieira}, \bibinfo{author}{R.~H. Loschi},
  \bibinfo{author}{D.~Duarte}, \bibinfo{title}{Nonparametric mixtures based on
  skew-normal distributions: An application to density estimation},
  \bibinfo{journal}{Communications in Statistics--Theory and Methods}
  \bibinfo{volume}{44} (\bibinfo{year}{2015}) \bibinfo{pages}{1552--1570}.

\bibitem[{Wang et~al.(2024)Wang, Karling, Arellano‐Valle and Genton}]{wang24}
\bibinfo{author}{K.~Wang}, \bibinfo{author}{M.~J. Karling},
  \bibinfo{author}{R.~B. Arellano‐Valle}, \bibinfo{author}{M.~G. Genton},
  \bibinfo{title}{Multivariate unified skew-t distributions and their
  properties}, \bibinfo{journal}{Journal of Multivariate Analysis}
  \bibinfo{volume}{203} (\bibinfo{year}{2024}) \bibinfo{pages}{105322}.

\bibitem[{Yin and Balakrishnan(2024)}]{yin2024stochastic}
\bibinfo{author}{C.~Yin}, \bibinfo{author}{N.~Balakrishnan},
  \bibinfo{title}{Stochastic representations and probabilistic characteristics
  of multivariate skew-elliptical distributions}, \bibinfo{journal}{Journal of
  Multivariate Analysis} \bibinfo{volume}{199} (\bibinfo{year}{2024})
  \bibinfo{pages}{105240}.

\bibitem[{Zellner(1976)}]{zellner1976bayesian}
\bibinfo{author}{A.~Zellner}, \bibinfo{title}{Bayesian and non-{B}ayesian
  analysis of the regression model with multivariate Student-t error terms},
  \bibinfo{journal}{Journal of the American Statistical Association}
  \bibinfo{volume}{71} (\bibinfo{year}{1976}) \bibinfo{pages}{400--405}.

\bibitem[{Zhang et~al.(2023)Zhang, Arellano-Valle, Genton and Huser}]{zhang23}
\bibinfo{author}{Z.~Zhang}, \bibinfo{author}{R.~B. Arellano-Valle},
  \bibinfo{author}{M.~G. Genton}, \bibinfo{author}{R.~Huser},
  \bibinfo{title}{Tractable {B}ayes of skew-elliptical link models for
  correlated binary data}, \bibinfo{journal}{Biometrics} \bibinfo{volume}{79}
  (\bibinfo{year}{2023}) \bibinfo{pages}{1788--1800}.

\end{thebibliography}
\end{document}